\documentclass[reqno]{amsart}
\usepackage{amsmath,mathtools}
\usepackage{graphicx}
\usepackage[dvipsnames]{xcolor}
\usepackage{chngcntr}
\usepackage[normalem]{ulem}
\usepackage{amsmath,amsfonts, amssymb,amsfonts,amsthm}
\numberwithin{equation}{section}

\usepackage{wrapfig}
\usepackage{chngcntr}
\usepackage{apptools}
\usepackage{paralist}
\usepackage{graphics} %% add this and next lines if pictures should be in esp format
  \usepackage{epsfig} %For pictures: screened artwork should be set up with an 85 or 100 line screen
 \usepackage[colorlinks=true]{hyperref}
 \usepackage{graphicx}  \usepackage{epstopdf}
 \usepackage{color}
 
 \usepackage[toc,page]{appendix}

\usepackage{tikz}
\usepackage{pgfplots}
\usepackage{bm}
\usepgfplotslibrary{fillbetween}
\usetikzlibrary{backgrounds}
\usetikzlibrary{patterns}
\usetikzlibrary{scopes}
\usetikzlibrary{shapes,arrows}
\usetikzlibrary{plotmarks,intersections}
\usetikzlibrary{calc}
\usepackage{bbm}
\hypersetup{urlcolor=blue, citecolor=red}

\newtheorem{theorem}{Theorem}[section]

\newtheorem{lemma}[theorem]{Lemma}
\newtheorem{proposition}{Proposition}

\theoremstyle{definition}
\newtheorem{definition}[theorem]{Definition}
\newtheorem{remark}{Remark}

\title[]
      {On existence and uniqueness to homogeneous Boltzmann flows of monatomic gas mixtures }

\author[Irene M. Gamba and Milana Pavi\'c-\v Coli\'{c}]{}

\newcommand{\changelocaltocdepth}[1]{%
	\addtocontents{toc}{\protect\setcounter{tocdepth}{#1}}%
	\setcounter{tocdepth}{#1}%
}

\newcommand{\nocontentsline}[3]{}
\newcommand{\tocless}[2]{\bgroup\let\addcontentsline=\nocontentsline#1{#2}\egroup}

\newcommand{\wideBar}[1]{\overbracket[0.3pt][0pt]{\overbracket[0.3pt][0pt]{\mkern-2.8mu #1}}}

\begin{document}
\maketitle

% Enter the first author's name and address:

\centerline{\scshape Irene M. Gamba}
\medskip
{\footnotesize
	\centerline{Department of Mathematics and Oden Institute of Computational Engineering and Sciences}
	\centerline{University of Texas at Austin}
	\centerline{2515 Speedway Stop C1200
		Austin, Texas 78712-1202}
	%\centerline{209 South 33rd Street, Philadelphia, PA 19104}
	
}

\medskip

\centerline{\scshape Milana Pavi\'c-\v Coli\'{c}}
\medskip
{\footnotesize
  \centerline{Department of Mathematics and Informatics}
  \centerline{Faculty of Sciences, University of Novi Sad}
 \centerline{Trg Dositeja Obradovi\'ca 4, 21000 Novi Sad, Serbia}
}

\bigskip

%\centerline{\today}

\bigskip

% The name of the associate editor will be entered by an editorial staff
% "Communicated by the associate editor name" is not needed for special issue.
% \centerline{(Communicated by )}

%The abstract of your paper
\begin{abstract}
We solve the Cauchy problem for the full non-linear homogeneous Boltzmann system of equations describing multi-component monatomic gas mixtures for binary interactions in three dimensions. More precisely,  we show existence and uniqueness of the vector value solution  by means of an existence theorem for ODE systems in Banach spaces under the transition probability rates assumption corresponding to hard potentials rates  in the interval $(0,1]$, with an angular section modeled by  an integrable function of the  angular  transition rates modeling binary scattering effects.  The initial data for the vector valued solutions needs to be a vector of non-negative measures with finite total number density, momentum and strictly positive energy, as well as to have a finite  $L^1_{k_*}(\mathbb{R}^3)$-integrability property corresponding to a sum across each species of  $k_*$-polynomial weighted norms depending on the corresponding mass fraction parameter for each species as much as on the intermolecular potential rates, referred as to the scalar polynomial moment of order $k_*$. 
The existence and uniqueness rigorous results rely on a new angular averaging lemma adjusted to vector values solution that yield a Povzner estimate with  constants that decay with the order of the corresponding dimensionless scalar polynomial moment. In addition, such initial data yields global generation of such scalar polynomial moments at any order as well as their summability of moments to obtain estimates for corresponding scalar exponentially decaying high energy tails, referred as to scalar exponential moments associated to the system solution.   Such scalar polynomial and exponential moments propagate as well.
\end{abstract}

\medskip


\

\tableofcontents

%The title of your section 1

\section{Introduction}

We consider a mixture of $I$ monatomic gases, labeled with $\mathcal{A}_1, \dots, \mathcal{A}_I$.  In the kinetic theory framework, each species of the mixture $\mathcal{A}_i$ is statistically  described with its own distribution function $f_i:=f_i(t,x,v)$, that in general depends on time $t\geq0$, space position $x\in \mathbb{R}^3$ and velocity of molecules $v \in \mathbb{R}^3$ (in this manuscript we restrict ourselves to the spatially homogeneous case, that is, we drop dependence on space position $x$).  The distribution function $f_i$ changes due to binary interactions (or collisions) with other particles. In the mixture setting, these particles can belong to other species  $\mathcal{A}_j$, $j\neq i$. Therefore, the evolution of each $f_i$ involves not only the particle-particle interaction of specie  $\mathcal{A}_i$, but also interactions between $\mathcal{A}_i$ and $\mathcal{A}_j$, $j\neq i$.

In the mixture framework, the evolution of each distribution function $f_i$ describing the mixture component $\mathcal{A}_i$, is governed by the Boltzmann-like equation, that traditionally introduces collision operator as a measure of   its change. Now,  one has multi-species collision operators and  their transition probabilities, or cross sections,  between the different distribution functions describing each component of the mixture \cite{Sir62}. Since all  species are considered simultaneously in a system of  species with binary interactions, one is led to introduce a vector valued set of distribution functions $\mathbb{F}=\left[f_i\right]_{1\leq i \leq I}$, whose evolution is governed by a vector of collision operators, whose $i-$th component (that describes precisely evolution of $f_i$) is $\left[\mathbb{Q}(\mathbb{F})\right]_{i}=\sum_{j=1}^I Q_{ij}(f_i,f_j)$. In this formula, operator $Q_{ij}(f_i,f_j)$ describes influence of species $\mathcal{A}_j$ for the distribution function $f_j$ on species $\mathcal{A}_i$ with the distribution function $f_i$. Note that summation over all $j=1,\dots,I$ is in the spirit of taking into account influence of all species  $\mathcal{A}_j$, $j=1,\dots,I$, on the  considered species $\mathcal{A}_i$.

From a mathematical viewpoint,  the challenging situation occurs when masses of species molecules are not equal (i.e. $m_i\neq m_j$). In such a situation, underlying binary  collisions between molecules lose some symmetry properties which can dramatically change mathematical treatment, for instance  in order to study diffusion asymptotics  when one needs  to show the  compactness of a part of linearized Boltzmann operator  \cite{BGPS}.  In the mixture framework,  a linear system of linearized Boltzmann equations has been recently studied in \cite{BriantDaus16}, corresponding to the perturbative setting of our model  when the non-linear system is  linearized near Maxwellian states corresponding to each species. In this case   authors showed  existence, uniqueness, positivity and exponential trend to equilibrium. \\

In this work, we give  the first existence and uniqueness result for the non-linear system of  spatially homogeneous Boltzmann equations for  multi-species   mixtures with binary interactions in a suitable Banach space.  
We also emphasize that our approach for solving the Cauchy problem  for the Boltzmann equation with variable hard potentials relies on some specific conditions on the initial moments, without requesting entropy boundedness. The hard potentials assumption correspond to collision cross sections  related to the species $\mathcal{A}_i$ and $\mathcal{A}_j$ proportional to the local relative speed with a power exponent $\gamma_{ij}\in(0,1]$,  and $L^1-$integrable   angular part $b_{ij}$, as function of the scattering direction. 

 In addition, the  existence and uniqueness of a vector value solution $\mathbb{F}(t,v)$  need to assume that initially  its scalar zero and second moment (i.e. the  scalar number density and energy of the mixture) are  strictly positive and finite, and additionally that this function has at least an upper bounded $k_*$-polynomial moments,   where $k_*:= \max\{\overline{k},  2+2 \bar{\bar{\gamma}} \}$, for    $\overline{k}=\max_{1\leq i, j \leq I}\{ k^{ij}_*\}$  and   $\bar{\bar{\gamma}} =\max_{1\leq i, j \leq I}\gamma_{ij}$,  is sufficiently large to ensure the prevail of the polynomial  moments of loss term with respect to those same moments of the gain term. Each
 $k^{ij}_* $ depends on the  angular transition rate $b_{ij}$ as well as on the two-body mass fraction $r_{ij}:= m_i/(m_i+m_j)$ associated to  each component on the vector solution.
 All these parameters are defined in the next Section~\ref{Section notation} dedicated to notation, preliminaries and main results.
\\

The result is obtained following general ODE theory that studies    differential equations in  suitable Banach spaces  \cite{Martin}.  In the context of (single) Boltzmann equation, this theory was proposed  as a  main tool in \cite{Bressan} for solving the Cauchy problem with  hard spheres in three dimensions and constant angular transition probability kernel. However, the notes \cite{Bressan}  do not completely verify all conditions of general ODE theory for the Boltzmann equation. This was motivation for \cite{GambaAlonso18} to revise the application of ODE theory from \cite{Martin} in the case of Boltzmann equation with more general hard potentials and integrable angular cross section, and in particular, to provide a complete proof of sub-tangent condition. 
	
One very interesting new aspect from this approach is that the ODE flow in a suitable Banach space without imposing initial bounded entropy  condition yields an alternative approach that allows for  a rather general theory for gathering estimates where one can apply a rather general result  in order to find solutions to the Cauchy problem for Boltzmann type flows  where there is no classical entropy that is dissipated,  or even some conservation laws may not be satisfied. Such problems have already been solved in for polymers kinetic problems \cite{AlonsoLods18},  quantum Boltzmann equation for bosons in very low temperature \cite{GambaAlonsoTran17} and more recently to study the weak wave turbulence models for stratified flows \cite{GambaSmithTran18}.\\

After proving the existence and uniqueness of the vector value solution $\mathbb{F}$ to the Boltzmann system, we turn to the study of generation and propagation of scalar polynomial and exponential  moments of its solution $\mathbb{F}$.

The techniques we use in this manuscript are  adaptations or extensions  of results that have been developed for  scalar Boltzmann type equations models.

 In the case of the classical Boltzmann equation for the single elastic monatomic gas model, polynomial moments have been exclusively  considered, for instance,  in \cite{Des93} and \cite{wennberg97} for hard  potentials where propagation and generation of such moments was proved.  About the same time, Bobylev introduced in \cite{Bob97} the concept of exponential moment as a measure of the distribution solution tail, referred as to {\sl tail temperature}, 
 by showing that solutions to the Boltzmann equation for monatomic gases, modeled by elastic hard spheres (i.e. power exponent  $\gamma=1$) 
 in three dimensions with a constant angular dependent cross-section as a function of the scattering direction, have inverse Maxwellian weighted moments, globally in time,  whose tail decay rate depend on moments of the initial data. His proof  consists in showing that  infinite sums of renormalized polynomial moments  are summable whose  limit is proportional to a $L^1$- Gaussian weighted norm for the unique probability density function solving the initial value problem associated to the Boltzmann equation, whose rate depends on the initial data  that must also be integrable with a Gaussian weight.   These techniques of understanding moments summability in order to obtain high energy tail behavior for the solution of the Boltzmann equation was extended to  inelastic interactions with stochastic heating sources, shear flows or self-similarity scalings to obtain non-equilibrium statistical  stationary (NESS) states   \cite{GambaBobPanf04} where the exponential rates did not necessarily correspond to Gaussian weighted moments.

This concept in the elastic case was further extended by  \cite{GambaPanfVil09} to collision kernels for  hard potentials (i.e. $\gamma\in (0,1]$) for any angular section
with $L^{1+}$-integrability. Further, generation of exponential moments of order $\gamma/2$ with bounded angular section were shown in \cite{Mouhot06}. 

By then it became clear that the  study of general forms of exponential moments resulted  as a by-product of the analysis of polynomial  moments (or tails), and so a spur of work arose for the improvement of conditions and results that will allow to estimate, globally in time.
These  results were extended to collision kernels for hard potentials  with $\gamma\in(0,2]$  for any angular section
with just $L^{1}$-integrability by a new approach using partial sums summability techniques, rather than using summability studies by power series associated to 
renormalized moments as proposed in \cite{Bob97, GambaBobPanf04, GambaPanfVil09,Mouhot06}. The generation results were improved to obtain exponential moments of order $\gamma$, while Gaussian moments were propagated  for any initial data that would have that property, 
independent of $\gamma$. All these results were  extended to  the angular non-cutoff regime (lack of angular integrability) in \cite{GambaTask18, LuMouhot12} still for hard potentials with $\gamma\in(0,2]$, and in \cite{BobGamba17, PT} for pseudo-Maxwellian and Maxwellian case ($\gamma=0$). 
In the later referenced work, these non-Gaussian tailed moments are called Mittag-Leffler moments as in fact the summability of partial sums is shown to converge  to  an   $L^1$-Mittag-Leffler function weighted norm  for the unique probability density function solving  the initial value problem associated to the Boltzmann equation, whose order and rate depend on the initial data as much as on the order of singularity in the angular section.

A very important tool for the success of summability properties for polynomial moments relies on the fact that such moments are both created and propagated depending on how moments of the collision operator can be estimated: the positive part of the (gain) collision operator must have a decay rate with respect to the moment order while   the negative part of such moments prevails  in the dynamics,  when sufficiently many moments are taken into account.

This is indeed a key step, arising  as a consequence of an {\em angular averaged Povzner lemma}.  In the case of single gas components,
these estimates  are based on integration of the collision operator against  polynomial test functions on the pre-collisional velocities in the sphere.
 While these objects were  originally introduced by Povzner \cite{Pov} in 1960s, a   sharper form that uses the conservation of energy and angular averaging was introduced  in \cite{Bob97} for the case of hard spheres in three dimensions with a constant angular cross section, where the polynomial test functions are proportional to even powers of the velocity magnitude. Later this technique was extended  in \cite{GambaBobPanf04} for the inelastic collision with heating sources, 
in \cite{GambaPanfVil09} to the elastic case with hard potentials with $L^{1+}$ integrable angular cross section, as well as  in \cite{Gamba13}  for the case with just $L^{1}$ integrable angular cross section. Further, the approach was enlarged to hard spheres with non-integrable angular cross section in \cite{LuMouhot12} and \cite{GambaTask18} for hard potentials. { {\sl All of these estimates were developed for the mono-component model. } }

  Hence, the angular averaged  Povzner lemma is our starting point in the case of mixtures as well.  However,  it requires a subtle modification of the polynomial weight that define {\emph{the scalar  moment  for the mixture}}, to be defined in \eqref{SPM} next section, that renormalizes the polynomial test function from just even powers of the magnitude of the velocity vector  to a dimensionless bracket form independent of mass density units, as the mono-component treatment  to obtain moment estimates from \cite{Bob97} for the elastic case, or from \cite{GambaPanfVil09} for inelastic hard sphere interactions,  can not be directly extended to the mixture case, when masses are possibly different. 

This facts enticed us to introduce a new approach that relies on the way to 
rewrite collisional rules and scalar polynomial moments  in such dimensionless, independent of mass density units form  that is very  convenient to obtain a  convex combination form between the conserved local quantities for a binary interaction, namely, local center of mass and  energy. 
As a consequence, we conclude that  averaging over the $S^2$-sphere yield  decay properties as a function of the moment order  for as long as angular kernel is $L^1-$integrable  on $S^2$. In particular, these decay properties  will be significantly influenced by  the fact how much species masses are disparate. It will be shown that as much as renormalized species masses deviates one from each other, the decay rate will be more slowly.\\

The paper is organized as follows. In Section \ref{Section notation} we introduce notation and preliminaries, and state the main results, namely the Existence and Uniqueness Theorem for the vector value solution of the homogeneous Boltzmann system, and then generation and propagation of both  scalar polynomial and exponential moments. Then in Section \ref{Section Kinetic Model} we describe in details kinetic model that we use. Section \ref{Section prelim Lemmas} contains two preliminary Lemmas that we need for further work, including Povzner lemma. Sections \ref{Section Ex Uni proof}, \ref{section proof generation of poly} and \ref{Section gen prop exp mom} are devoted to  proofs of our main results. A final Appendix contains some auxiliary calculations relevant to our estimates.

\section{Notation, Preliminaries and Main Results}\label{Section notation}

\subsection{Notation and Preliminaries}

In this paper, we consider mixture of $I$ gases, and we label its components with $\mathcal{A}_1$, $\dots$, $\mathcal{A}_I$. Each component of the mixture $\mathcal{A}_i$, $i=1,\dots,I$, is described with its own distribution function, denoted with $f_i:=f_i(t,v)\geq0$, that, in this manuscript, depends on time $t>0$ and velocity $v\in \mathbb{R}^3$. Fixing some $i\in \left\{1,\dots,I\right\}$,  distribution function $f_i$ satisfy Boltzmann like equation, which now, in the mixture context, has to take into account influence of all other components of the mixture on  species $\mathcal{A}_i$. In the kinetic theory style, this is achieved by defining collision operator $Q_{ij}$ for each $j=1,\dots,I$ that measures interaction between species $\mathcal{A}_i$ that we fixed and all the others $\mathcal{A}_j$, $j=1,\dots,I$, including itself $\mathcal{A}_i$. If the species $\mathcal{A}_j$ are described with distribution functions $f_j$, then the evolution of $f_i$ is described via 
\begin{equation}\label{boltzmann i}
\partial_t f_i(t,v) = \sum_{j=1}^I Q_{ij}(f_i,f_j)(t,v), \qquad i=1,\dots,I. 
\end{equation} 

The form of $Q_{ij}$, for  distribution functions $f$ and $g$ and  any $i,j=1,\dots I$,  is given by the non-local bilinear form
\begin{equation}\label{QijJ}
Q_{ij}(g,h)(v)= \int_{\mathbb{R}^3} \int_{S^2} \left(  \frac1{\mathcal{J}} \, g(v'_{ij}) h(v'_{*ij}) - g(v) h(v_*)  \right) \mathcal{B}_{ij}(v,v_*,\sigma) \, \mathrm{d} \sigma \, \mathrm{d}v_*,
\end{equation}
where pre-collisional quantities $v'_{ij}$ and $v'_{*ij}$ depend on post-collisional ones $v$, $v_*$ and parameter $\sigma$, as much as on the   masses $m_i$ and $m_j$ mass of colliding  particles of species $\mathcal{A}_i$ and $\mathcal{A}_j$ respectively, in the following manner
\begin{equation}\label{collisional rules intro}
v'_{ij} =\frac{m_i v + m_j v_*}{m_i + m_j} + \frac{m_j}{m_i + m_j} \left| v-v_* \right| \sigma, \quad v'_{*ij} =\frac{m_i v + m_j v_*}{m_i + m_j} - \frac{m_i}{m_i + m_j} \left| v-v_* \right| \sigma.
\end{equation}
%with $m_i$ being the mass of a particle of species $\mathcal{A}_i$, $i=1,\dots, I$. 
The collisional rules \eqref{collisional rules intro} can be written in scattering direction coordinates  (or in a center of mass reference framework) by introducing the velocity of center of mass $V_{ij}$ and relative velocity $u$ of the two colliding particles,
\begin{equation}\label{center-of-mass}
V_{ij}:=\frac{m_i v + m_j v_*}{m_i + m_j}, \qquad u:=v-v_*,
\end{equation}
as follows 
\begin{equation}\label{cof-rv-coor}
v'_{ij} =V_{ij}+ \frac{m_j}{m_i + m_j} \left| u \right| \sigma, \quad v'_{*ij} =V_{ij} - \frac{m_i}{m_i + m_j} \left| u \right| \sigma,
\end{equation}
or equivalently, introducing the two-body mass fraction parameter $r_{ij}=  \frac{m_i}{m_i + m_j}\in (0,1)$, associated to one of the particles, say $m_i$  without loss of generality, 
\begin{equation}\label{cof-rv-coor-2}
v'_{ij} =V_{ij}+ (1-r_{ij})\left| u \right| \sigma, \quad v'_{*ij} =V_{ij} - r_{ij} \left| u \right| \sigma.
\end{equation}

\begin{remark}\label{eliminate subindex ij}
For simplicity of notation, from now on, we will eliminate subindices $i,j$ from $v'_{ij}$, $v'_{*ij}$, $V_{ij}$ and $r_{ij}$.
\end{remark}

The transition probability rates  or collision cross section terms $\mathcal{B}_{ij}$ are positive functions  supposed to satisfy the following micro-reversibility assumptions
\begin{equation}\label{Bij micro}
\mathcal{B}_{ij}(v,v_*,\sigma) = \mathcal{B}_{ij}(v',v'_*,\sigma')=\mathcal{B}_{ji}(v_*,v,-\sigma),
\end{equation}
where  $\sigma=u'/\left|u'\right|$ and $u'=v'-v'_*$ (note that then   $\sigma'=u/\left|u\right|$).

The factor in the positive non-local binary term  ${\mathcal{J}}= \left| \det J_{(v',v'_*,\sigma')/(v,v_*,\sigma)} \right|$  is  the absolute value of determinant of the Jacobian  associated to the exchange of velocity variables transformation \eqref{collisional rules intro} from pre to post for the given binary interaction.   The Jacobian of this transformation  can be easily computed by passing to the scattering direction coordinates i.e by considering the following mappings $(v',v'_*,\sigma') \mapsto (u',V',\sigma')\mapsto (\left|u'\right|, \frac{u'}{\left|u'\right|} , V', \sigma') \mapsto (\left|u\right|, \frac{u}{\left|u\right|} , V, \sigma) \mapsto (u, V, \sigma) \mapsto (v, v_*, \sigma)$, with the notation \eqref{center-of-mass} and using Remark \ref{eliminate subindex ij}. The first mapping is of unit Jacobian from definition of $u$ and $V$, the second one is passage from Cartesian to spherical coordinates for $u'$. Since from the collisional rules \eqref{collisional rules intro}  it follows $\left|u'\right| = \left|u\right|$ and $V'=V$ the passage from primes to non-primes described in  the third mapping is of unit Jacobian. Then we pass from spherical to Cartesian coordinates for $u$ and finally go back to the original variables $(v,v_*,\sigma)$. Thus, the Jacobian is computed as the decomposition of the mentioned mappings, 
\begin{equation*}
\mathcal{J} = 1 \cdot \frac{1}{\left|u'\right|^2} \cdot 1 \cdot \left| u \right|^2 \cdot 1  = 1,
\end{equation*}	
since $\left|u'\right| = \left|u\right|$.  Therefore,  each $Q_{ij}$ from \eqref{QijJ} simple becomes,
\begin{equation}\label{Qij}
Q_{ij}(g,h)(v)= \int_{\mathbb{R}^3} \int_{S^2} \left(g(v') h(v'_*) - g(v) h(v_*)  \right) \mathcal{B}_{ij}(v,v_*,\sigma) \, \mathrm{d} \sigma \, \mathrm{d}v_*.
\end{equation}

\bigskip

Since we consider a mixture as a whole, it will be convenient to introduce the following vector notation. We put all distribution functions $f_i$, $i=1,\dots, I$ into vector of distribution functions 
\begin{equation}\label{distribution function vector form}
\mathbb{F}=[f_i]_{1\leq i \leq I}.
\end{equation} 
Moreover, a vector value collision operator is defined
\begin{equation}\label{collision operator vector form}
\mathbb{Q}(\mathbb{F})=\left[\sum_{j=1}^I Q_{ij}(f_i,f_j)\right]_{1\leq i \leq I}.
\end{equation}
Then the system of  Boltzmann equations \eqref{boltzmann i} can be written in a vector form
\begin{equation}\label{BE intro}
\partial_t \mathbb{F}(t,v) = \mathbb{Q}(\mathbb{F})(t,v).
\end{equation}

\begin{definition}[\sl Bracket forms for the mixture's dimensionless polynomial moments independent of mass density units]\label{SPM}
Let $\mathbb{F}=[f_i]_{1\leq i \leq I}$ be a suitable vector value distribution function. Let mixture's bracket forms be denoted  by 
\begin{equation}\label{vi}
\left\langle v \right\rangle_i :=\sqrt{ 1+ \frac{m_i}{\sum_{j=1}^I m_j} \left| v \right|^2}, \qquad v\in \mathbb{R}^3. 
\end{equation}

 \emph{Scalar polynomial moments independent of mass density units} of order $q\geq 0$ for $\mathbb{F}$ is defined with
	\begin{equation}\label{poly moment}
	\mathfrak{m}_q[\mathbb{F}](t) =  \sum_{i=1}^I \int_{\mathbb{R}^3} f_i(t,v) \left\langle v \right\rangle_i^q \mathrm{d}v.
	\end{equation}
	In particular, we define scalar polynomial moment of zero order for each species $\mathcal{A}_i$
	\begin{equation*}
	\mathfrak{m}_{0,i}[\mathbb{F}](t) = \int_{ \mathbb{R}^3} f_i(t,v) \, \mathrm{d}v, \quad i=1,\dots,I,
	\end{equation*}
	having in mind that $\sum_{i=1}^I \mathfrak{m}_{0,i}[\mathbb{F}] = \mathfrak{m}_0[\mathbb{F}]$.
	
	 \emph{Scalar exponential moment}, or exponential weighted $L^1-$forms, for $\mathbb{F}$ of a rate $\boldsymbol{\alpha}:=(\alpha_1,\dots,\alpha_I)$, $\alpha_i>0$, and an order $\mathbf{s}:=(s_1,\dots,s_I)>0$, $0<s_i\leq 2$,   is defined by 
	\begin{equation}\label{exp moment}  
	\mathbb{\mathcal{E}}_{ \mathbf{s}}[\mathbb{F}](\boldsymbol{\alpha},t) =  \sum_{i=1}^I \int_{\mathbb{R}^3} f_i(t,v) e^{\alpha_i \left\langle v \right\rangle_i^{s_i}} \mathrm{d}v. 
	\end{equation} 
	The case $s_i =2$, $\forall i$, is referred to as inverse Maxwellian (or Gaussian) moment, otherwise they are super exponential moments (some authors referred as stretched exponentials though this concept usually refers to exponential times).

\end{definition} 

\begin{remark}
It can be noticed that  such both dimensionless polynomial and  exponential moments for the mixture are defined as a sum of the resulting moments corresponding to each species independent of mass density units. In particular, when $\mathbb{F}$ solves the Boltzmann system of equations \eqref{BE intro}, then $\mathfrak{m}_{0,i}[\mathbb{F}]$ is interpreted as number density of the species $\mathcal{A}_i$, for any $i=1,\dots,I$, while the zeroth scalar moment  $\mathfrak{m}_0[\mathbb{F}]$ is the total number density of the mixture. Moreover, the second scalar moment $\mathfrak{m}_2[\mathbb{F}]$ represents total energy of the mixture.
\end{remark}

\begin{remark}
	 If, for given exponential moments individually for each species $\mathcal{A}_i$, we seek for the maximum value of both their  rate and order, i.e. 
	\begin{equation}\label{max rate max order}
	\hat{\alpha}= \max_{1\leq i \leq I} \alpha_i, \quad 	\hat{s}= \max_{1\leq i \leq I} s_i,
	\end{equation}
		then $$ \mathbb{\mathcal{E}}_{ \mathbf{s}}[\mathbb{F}](\boldsymbol{\alpha},t) \leq  \sum_{i=1}^I \int_{\mathbb{R}^3} f_i(t,v) \, e^{\hat{\alpha} \left\langle v \right\rangle_i^{\hat{s}}} \mathrm{d}v =:  \mathcal{E}_{\hat{s}}[\mathbb{F}](\hat{\alpha},t)$$
Therefore, finiteness of the  exponential moment $\mathbb{\mathcal{E}}_{ \mathbf{s}}[\mathbb{F}](\boldsymbol{\alpha},t)$ is  a consequence of the finiteness of $\mathcal{E}_{\hat{s}}[\mathbb{F}](\hat{\alpha},t)$, with $\hat{\alpha}$ and $\hat{s}$ as in \eqref{max rate max order}, for any time $t\geq 0$. 
\end{remark}

\subsubsection{Functional space}
We work in $L^1$ space weighted  polynomially  in velocity $v$ and summed over all species, that is
\begin{equation}
\begin{split}\label{space L_k^1}
L_k^1 &= \left\{ \mathbb{F}= [f_i]_{1\leq i \leq I} \ \text{measurable}: \sum_{i=1}^I \int_{\mathbb{R}^3} \left| f_i(v)  \right|\left\langle v \right\rangle_i^k \mathrm{d}v < \infty, \ k\geq 0 \right\}\\
\end{split}\end{equation} 
where the polynomial weight  was defined in \eqref{vi} by $\left\langle v \right\rangle_i ={\left( 1+ \frac{m_i}{\sum_{j=1}^I m_j} \left| v \right|^2 \right)^{1/2}}.$

Its associated norm is 
\begin{equation}\label{norm I}
\left\| \mathbb{F} \right\|_{L_k^1} = \sum_{i=1}^I \int_{\mathbb{R}^3} \left| f_i(v)  \right|\left\langle v \right\rangle_i^k \mathrm{d}v.
\end{equation}
Note that if $\mathbb{F}\geq 0$, then its norm in $L_k^1$ is precisely its polynomial moment of order $k$, i.e.
$
\left\| \mathbb{F} \right\|_{L_k^1} := \mathfrak{m}_k[\mathbb{F}].
$

Sometimes we will consider species separately, i.e. fix some component of the mixture $\mathcal{A}_i$. We define a space together with its norm
\begin{equation*}
L_{k,i}^1 = \left\{ g \ \text{measurable}: \int_{\mathbb{R}^3} \left|g(v)\right|  \left\langle v \right\rangle_i^{k} \mathrm{d}v < \infty, k\geq 0 \right\}, \left\| g \right\|_{L_{k,i}^1} = \int_{\mathbb{R}^3} \left|g(v)\right| \left\langle v \right\rangle_i^{k} \mathrm{d}v.
\end{equation*}
Note that the norm of $\mathbb{F}$ in $L_k^1$ is related to the norm of its components $f_i$ in the space $L_{k,i}^1$ via 
$
\left\| \mathbb{F} \right\|_{L_k^1} = \sum_{i=1}^{I} \left\| f_i \right\|_{L_{k,i}^1}.
$

Finally, since we use bracket forms $\left\langle \cdot \right\rangle$ defined in \eqref{vi}, the monotonicity property holds, i.e.
\begin{equation}\label{monotonicity of norm}
\left\| f_i \right\|_{L_{k_1,i}^1} \leq \left\| f_i \right\|_{L_{k_2,i}^1} \ \text{and} \ \left\| \mathbb{F} \right\|_{L_{k_1}^1} \leq \left\| \mathbb{F} \right\|_{L_{k_2}^1}, \text{whenever} \ 0\leq \ k_1 \leq k_2.
\end{equation}

\subsection{Main Results}

We study the Cauchy problem for system of spatially homogeneous Boltzmann equations  for the mixture of gases $\mathcal{A}_1$, $\dots$,  $\mathcal{A}_I$:
\begin{equation}\label{Cauchy problem}
\left\{ \begin{split} & \partial_t \mathbb{F}(t,v) = \mathbb{Q}(\mathbb{F})(t,v), \quad  t>0, \ v \in \mathbb{R}^3, \\& \mathbb{F}(0,v) = \mathbb{F}_0(v),\end{split} \right.
\end{equation}
where $\mathbb{F}$ is a vector of distribution functions $\mathbb{F}=[f_i]_{1\leq i \leq I}$,  $f_i$ being distribution function of the component $\mathcal{A}_i$, $i=1,\dots,I$, as defined in \eqref{distribution function vector form}, and  $\mathbb{Q}(\mathbb{F})=\left[\sum_{j=1}^I Q_{ij}(f_i,f_j)\right]_{1\leq i \leq I}$ is a collision operator introduced in (\ref{Qij}, \ref{collision operator vector form}).

We consider the particular case  when  the transition probability terms $\mathcal{B}_{ij}$, $i,j=1,\dots,I$  are assumed to   take the  form 
\begin{equation}\label{cross section}
\mathcal{B}_{ij}(v,v_*, \sigma) =\left|u\right|^{\gamma_{ij}} \, b_{ij}(\sigma \cdot \hat{u}), \ \gamma_{ij}\in (0,1], \ \text{and} \ b_{ij}(\sigma \cdot \hat{u})\in L^1(S^2; \mathrm{d}\sigma),
\end{equation}
where $u:=v-v_*$, $\hat{u}:=u/\left|u\right|$. This form of cross section corresponds to variable hard potentials with an integrable angular part.   In the mixture setting, both potential $\gamma_{ij}$ and angular kernel $b_{ij}$ may depend on species $\mathcal{A}_i$ and $\mathcal{A}_j$. In order to satisfy micro-reversibility assumptions \eqref{Bij micro}, it is supposed that
\begin{equation*}
\gamma_{ij} = \gamma_{ji}, \quad \text{and} \quad b_{ij}(\sigma \cdot \hat{u}) = b_{ji}(\sigma \cdot \hat{u}),
\end{equation*}
for any choice $i,j=1,\dots,I$. Moreover,  let $\overline{\gamma}$  and $\bar{\bar{\gamma}}$  denote respectively the minimum and  the maximum value of potentials $\gamma_{ij}$, i.e.
\begin{equation}\label{gamma max}
\overline{\gamma} =  \min_{1\leq i, j \leq I}\gamma_{ij}, \quad  \bar{\bar{\gamma}}= \max_{1\leq i, j \leq I}\gamma_{ij}.
\end{equation}

\subsubsection{Povzner lemma by angular averaging} The essential ingredient of this manuscript is the Povzner lemma obtained by averaging in the scattering angle representation of the collision kernel, originally introduced in \cite{Bob97}, \cite{GambaBobPanf04},  for the case of elastic and inelastic collisions.
It estimates the positive contribution of the collision operator after integration  against $\sigma \in S^2$, that is crucial for all further proofs. 

\begin{lemma}[Povzner  lemma by angular averaging for the mixing model]\label{Povzner intro} Let the angular part $b_{ij}(\sigma \cdot \hat{u})$ of the cross-section be integrable in $\sigma$ variable (that is $b_{ij} \in L^1(S^2; \mathrm{d}\sigma)$), $\hat{u}$ being the normalized relative velocity $u=v-v_*$.
	Let $v'$ and $v'_*$ be functions of $v,v_*, \sigma$ as in \eqref{collisional rules intro}, with $m_i, m_j >0$.  Then the following estimate holds  for any  fixed $i, j$,
	\begin{equation}\label{Povzner estimate gain}
	\int_{S^2} \left( \left\langle v' \right\rangle_i^k +  \left\langle v'_* \right\rangle_j^k  \right) \, b_{ij}(\sigma \cdot \hat{u}) \, \mathrm{d}\sigma \leq  \boldsymbol{\mathcal{C}}^{ij}_\frac{k}{2} \left(\left\langle v \right\rangle_i^2 + \left\langle v_* \right\rangle_j^2 \right)^\frac{k}{2},
	\end{equation}
	where constant $\boldsymbol{\mathcal{C}}^{ij}_\frac{k}{2}$ tends to zero as $k$  grows 
	and moreover
	\begin{equation}\label{povzner constant prop}
\boldsymbol{\mathcal{C}}^{ij}_\frac{k}{2} -	\left\| b_{ij} \right\|_{L^1(\mathrm{d}\sigma)} <0, \qquad \text{for  any} \ k\ge  k^{ij}_*, \ \  \ 1\leq i, j \leq I,
	\end{equation}	
where each   $k^{ij}_*$ depends on $b_{ij}$ and  $r_{ij}$.
\end{lemma}

The proof of  Lemma~\ref{Povzner intro} genuinely reflects difference between single and multicomponent gas, with an accent on writing collisional rules in a  convex combination form for mixtures, in contrast to symmetric or ``half-half" writing for the single component gas. It turned out that single component case due to symmetry  had a lot of room for estimates and further simplification, presented in \cite{GambaBobPanf04} for example. For  mixtures, this is not the case any longer, and writing should be exact as much as possible: we use Taylor expansion of second order with a reminder in the integral form, and estimates are done only in the reminder.

A very important consequence of the Povzner lemma is the ability to estimate moments of the collision operator. In particular, averaging over the sphere yields decay properties of the gain term polynomial  moment with respect to its order. This decay allows polynomial moments of loss term to prevail in dynamics, when sufficiently many moments are taken into account. In a single component gas, it suffices to take $2+$ order of polynomial moment, that is  slightly more than energy, to obtain this property \cite{GambaBobPanf04}.  Mixtures bring great novelty in this aspect, too: decay properties of the constant issuing from the Povzner lemma strongly depend on the two-body mass fraction parameter $r_{ij}$. We study this issue in detail in the case $b_{ij} \in  L^\infty(S^2; \mathrm{d}\sigma)$ when it is possible to explicitly calculate the  constant $\boldsymbol{\mathcal{C}}^{ij}_{k/2}$ from \eqref{povzner constant prop}. It will be shown that when $r_{ij}=1/2$ (which corresponds to $m_i=m_j$), we recover the same decay properties of the constant $\boldsymbol{\mathcal{C}}^{ij}_{k/2}$ as in the case of single gas component. However, when mixtures are considered, we observe that as much as $r_{ij}$ deviates from $1/2$, the  larger $k^{ij}_*$ that ensures  \eqref{povzner constant prop} is,  or  larger and larger order of moment that guarantees prevail of loss term moment is.

\subsubsection{Existence and uniqueness theory}
In this manuscript, we discuss existence and uniqueness for the vector value   solution $\mathbb{F}$ to the initial value problem \eqref{Cauchy problem} of space homogeneous Boltzmann equations for monatomic gas mixtures, with transition probabilities (or collision kernels)  associated to species $\mathcal{A}_i$ and $\mathcal{A}_j$, $i,j \in \left\{1,\dots,I\right\}$ having hard potential growth of order $\left|u\right|^{\gamma_{ij}}$ for $\gamma_{ij} \in (0,1]$ and an  integrable angular part $b_{ij}$, with an initial total mixture number density and energy bounded below (i.e. the initial data can not be singular measure), and  have at least  a $k_*$ (scalar) polynomial moments, 
\begin{multline}\label{kstar}
k_* \geq \max\{ \overline{k} , 2+2\bar{\bar{\gamma}} \}  \ \quad\ \text{for} \ \ 
\overline{k} =  \max_{1\leq i, j \leq I}\{ k^{ij}_*\} \\  \text{and}\ \ \overline{\gamma}=   \min_{1\leq i, j \leq I}\gamma_{ij}, \quad \bar{\bar{\gamma}} = \max_{1\leq i, j \leq I}\gamma_{ij}.
\end{multline} 
 chosen to ensure the inequality \eqref{povzner constant prop} holds for any $i,j=1,\dots,I$.

Such a study  fits into an abstract framework of ODE theory in Banach spaces, which can be found in \cite{Martin}. For the Boltzmann equation, the application of this theory was clarified in \cite{GambaAlonso18}, after being recognized in \cite{Bressan}. The formulation of  theorem that we apply in this manuscript is given in Appendix \ref{Appendix exi and uni}. As for the choice of Banach space, it is known that natural Banach space to solve the Boltzmann equation is $L^1$ polynomially weighted, or in mixture setting space $L_k^1$ defined in \eqref{space L_k^1}. More precisely, here we take $k=2$, because the norm in that space is related to energy whose conservation  is exploited. \\

In order to apply Theorem \ref{Theorem general}, we need to find an invariant region $\Omega \subset L_2^1$ in which collision operator $\mathbb{Q}: \Omega \rightarrow L_2^1$ will satisfy \emph{(i) H\"{o}lder continuity},  \emph{(ii) Sub-tangent} and \emph{(iii) one-sided Lipschitz} conditions. 

To that end, we first study the  map $\mathcal{L}_{\overline{\gamma}, k_*}:[0,\infty) \rightarrow \mathbb{R}$,  defined with $$\mathcal{L}_{\overline{\gamma}, k_*}(x)=- Ax^{1+ \frac{\overline{\gamma}}{k_*}}+  B   x,$$ 
where $A$ and $B$ are positive constants,  $\overline{\gamma}\in (0,1]$ and  $k_*$ defined in  \eqref{kstar}. This map has only one root, denoted with $x^*_{\overline{\gamma}, k_*}$, at which $\mathcal{L}_{\overline{\gamma}, k_*}$ changes from positive to negative. Thus, for any $x\geq0$, we may write
\begin{equation*}
\mathcal{L}_{\overline{\gamma}, k_*}(x) \leq \max_{0\leq x \leq x^*_{\overline{\gamma}, k_*}} \mathcal{L}_{\overline{\gamma}, k_*}(x) =: \mathcal{L}^*_{\overline{\gamma}, k_*}.
\end{equation*}
Define
\begin{equation}\label{c dva plus dva gama def}
C_{k_*} := x^*_{\overline{\gamma}, k_*} +  \mathcal{L}^*_{\overline{\gamma}, k_*}.
\end{equation}

Now, we are in position to define  the bounded, convex and closed subset $\Omega \subset L_2^1$,
\begin{multline*}
\Omega = \Big\{ \mathbb{F}(t, \cdot) \in L_2^1: \mathbb{F}\geq 0 \ \text{in} \ v,  \sum_{i=1}^I \int_{\mathbb{R}^3} m_i v \, f_i(t,v)  \mathrm{d}v=0,  
\\
\exists \,  c_{0}, C_0, c_{2},  ,  C_{2},  C_{2+\varepsilon} >0, \text{and} \ C_0 < c_{2}, \text{ such that} \ \forall t\geq0, 
\\
c_{0} \leq \mathfrak{m}_0[\mathbb{F}](t) \leq C_0,  \  c_{2} \leq \mathfrak{m}_2[\mathbb{F}](t) \leq C_{2},   
\\
\mathfrak{m}_{2+\varepsilon}[\mathbb{F}](t) \leq C_{2+\varepsilon}, \ \text{for } \ \varepsilon>0,  
\\
\mathfrak{m}_{k_*}[\mathbb{F}](t) \leq C_{k_*}, \ \text{with} \ C_{k_*} \ \text{from } \eqref{c dva plus dva gama def}   \  \Big\},
\end{multline*}
where 
\begin{equation*}
\mathfrak{m}_{2+\varepsilon}[\mathbb{F}](t) = \left\| \mathbb{F} \right\|_{L_{2+\varepsilon}^1}= \sum_{i=1}^I \int_{\mathbb{R}^3} \left| f_i(t,v)  \right|\left\langle v \right\rangle_i^{2+\epsilon} \mathrm{d}v,
\end{equation*}
for any $\epsilon>0$, which can be arbitrary small.

Then, existence and uniqueness  theory of  a vector value $\mathbb{F}$ solution to the Cauchy problem \eqref{Cauchy problem} fits into the study of ODE in  a Banach space $(L_2^1, \left\| \cdot \right\|_{L_2^1})$ and its bounded, convex and closed subset $\Omega$. The collision operator $\mathbb{Q}$ is  viewed as a map $\mathbb{Q}: \Omega \rightarrow L_2^1$. We will show that it satisfies H\"{o}lder continuity, sub-tangent and one-sided Lipschitz conditions, which will enable us to prove the following Theorem.

\begin{theorem}[Existence and Uniqueness]\label{theorem existence uniqueness} Assume that $\mathbb{F}(0,v)=\mathbb{F}_0(v) \in \Omega$. Then the Boltzmann system \eqref{Cauchy problem} for the cross section \eqref{cross section} has the unique solution in $\mathcal{C}\left(\left[0, \infty \right), \Omega \right) \cap \mathcal{C}^1\left(\left(0,\infty  \right), L_2^1 \right)$.
\end{theorem}

\begin{remark} Let us point out that for the existence and uniqueness result no conditions on initial entropy are necessary. However, if the initial data has finite entropy, then the entropy inequality implies that it will remain bounded for all times. Let us give a sketch of the proof. Definition of the entropy and entropy inequality is taken from \cite{DesMonSalv}, Proposition 1. 
\begin{definition}[Mixture entropy and entropy production]
Let $\mathbb{F}$ be a vector value distribution function as in \eqref{distribution function vector form}. The (mixture) entropy is defined as
\begin{equation}\label{entropy def}
\eta(t) = \sum_{i=1}^I \int_{ \mathbb{R}^3} f_i \log f_i \mathrm{d}v,
\end{equation}
while the (mixture) entropy production is given with
\begin{equation}\label{entropy production def}
D(\mathbb{F})= \sum_{i=1}^I \int_{ \mathbb{R}^3} \left[\mathbb{Q}(\mathbb{F})\right]_i \log f_i \mathrm{d}v.
\end{equation}
\end{definition}
Then the following Proposition holds.
\begin{proposition}[Entropy inequality or the first part of the H-theorem, \cite{DesMonSalv}]
Let us assume that the cross section terms $\mathcal{B}_{ij}$, $1\leq i, j \leq I$, are positive almost everywhere and that  $\mathbb{F}\geq0$ is such that both collision operator $\mathbb{Q}(\mathbb{F})$ and entropy production are well defined. Then the entropy production is non-positive, i.e.
$
D(\mathbb{F}) \leq 0.
$
\end{proposition}
As an immediate consequence, we get from the Boltzmann equation that $\partial_t \eta \leq 0$, or in other words, $\eta(t)\leq \eta(0)$, for any $t\geq0$. Therefore, we conclude that the entropy inequality implies that mixture entropy remains bounded at any time if initially so.
\end{remark}

\subsubsection{Generation and propagation of polynomial moments} The second part of the manuscript is devoted to the study of generation and propagation of scalar polynomial moments associated to the solution of the Boltzmann system \eqref{Cauchy problem} for the cross section \eqref{cross section}, that initially belongs to $\Omega$. 

First, in the following Lemma, we derive from the Boltzmann system \eqref{Cauchy problem} an ordinary differential inequality for polynomial moment of order $k$, $\mathfrak{m}_k[\mathbb{F}](t)$, for  large enough $k$,  that  relies on the Povzner estimate from Lemma~\ref{Povzner intro},  uniformly in each pair $i,j$. 

\begin{lemma}[Ordinary differential inequality for polynomial moments]\label{poly moments}
	Let $\mathbb{F}=\left[f_i\right]_{i=1,\dots,I}$ be a solution of the Boltzmann system \eqref{Cauchy problem} with the cross section \eqref{cross section}-\eqref{gamma max}. Then the polynomial moment \eqref{poly moment} satisfies the following Ordinary Differential Inequality
	\begin{equation}\label{pomocna 100}
	\frac{\mathrm{d}}{\mathrm{d} t} \mathfrak{m}_k[\mathbb{F}](t) = \sum_{i=1}^I \left[ \mathbb{Q}(\mathbb{F}) \right]_i \left\langle v \right\rangle_i^k \mathrm{d} v
	\leq - A_k \, \mathfrak{m}_k[\mathbb{F}](t)^{1+\frac{\overline{\gamma}}{k}}+  B_k \,  \mathfrak{m}_k[\mathbb{F}](t),
	\end{equation}
	for   large enough $k$ to ensure \eqref{kstar}, and some positive constants $A_k$ and $B_k$.
\end{lemma}
The proof of this Lemma   follows from  comparison principles for ODE's, which yields the generation and propagation estimates 
stated in  the following Theorem that is proved in Section \ref{section proof generation of poly}.

\begin{theorem}[Generation and propagation of polynomial moments]\label{theorem bound on norm}
	Let $\mathbb{F}$ be a solution of the Boltzmann system \eqref{Cauchy problem} with a cross section \eqref{cross section}-\eqref{gamma max} and  an initial data $\mathbb{F}(0,v)=\mathbb{F}_0(v) \in \Omega$. \\
	
	\begin{itemize}
		\item[1.](Generation)  There is a constant $\mathfrak{C}^{\mathfrak{m}}$  such that  for any $k>k_*$ defined in \eqref{kstar} ,
		\begin{equation}\label{bound on norm}
		\mathfrak{m}_k[\mathbb{F}](t)  \leq \mathfrak{C}^{\mathfrak{m}} \left(1-e^{-\frac{\overline{\gamma} B_k t}{k} }   \right)^{-\frac{k}{\overline{\gamma}}},  \qquad \forall t>0,
		\end{equation}
		where constants $\mathfrak{C}^{\mathfrak{m}}$ depend on $A_k$, $B_k$   from \eqref{pomocna 100} and $\overline{\gamma}$. \\
		
		\item[2.](Propagation) Moreover, if $\mathfrak{m}_k[\mathbb{F}](0)<\infty$, then 
	\begin{equation}\label{poly propagation}
	\mathfrak{m}_k[\mathbb{F}](t)  \leq \max\{\mathfrak{C}^{\mathfrak{m}}, \mathfrak{m}_k[\mathbb{F}](0)\}, 
	\end{equation}
	for all $t\geq 0$.
		\end{itemize}
\end{theorem}

Finally, we show that, under the assumed conditions on collision kernel form   \eqref{cross section},
the  renormalized series of moments is summable depending on the moments of the initial data yielding the following result on generation and propagation of exponential, or Mittag-Leffler moments.

\subsubsection{Generation and propagation of exponential moments} With bounds on polynomial moment at hand, one can deal with exponential moments. We prove the following Theorem.

\begin{theorem}[Generation and propagation of  exponential moments]\label{theorem gen prop ML}
	Let $\mathbb{F}$ be a solution of the Boltzmann system \eqref{Cauchy problem} with 	a cross section \eqref{cross section}-\eqref{gamma max}  where $\overline{\gamma}=\bar{\bar{\gamma}} = \gamma_{ij}$ for all $i,j\in \left\{1,\dots,I \right\}$, and  an initial data $\mathbb{F}(0,v)=\mathbb{F}_0(v) \in \Omega$. \\

	\begin{itemize}
		\item[(a)](Generation) There exist constants $\alpha>0$ and $\mathfrak{B}^{\mathcal{E}}>0$ such that
		\begin{equation*}
		\mathcal{E}_{\overline{\gamma}}[\mathbb{F}](\alpha \min\left\{t,1\right\}, t) \leq \mathfrak{B}^{\mathcal{E}}, \quad \forall t\geq 0.
		\end{equation*}   \\
		
		\item[(b)](Propagation) Let $0<s\leq 2$. Suppose that there exists a  constant $\alpha_0 >0$, such that
		\begin{equation}\label{initial data exp prop}
		 \mathcal{E}_{s}[\mathbb{F}]({\alpha_0},0) \leq M_0 < \infty.
		\end{equation}
		 Then there exist  constants $0<\alpha \leq \alpha_0$ and $\mathfrak{C}^{\mathcal{E}}>0$ such that
		\begin{equation}\label{exp propagation p}
	\mathcal{E}_{s}[\mathbb{F}]({\alpha},t)  \leq \mathfrak{C}^{\mathcal{E}}, \qquad \forall t\geq 0.
		\end{equation}
	\end{itemize}
\end{theorem}

\section{Kinetic Model}\label{Section Kinetic Model}

\subsection{Study of collision process}
In our setting molecules are assumed to  interact via elastic collisions. Let us fix two colliding molecules; one of the species $\mathcal{A}_i$ having mass $m_i$ and pre-collisional velocity $v'$ and the another one belonging to the species $\mathcal{A}_j$ with mass $m_j$ and pre-collisional velocity $v'_*$   (note that we here immediately adopted the simplicity of notation pointed out in Remark \ref{eliminate subindex ij}). If the post-collisional velocities are denoted with $v$ and $v_*$, respectively, than the momentum and kinetic energy during the collision are conserved
\begin{align}
m_i v' + m_j v'_* &= m_i v + m_j v_*, \nonumber\\
 m_i \left| v' \right|^2 +  m_j \left| v'_* \right|^2 &= m_i \left| v \right|^2 +  m_j \left| v_* \right|^2. \label{CL micro energy}
\end{align}
As usual, we parametrize these equations with a parameter $\sigma \in S^2$, in order to express pre-collisional velocities in terms of post-collisional ones,
\begin{equation}\label{collisional rules bi}
v' =\frac{m_i v + m_j v_*}{m_i + m_j} + \frac{m_j}{m_i + m_j} \left| v-v_* \right| \sigma, \quad v'_* =\frac{m_i v + m_j v_*}{m_i + m_j} - \frac{m_i}{m_i + m_j} \left| v-v_* \right| \sigma.
\end{equation}
Note that if $m_i=m_j$, then the collisional rules simplify and take the usual single component gas form
\begin{equation}\label{collisional rules equal masses}
v' =\frac{v +  v_*}{2} + \frac{1}{2} \left| v-v_* \right| \sigma, \quad v'_* =\frac{ v +  v_*}{2} - \frac{1}{2} \left| v-v_* \right| \sigma.
\end{equation}
\begin{figure}
	\begin{center}
		\begin{tikzpicture}[scale=2]
		\node (v) at (1.93,0.52) [preaction={shade, ball color=red},pattern=vertical lines, pattern color=black,circle,scale=2.5] {};
		\node (vs) at (-1.93,-0.52) [preaction={shade, ball color=white!70!gray},circle,scale=1, ,pattern=crosshatch dots, pattern color=black] {};
		\node (V) at (1.158,0.312)  [circle,scale=1.7, shade,ball color=blue] {};
		\node (vp) at (1.72,-0.25)  [preaction={shade,ball color=red},circle,scale=2.5, ,pattern=vertical lines, pattern color=black] {};
		\node (vps) at (-1.1,2.57) [preaction={shade,ball color=white!70!gray}, circle,scale=1,pattern=crosshatch dots, pattern color=black ] {};
		\node (V old) at (0,0)	 [circle,scale=1, shade,ball color=blue] {};
		\node (start) at (-1,-3)  [circle,scale=0.1, gray] {};
		\node (vp old) at (1.414,-1.414)   [circle,scale=1, shade,ball color=green!30!black!70!white] {};
		\node (vps old) at (-1.414,1.414)   [circle,scale=1, shade,ball color=green!30!black!70!white] {};
		
		\begin{scope}[>=latex]
		\draw[->,line width=0.6ex,gray!80!black] (V old) -- (0.7071,-0.7071) node[midway, above,black] {$\sigma$};
		\draw[->,line width=0.2ex] (start) --  (vs) node[left,pos=0.98] {$v_*$};
		\end{scope}
	    \begin{scope}[>=latex, on background layer]
	    \draw[densely dotted,color=green!30!black!70!white,line width=0.2ex] (0,0) circle (2cm);
	    \draw[densely dotted,->,color=green!30!black!70!white,line width=0.2ex] (vps old) -- (vp old) node[densely dotted,color=green!30!black,pos=0.3,above right=-0.1cm,line width=0.2ex] {$  \left|u \right| \sigma$};
		\draw[->,line width=0.2ex] (start) -- (v)  node[pos=0.96,right=0.1cm] {\color{red}$v$};
		\draw[->,line width=0.2ex] (vs) -- (v) node[pos=0.3,sloped, black,above left] {$u$};
		\draw[densely dotted,->,color=green!30!black!70!white,line width=0.2ex] (start) -- (V old) node[color=green!30!black,pos=0.9,below left=-0.1cm] {$\frac{v+v_*}{2}$};
		\draw[->,dashdotted,color=blue,line width=0.2ex]  (vps)--(vp) node[color=green!30!black,above right=-0.1cm,midway] {\color{blue}$u'$};
		\draw[densely dotted,->,color=green!30!black!70!white,line width=0.2ex] (start) -- (vp old) node[color=green!30!black,below right,line width=0.2ex] {$ \frac{v +  v_*}{2} + \frac{1}{2} \left|u \right| \sigma$};
		\draw[densely dotted,->,color=green!30!black!70!white,line width=0.2ex] (start) -- (vps old) node[color=green!30!black, left=0.1cm] {$ \frac{v +  v_*}{2} - \frac{1}{2} \left|u \right| \sigma$};
		\draw[->,densely dashed,blue,line width=0.6ex] (V old) --  (V) node[midway, sloped,above, blue] {$\frac{m_i-m_j}{2(m_i+m_j)}u$};
		\draw[dashdotted,->,color=blue,line width=0.2ex] (start) -- (V) node[pos=0.98,below right=-0.05cm] {\color{blue}$V_{ij}$};
		\draw[dashdotted,->,color=blue,line width=0.2ex] (start) -- (vps) node[pos=0.98,left] {\color{black}$v'_*$};
		\draw[dashdotted,->,color=blue,line width=0.2ex] (start) -- (vp) node[pos=0.98,below right=-0.1cm] {\color{red}$v'$};
		\end{scope}	
\end{tikzpicture}
	\end{center}
	\caption{Illustration of the collision transformation, with notation $V_{ij}:=\frac{m_i v + m_j v_*}{m_i + m_j}$, $u:=v-v_*$,  $u':=v'-v'_*$. The displacement of the center of mass with respect to a single component elastic binary interaction is given by $(r_{ij}-\frac 12)u = \frac{m_i-m_j}{2(m_i+m_j)}u$, if $m_i>m_j$.  
	Solid lines denote vectors after collision, or given data. Dash-dotted vectors  represent primed (pre-collisional) quantities  that can be calculated from the given data, and compared to the case $m_i=m_j$, represented by dotted vectors. Dashed vector direction is the displacement along the direction of the relative velocity $u$ proportional to the half difference of relative masses, (which clearly vanishes for  $m_i = m_j$, reducing the model to a classical collision). Note that the scattering direction $\sigma$  is preserved as  the pre-collisional relative velocity $u'$ keeps the same magnitude as the post-collisional $u$,   $u'$ is parallel the reference elastic  pre-collisional relative velocity $|u|\sigma$.  }
	\label{picture}
\end{figure}

Figure \ref{picture} illustrates the collision transformation \eqref{collisional rules bi} and aims at explaining its difference with respect to the collision transformation \eqref{collisional rules equal masses} when masses are equal. Namely, for given $v$, $v_*$, $\sigma$ and $m_i$, $m_j$, we calculate center of mass $V=\frac{m_i v + m_j v_*}{m_i + m_j}$, and velocities $v'$ and $v'_*$ according to \eqref{collisional rules bi}. One can notice that the magnitude of the relative velocity does not change during the collision, i.e.  $\left|v-v_*\right|= \left|v'-v'_*\right|$, as it is when masses are the same. Difference comes with the vector of center of mass:  the vector of center of mass for equal masses $\frac{v+v_*}{2}$ displaces by adding a quantity that is proportional to the difference of masses $m_i-m_j$  and thus is peculiar to the mixture case.   More precisely,
$$
V= \frac{v+v_*}{2} + \frac{m_i-m_j}{2(m_i+m_j)}u,
$$
with $u:=v-v_*$.

\subsection{Collision operators} Collision operators $Q_{ij}$, as defined in \eqref{Qij},  describe binary interactions between molecules of species $\mathcal{A}_i$ and $\mathcal{A}_j$, $i,j=1,\dots,I$.
Fix the species $\mathcal{A}_i$ for any $i=1,\dots,I$, and let its distribution function be $g$. On the other hand, let distribution function $h$ describe species $\mathcal{A}_j$.

Note that each $Q_{ij}$ for a fix $(i,j)$-pair has its corresponding  counterpart, $Q_{ji}$,  that describes interaction of molecules of species $\mathcal{A}_j$ with molecules of species $\mathcal{A}_i$
\begin{equation}\label{collision operator ji}
Q_{ji}(h,g)(v) = \int_{\mathbb{R}^3} \int_{S^2} \left( h(w') g(w'_*) - h(v) g(v_*)  \right) \mathcal{B}_{ji}(v,v_*,\sigma) \, \mathrm{d} \sigma \, \mathrm{d}v_*,
\end{equation}
where pre-collisional velocities $w'$ and $w'_*$ now differ from the previous ones given in \eqref{collisional rules bi} by an exchange of mass $m_i \leftrightarrow m_j$, i.e.
\begin{equation}\label{collisional rules symmetry}
w' =\frac{m_j v + m_i v_*}{m_i + m_j} + \frac{m_i}{m_i + m_j} \left| v-v_* \right| \sigma, \quad w'_* =\frac{m_j v + m_i v_*}{m_i + m_j} - \frac{m_j}{m_i + m_j} \left| v-v_* \right| \sigma.
\end{equation}
When $m_i=m_j$, although primed velocities are the same,  $Q_{ij}$ and $Q_{ji}$ still defer, because of the cross section.

\subsection{Weak form of collision operator}
Testing the collision operator against some suitable test functions  $\psi(v)$ and $\phi(v)$ yields
\begin{multline*}
\int_{\mathbb{R}^3} Q_{ij} (g,h)(v) \psi(v) \mathrm{d} v  \\ =  \iiint_{\mathbb{R}^3 \times \mathbb{R}^3\times S^2} g(v) h(v_*) \left( \psi(v')  - \psi(v) \right) \mathcal{B}_{ij}(v,v_*,\sigma) \mathrm{d}\sigma \mathrm{d}v_* \mathrm{d} v,
\end{multline*}
and
\begin{multline*}
\int_{\mathbb{R}^3} Q_{ji} (h,g)(v) \phi(v) \mathrm{d} v  \\ =  \iiint_{\mathbb{R}^3 \times \mathbb{R}^3 \times S^2}  h(v_*)g(v) \left(  \phi(v'_*) -  \phi(v_*) \right) \mathcal{B}_{ij}(v,v_*,\sigma) \mathrm{d}\sigma \mathrm{d}v_* \mathrm{d} v,
\end{multline*}
where now $v'$ and $v'_*$ are denoting the post-collisional velocities as defined by \eqref{collisional rules bi}. Therefore, looking at these two integrals pairwise, meaning that each time when $Q_{ij}$ is considered we add his pair  $Q_{ji}$, we have
\begin{multline}\label{weak form ij+ji}
\int_{\mathbb{R}^3} \left( Q_{ij} (g,h)(v) \psi(v) + Q_{ji} (h,g)(v) \phi(v) \right) \mathrm{d} v  \\ =  \iiint_{\mathbb{R}^3 \times \mathbb{R}^3 \times S^2} g(v) h(v_*) \left( \psi(v') +  \phi(v'_*)  - \psi(v) - \phi(v_*) \right) \mathcal{B}_{ij}(v,v_*,\sigma) \mathrm{d}\sigma \mathrm{d}v_* \mathrm{d} v,
\end{multline}
with $v'$ and $v'_*$  are now  given by the post-collisional velocities as defined by \eqref{collisional rules bi}. \\

Some choice of test function leads to annihilation of the weak form. Namely, from the conservation laws during collision process,  we see
\begin{equation}\label{collision invariants bi mass}
\int_{\mathbb{R}^3} Q_{ij} (g,h)(v)  \mathrm{d} v  =  0,
\end{equation}
as well as
\begin{equation}\label{collision invariants bi}
\int_{\mathbb{R}^3} \left( Q_{ij} (g,h)(v) \left( \begin{matrix} m_i v \\  m_i \left| v \right|^2 \end{matrix} \right) + Q_{ji} (h,g)(v) \left( \begin{matrix} m_j v \\ m_j \left| v \right|^2 \end{matrix} \right) \right) \mathrm{d} v  = 0.
\end{equation}

Therefore, if we consider distribution function $\mathbb{F}=\left[f_i\right]_{1\leq i \leq I}$, then the weak form \eqref{weak form ij+ji} yields 
\begin{multline}\label{weak form any test}
2 \sum_{i=1}^I \sum_{j=1}^I \int_{\mathbb{R}^3} Q_{ij} (f_i,f_j)(v) \psi_i(v)  \mathrm{d} v =  \sum_{i=1}^I \sum_{j=1}^I \iiint_{\mathbb{R}^3 \times \mathbb{R}^3 \times S^2} f_i(v) f_j(v_*)  \\ \times  \left( \psi_i(v') +  \psi_j(v'_*)  - \psi_i(v) - \psi_j(v_*) \right) \mathcal{B}_{ij}(v,v_*,\sigma) \mathrm{d}\sigma \mathrm{d}v_* \mathrm{d} v.
\end{multline}

\subsection{Conservation laws} Weak forms of collision operator imply some its conservative properties. More precisely, for any suitable $\mathbb{F}$,   \eqref{collision invariants bi mass} implies
\begin{equation}\label{cons operator mass}
\int_{\mathbb{R}^3} \left[\mathbb{Q}(\mathbb{F}) \right]_i\mathrm{d}v  %= \mathfrak{m}_{0,i}[\mathbb{Q}(\mathbb{F})](t) 
= 0, \qquad \text{for any} \ i=1,\dots,I,
\end{equation}
and moreover,  from \eqref{weak form any test} choosing $\psi_\ell(x)=m_\ell \left|x\right|^2$, and $\psi_\ell(x)=m_\ell  x$, $x\in \mathbb{R}^3$, one has
\begin{equation}\label{cons operator energy}
\sum_{i=1}^I \left[ \mathbb{Q}(\mathbb{F}) \right]_i m_i \left|v\right|^2 \mathrm{d} v=0,
\end{equation}
and 
\begin{equation*}
\sum_{i=1}^I \left[ \mathbb{Q}(\mathbb{F}) \right]_i m_i v\, \mathrm{d} v=0,
\end{equation*}
for any time $t\geq 0$.

If $\mathbb{F}$ is a solution to the Boltzmann system \eqref{Cauchy problem}, then these properties imply conservation laws for number density of each species $\mathcal{A}_i$, $i=1,\dots,I$, and total energy of the mixture. Indeed, 
\begin{equation}\label{conservation mass energy}
\partial_t \mathfrak{m}_{0,i}[\mathbb{F}](t)=0, \ \forall i=1,\dots,I,  \qquad \partial_t \mathfrak{m}_2[\mathbb{F}](t)=0.
\end{equation}

\section{Proof of Povzner lemma~\ref{Povzner intro} }\label{Section prelim Lemmas}

The proof of  Povzner lemma~\ref{Povzner intro} by angular averaging for the mixing model entices to obtain estimates for the quantity $ \left\langle v' \right\rangle_i^k +  \left\langle v'_* \right\rangle_j^k $ integrated over sphere $S^2$, that represents  the gain part of \eqref{weak form any test} for $\psi_i(x) = \left\langle x \right\rangle_i^k$. The usual techniques used in \cite{Gamba13} for example, can not be directly adapted when $m_i\neq m_j$. This becomes clear when one writes local kinetic energies of each colliding molecule pair. When $m_i\neq m_j$, these energies can be written  as a certain convex combination, while single component case (or in the same fashion when $m_i=m_j$) correspond to the ``middle" of this convex combination, or to the ``halfs" (see Remark \ref{Remark single vs mixture} below). Single component situation (or when $m_i=m_j$) is therefore ``symmetric", in a sense, and the techniques for  proof of a sharper   Povzner lemma by angular averaging, as done by \cite{Bob97} or \cite{GambaPanfVil09},  
can not be  extended to the mixture case in a straight forward form.

Indeed, in the mixture setting when $m_i\neq m_j$, the proof  of the Povzner lemma~\ref{Povzner intro} in the cases of non-linear gas mixture system uses a non-trivial modification of a  powerful energy identity in scattering angle coordinates. This identity is needed in order to compute moment estimates that clearly show positive moments from the gain collision operator part  are dominated by the moments of the corresponding loss part, which   yields a very sharp estimate sufficient to obtain not only moments propagation and generation, but also  their scaled summability that prove propagation and generation of exponential moment estimates as well.  An energy identity in scattering angle coordinates was first developed in \cite{Bob97,GambaBobPanf04} for the elastic and inelastic case for scalar Boltzmann binary models. While  such identity is rather easy in the elastic single species setting, where local energies are just  the sum of the collision invariant $|v|^2$ and just its interacting counterpart $|v_*|^2$, in the mixing case under consideration the problem becomes highly non-trivial and the local energies to be estimated now depend on binary sums 
of $\left\langle v \right\rangle_i^2 $ and $\left\langle v_* \right\rangle_j^2$ and their corresponding post collisional  sum of  $\left\langle v' \right\rangle_i^2 $ and $\left\langle v'_* \right\rangle_j^2$.

\begin{lemma}[{Energy identity in scattering direction coordinates for the $(i,j)-$pair of colliding particles}]\label{energy-identity-mixture}
Consider any $(i,j)-$pair of interacting velocities $v$ and $v_*$ corresponding to particles masses $m_i$ and $m_j$, respectively, with $i, j$ fixed. 
Let their local micro energy be  $E_{ij} = \left\langle v \right\rangle_i^2 + \left\langle v_* \right\rangle_j^2$,  with $\left\langle v \right\rangle_i^2$ and $\left\langle v_* \right\rangle_j^2$ defined according to \eqref{vi}, and recall  the two-body mass fraction parameter $r_{ij}=\frac{m_i}{m_i+m_j}$  introduced in \eqref{cof-rv-coor}.

Then, there exists a couple of  functions $p_{ij}=p_{ij}(v,v_*,m_i,m_j) $ and $q_{ij}=q_{ij}(v,v_*,m_i,m_j)$ such that,  $p_{ij}+q_{ij} = E_{ij}$ and 
the following representation holds
	\begin{equation}\label{EI1}
	\langle  v'_{ij} \rangle_i^2= p_{ij} + \lambda_{ij}\, \sigma \cdot \hat{V}_{ij}, \ \ \ \qquad 
	\langle  v'_{*ij} \rangle_j^2= q_{ij}  - \lambda_{ij}  \, \sigma \cdot \hat{V}_{ij}.
	\end{equation}
where $\lambda_{ij}:= 2 \sqrt{r_{ij} (1-r_{ij}) (sE_{ij} -1)((1-s_{ij})E_{ij}-1)}$ with $s_{ij}=s_{ij}(v,v_*,m_i,m_j)\in [0,1]$. In particular, this representation preserves the local energy identity 
\begin{equation}\label{EI2}
	\langle  v'_{ij} \rangle_i^2 +	\langle  v'_{*ij} \rangle_j^2\  =\  p_{ij}+q_{ij}\ =\ E_{ij} \ = \ \left\langle v \right\rangle_i^2 + \left\langle v_* \right\rangle_j^2.
\end{equation}
Moreover, the following inequalities hold
\begin{equation}\label{key point}
p_{ij} +\lambda_{ij}  \leq E_{ij} , \quad q_{ij}  + \lambda_{ij}  \leq E_{ij} ,
\end{equation}
for any velocities $v, v_* \in \mathbb{R}^3$ and any masses $m_i, m_j >0$.
\end{lemma}
As we mentioned earlier in Remark \ref{eliminate subindex ij}, we eliminate subindex $ij$ from $E_{ij}$, $p_{ij}$, $q_{ij}$, $\lambda_{ij}$, $s_{ij}$ as we did in Remark \ref{eliminate subindex ij} for $v'_{ij}$,  $v'_{*ij}$, $V_{ij}$ and $r_{ij}$. 
\begin{proof}[Proof of Lemma~\ref{energy-identity-mixture}]  As anticipated, we represent the exchange of coordinates at the interaction using  the center of mass and relative velocity reference frame   \eqref{collisional rules intro} (with its symmetric form \eqref{collisional rules symmetry})
where the angular integration if  performed in the scattering direction corresponding to the post-collisional relative velocity $\sigma=\hat{u'}$. Thus, let's denote with $V$ the vector of center-of-mass  and with $u$ the  relative velocity as in \eqref{center-of-mass},
	\begin{equation*}
	V=\frac{m_i v + m_j v_*}{m_i + m_j}, \qquad u=v-v_*.
	\end{equation*}
	Then, taking the squares of the magnitudes of the post-collisional velocities  given in \eqref{collisional rules bi}, one obtains
	\begin{equation*}
	\begin{split}
	\left|v'\right|^2 &= \left|V\right|^2 + \frac{ m_j^2}{(m_i+m_j)^2} \left|u\right|^2 + \frac{2 m_j}{m_i + m_j} \left|u\right| \left|V\right| \sigma \cdot \hat{V},\\
	\left|v'_*\right|^2 &= \left|V\right|^2 + \frac{m_i^2}{(m_i+m_j)^2} \left|u\right|^2 - \frac{2 m_i }{m_i + m_j} \left|u\right| \left|V\right| \sigma \cdot \hat{V},
	\end{split}
	\end{equation*}
	where $\hat{V}$ denotes the unit vector of $V$. Passing to $\langle \cdot \rangle$  bracket forms from \eqref{vi},  implies
	\begin{equation}\label{kinetic energies squares}
	\begin{split}
	\left\langle v' \right\rangle_i^2 &= 1+  \frac{m_i}{\sum_{\ell=1}^I m_\ell}  \left|V\right|^2 + \frac{m_i m_j^2}{(m_i+m_j)^2 \sum_{\ell=1}^I m_\ell} \left|u\right|^2 \\ &\hspace*{5cm}+ \frac{2 m_i m_j}{(m_i + m_j)\sum_{\ell=1}^I m_\ell} \left|u\right| \left|V\right| \sigma \cdot \hat{V},\\
	\left\langle v'_* \right\rangle_j^2 &= 1+ \frac{m_j}{\sum_{\ell=1}^I m_\ell} \left|V\right|^2 + \frac{m_j m_i^2}{(m_i+m_j)^2\sum_{\ell=1}^I m_\ell} \left|u\right|^2 \\ &\hspace*{5cm}- \frac{2 m_i m_j}{(m_i + m_j)\sum_{\ell=1}^I m_\ell} \left|u\right| \left|V\right| \sigma \cdot \hat{V}.
	\end{split}
	\end{equation}
	Let us introduce the total  energy $E$ of two colliding particles in $\langle \cdot \rangle$ bracket forms, which  is conserved during collision process by  \eqref{CL micro energy},
	\begin{equation*}
	\begin{split}
	E := \left\langle v \right\rangle_i^2 + \left\langle v_* \right\rangle_j^2 =	\left\langle v' \right\rangle_i^2 + 	\left\langle v'_* \right\rangle_j^2.
	\end{split}
	\end{equation*}
	Using the above equations \eqref{kinetic energies squares}, the energy $E$ can be written in $u-V$ notation, as well,
	\begin{equation}\label{convex-E}
	\begin{split}
	E &= 2 + \frac{m_i+m_j}{\sum_{\ell=1}^I m_\ell} \left| V \right|^2 + \frac{m_i m_j}{(m_i+m_j)\sum_{\ell=1}^I m_\ell} \left| u \right|^2.
	\end{split}
	\end{equation}
	\\
	
	The aim is to represent the squares of the post-collisional velocities 	$\left\langle v' \right\rangle_i^2$ and $\left\langle v'_* \right\rangle_j^2$  as a scalar convex combination of different ``parts" of the energy $E$. 	This is achieved by introducing two quantities,
\begin{itemize}
	\item[i)] the  parameter $r\in(0,1)$,  that distributes masses in the following convex pair
 \begin{equation}\label{parameter-r}
	r=\frac{m_i}{m_i+m_j} \  \ \text{and}\ \ 1-r = \frac{m_j}{m_i+m_j},
	\end{equation}
	\item[ii)]  the function $s\in [0,1]$ that convexly partitions the  energy $E $ into two components, one related to $\left| u \right|^2$ and another to $\left| V \right|^2$, 
 using the above identity \eqref{convex-E} as follows
	\begin{equation}\label{variable-s}
	sE= 1+ \frac{m_i m_j}{(m_i+m_j)\sum_{\ell=1}^I m_\ell} \left| u \right|^2 \ \  \text{and}\   \ (1-s)E= 1 +  \frac{m_i+m_j}{\sum_{\ell=1}^I m_\ell} \left| V \right|^2.
	\end{equation} 
\end{itemize}
	 
Finally, each of  the post-collisional quantities  $\left\langle v' \right\rangle_i^2$ and $\left\langle v'_* \right\rangle_j^2$  as written  the representation as in \eqref{kinetic energies squares}, can be recast
through the energy $E$ and the dot product between center of mass vector $V$ and the scattering direction $\sigma$ as follows 
	\begin{equation}\label{kin energy mixture}
	\begin{split}
	\left\langle v' \right\rangle_i^2&= r (1-s) E + (1-r) s E + 2 \sqrt{r (1-r) (sE -1)((1-s)E-1)} \, \sigma \cdot \hat{V},\\
	\left\langle v'_* \right\rangle_j^2&= r s E + (1-r) (1-s) E   - 2 \sqrt{r (1-r) (sE -1)((1-s)E-1)} \, \sigma \cdot \hat{V}, 
	\end{split}
	\end{equation}
which  yields the important relation that expresses the post - collisional local micro energy $E$ as  a rotation of  factors of  $E$ and $V\cdot\sigma$, while preserving the local energy itself.  
Indeed, denoting
	\begin{equation*}
	\begin{split}
	& p= r (1-s) E + (1-r) s E, \\ & q=E-p= r s E + (1-r) (1-s) E, \\ &\lambda= 2 \sqrt{r (1-r) (sE -1)((1-s)E-1)} \\& \phantom{\lambda}= 2 \sqrt{r (1-r)}\frac{\sqrt{m_i m_j}}{\sum_{\ell=1}^I m_\ell} \left| u \right| |V|,
	\end{split}
	\end{equation*}
	clearly $p+q=E$  and the representation \eqref{EI1} follows while preserving the binary micro energy relation \eqref{EI2}
	\begin{equation*}
	\left\langle v' \right\rangle_i^2= p + \lambda\, \sigma \cdot \hat{V}, \quad 
	\left\langle v'_* \right\rangle_j^2= q  - \lambda  \, \sigma \cdot \hat{V},
	\end{equation*}
which completes the proof of the energy identities \eqref{EI1} and \eqref{EI2}. 

Moreover, it follows
\begin{equation*}
\frac{1}{E} \left( p+ \lambda\right) \leq \left( \sqrt{r(1-s)} +  \sqrt{(1-r)s} \right)^2 \leq 1,
\end{equation*}
since
\begin{equation*}
\max_{\substack{ 0<r<1 \\ 0\leq s \leq 1}} \left( \sqrt{r(1-s)} +  \sqrt{(1-r)s}  \right) =1.
\end{equation*}
Similarly,
\begin{equation*}
\frac{1}{E} \left( q+ \lambda\right) \leq \left( \sqrt{rs} +  \sqrt{(1-r)(1-s)} \right)^2 \leq 1,
\end{equation*}
uniformly in any $(i,j)-$pair, which concludes the proof of Lemma.
\end{proof}

	\begin{remark}\label{Remark single vs mixture} Let us elaborate more on the difference between writing kinetic energies \eqref{kinetic energies squares} when $m_i\neq m_j$  versus $m_i=m_j$. In order to be more precise, we will put a bar on a quantity when assuming the same masses. For instance, total energy of the two colliding particles of the same masses $m_i$ is
\begin{equation*}
		\bar{E}= \left\langle v \right\rangle_i^2 + \left\langle v_* \right\rangle_i^2  = 2 + \frac{2 m_i}{\sum_{j=1}^I m_j} \left|V\right|^2 + \frac{ m_i}{2 \sum_{j=1}^I m_j} \left|u\right|^2.
		\end{equation*}	
		When $m_i=m_j$ then the parameter $r=1/2$, and consequently for $\bar{p}:=p(v, v_*, m_i, m_i)$, $\bar{q}:=q(v, v_*, m_i, m_i)$ and $\bar{\lambda}$ we have
		\begin{equation*}
\bar{p} =  \bar{q}  = \frac{1}{2} \, \bar{E}, \quad  \bar{\lambda} =  \frac{ m_i}{\sum_{j=1}^I m_j} {\left|u\right|  \left|V\right| },
		\end{equation*}	
		which gives the squares of the magnitudes of the post-collisional velocities when $m_i=m_j$,
		\begin{equation}\begin{split}\label{kin energy single}
		&\left\langle v' \right\rangle_i^2 = \bar{E} \left( \frac{1}{2} + \frac{ m_i}{\sum_{j=1}^I m_j}  \frac{ \left|u\right| \left|V\right| }{\bar{E}} \sigma \cdot \hat{V} \right), \\
		 &\left\langle v'_* \right\rangle_i^2 = \bar{E}  \left( \frac{1}{2} - \frac{ m_i}{\sum_{j=1}^I m_j}  \frac{ \left|u\right| \left|V\right| }{\bar{E}} \sigma \cdot \hat{V}\right).
			\end{split} \end{equation}
	Now, the difference between 	\eqref{kin energy mixture} as a convex combination writing in the mixture setting and \eqref{kin energy single} as its special ``middle point", or ``half", case in single component case (or mixture for $m_i=m_j$) is clear.\\
	
	Another important aspect to be pointed out is the comparison of inequalities \eqref{key point} in the case $m_i\neq m_j$ versus $m_i=m_j$. When $m_i=m_j$, simply performing Young's inequality we get
	\begin{equation}\label{s=1/2}
	\frac{\bar{\lambda}}{\bar{E}}    \leq \frac{1}{2},
	\end{equation}
	which yields
	\begin{equation*}
\frac{1}{\bar{E}} \left( \bar{p}+ \bar{\lambda}\right) = \frac{1}{\bar{E}} \left( \bar{q}+ \bar{\lambda}\right) \leq 1.
	\end{equation*}
	This inequality is an analogue of  \eqref{key point} for $m_i = m_j$. Note than when masses are the same we can make use of the Young inequality, while in the case of different masses, we had to be more precise since both $\frac{1}{{E}} \left( {p}+ {\lambda}\right)$ and $ \frac{1}{{E}} \left( {q}+ {\lambda}\right)$ attain 1 as a maximal value for some values of their arguments, and therefore there is no room for any inequality. In particular, this inequality will be of decisive importance for the success of the Povzner lemma that will guarantee decay of the gain term with respect to the number of moments.  
	\end{remark}

\medskip
  
\begin{proof}[Proof of Povzner lemma~\ref{Povzner intro}]
In order to compute the angular average estimate \eqref{Povzner estimate gain}    we use the representation \eqref{EI1} and \eqref{EI2} from the energy identity Lemma~\ref{energy-identity-mixture} raised to power $k/2$. 
Then, the left hand side integral of  \eqref{Povzner estimate gain} becomes
	\begin{multline}\label{pomocna 2}
	\int_{S^2} \left( \left\langle v' \right\rangle_i^k +  \left\langle v'_* \right\rangle_j^k  \right) b_{ij}(\sigma\cdot \hat{u}) \, \mathrm{d}\sigma \\  = \int_{S^2}  \left( \left( p + \lambda\, \sigma \cdot \hat{V} \right)^\frac{k}{2} +  \left( q  - \lambda  \, \sigma \cdot \hat{V} \right)^\frac{k}{2}  \right) b_{ij}(\sigma\cdot \hat{u}) \, \mathrm{d}\sigma. 
	\end{multline}	
Now we use polar coordinates for $\sigma$ and $\hat{V}$ with zenith $\hat{u}$. Namely, denoting with $\theta$ the angle between $\sigma$ and $\hat{u}$, we decompose $\sigma$ as
\begin{equation}\label{sigma decomposed}
\sigma =  \cos\theta\, \hat{u} +  \sin\theta\, \omega, \ \text{with} \ \hat u\cdot\omega=0 \ \text{and} \  \omega=(\cos\varphi, \sin\varphi), \ \theta\in[0,\pi), \varphi\in[0,2\pi).
\end{equation}
In the same fashion we decompose $\hat{V}$, by denoting with $\alpha \in [0,\pi)$ the angle between $\hat{V}$ and $\hat{u}$,
\begin{equation*}
\hat{V} =  \cos\alpha\, \hat{u} +  \sin\alpha\, \Phi, \ \text{where} \ \Phi\in S^1 \ \text{with} \ \hat u\cdot\Phi=0.
\end{equation*}
Then the scalar product $\sigma \cdot \hat{V}$ becomes
\begin{align*}
\sigma \cdot \hat{V} 
&=\cos\theta\cos\alpha  +  \Phi\cdot\omega \sin\theta\sin\alpha.
\end{align*}
Defining $\tau:= \cos\theta$ and expressing $\sin\theta=\sqrt{1-\tau^2}$, since $\sin\theta \geq 0$ on the range of $\theta$, this scalar product reads
\begin{equation}\label{define_mu} 
\sigma \cdot \hat{V}  = \tau\cos\alpha  +  \Phi\cdot\omega \sqrt{1-\tau^2} \sin\alpha=:\mu =\mu(\tau,\alpha, \Phi\cdot\omega).
\end{equation}
In  the integral \eqref{pomocna 2}, we first express  $\sigma$ in its polar coordinates \eqref{sigma decomposed} and then change variables $\theta\mapsto\tau=\cos\theta$,  which yields
	\begin{multline*}
\int_{S^2}  \left( \left( p + \lambda\, \sigma \cdot \hat{V} \right)^\frac{k}{2} +  \left( q  - \lambda  \, \sigma \cdot \hat{V} \right)^\frac{k}{2}  \right) b_{ij}(\sigma\cdot \hat{u})\, \mathrm{d}\sigma
\\
= \int_{0}^{2\pi} \int_{0}^{\pi}  \left( \left( p + \lambda\, \sigma \cdot \hat{V} \right)^\frac{k}{2} +  \left( q  - \lambda  \, \sigma \cdot \hat{V} \right)^\frac{k}{2}  \right) b_{ij}(\cos\theta) \, \sin\theta \mathrm{d}\theta\mathrm{d}\varphi
\\
= \int_{0}^{2\pi} \int_{-1}^{1}  \left( \left( p + \lambda\,\mu \right)^\frac{k}{2} +  \left( q  - \lambda  \, \mu\right)^\frac{k}{2}  \right) b_{ij}(\tau) \mathrm{d}\tau\mathrm{d}\varphi.
\end{multline*}
For $k\geq 4$, Taylor expansion of $ \left( p + \lambda\, \mu \right)^{k/2}$  and $\left( q  - \lambda  \, \mu \right)^{k/2} $ around $\mu=0$ up to second order yields:
	\begin{equation*}
	\begin{split}
	\left( p + \lambda\, \mu \right)^\frac{k}{2} &= p^\frac{k}{2} + \tfrac{k}{2} p^{\tfrac{k}{2}-1} \lambda \mu + \tfrac{k}{2} \left(\tfrac{k}{2}-1\right) \lambda^2 \mu^2 \int_{0}^{1} (1-z) (p+\lambda \mu z)^{\frac{k}{2}-2} \mathrm{d}z,\\
	\left( q  - \lambda  \, \mu \right)^\frac{k}{2} &= q^\frac{k}{2} - \tfrac{k}{2} q^{\frac{k}{2}-1} \lambda \mu + \tfrac{k}{2} \left(\tfrac{k}{2}-1\right) \lambda^2 \mu^2 \int_{0}^{1} (1-z) (q - \lambda \mu z)^{\frac{k}{2}-2} \mathrm{d}z.
	\end{split}
	\end{equation*}
	For $2<k<4$, we stop at the first order and proceed similarly. 
	
	Now, let us analyze the integrands. By the Young inequality, for $\lambda$ the following estimates hold
	\begin{equation}\label{lambda estimate}
	\pm \lambda\leq q -1 \leq q, \quad \text{and} \quad \pm \lambda \leq p-1 \leq p.
	\end{equation}
We recall  definition of $p$ and $q$, 
\begin{equation*}
p= \left(r(1-s)+(1-r)s\right) E, \quad q= \left(rs+(1-r)(1-s)\right)E,
\end{equation*}
for $r\in (0,1)$ and $s\in [0,1]$. Considering $r$ as parameter, for both coefficients maximum with respect to variable $s$ is achieved on the boundary, i.e. for either $s=0$ or $s=1$, and moreover the following estimate holds for both coefficients
\begin{equation*}
r(1-s)+(1-r)s \leq \max\{r,(1-r)\}, \quad rs+(1-r)(1-s) \leq \max\{r,(1-r)\}.
\end{equation*}
	Denoting
	\begin{equation}\label{rmax}
	\overline{r}=\max\{r, 1-r\},
	\end{equation}
we conclude on upper bound for both $p$ and $q$,
\begin{equation*}
p \leq \overline{r} E, \qquad q\leq \overline{r} E.
\end{equation*}
Moreover, for $p$ and $q$ it holds
\begin{equation*}
p = r (1-s)E + (1-r) sE \geq \underline{r} E % \geq r(1-r)E
\quad \text{and} \quad q \geq \underline{r} E.%r(1-r)E.
\end{equation*}
where we have 	denoted
	\begin{equation}\label{rmin}
	\underline{r}=\min\{r, 1-r\}.
	\end{equation}
Taking into account inequalities above, one has
\begin{equation*}
p + \lambda \mu z \leq p+ q \mu z= E - q(1-\mu z) \leq E (1- \underline{r}(1- |\mu| z)),
\end{equation*}
and similarly
\begin{equation*}
q - \lambda \mu z \leq   E (1- \underline{r} (1-|\mu| z)).
\end{equation*}
Therefore,
\begin{multline*}
\left( p + \lambda\, \mu \right)^\frac{k}{2} +  \left( q  - \lambda  \, \mu \right)^\frac{k}{2} \leq p^\frac{k}{2} + q^\frac{k}{2} + \tfrac{k}{2}  \mu \left( p^{\frac{k}{2}}+q^{\frac{k}{2}} \right) \\ +  k \left(\tfrac{k}{2}-1\right) \overline{r}^2 \mu^2 E^{\frac{k}{2}} \int_{0}^{1} (1-z) (1- \underline{r}(1-|\mu| z))^{\frac{k}{2}-2} \mathrm{d}z.
\end{multline*}

Then 
\begin{equation*}
\int_{0}^{2\pi} \int_{-1}^{1}  \left( \left( p + \lambda\,\mu \right)^\frac{k}{2} +  \left( q  - \lambda  \, \mu\right)^\frac{k}{2}  \right) b_{ij}(\tau) \, \mathrm{d}\tau\,\mathrm{d}\varphi \leq  P_1 + P_2 + P_3,
	\end{equation*}
	with 
	\begin{equation*}
	\begin{split}
	P_1 &:=  \left(p^\frac{k}{2} + q^\frac{k}{2}\right) \int_{0}^{2\pi} \int_{-1}^{1}    b_{ij}(\tau) \, \mathrm{d}\tau \, \mathrm{d}\varphi ,\\
	P_2 &:= \tfrac{k}{2}  \left( p^{\frac{k}{2}}+q^{\frac{k}{2}} \right) \int_{0}^{2\pi} \int_{-1}^{1}  \mu \,  b_{ij}(\tau) \, \mathrm{d}\tau \, \mathrm{d}\varphi,  \\
 P_3  &:= k\left(\tfrac{k}{2}-1\right) \overline{r}^2 E^{\frac{k}{2}}  \int_{0}^{2\pi} \int_{-1}^1 \mu^2  \\ &\qquad  \qquad \qquad \qquad \times \left(\int_{0}^{1} (1-z) (1- \underline{r}(1-|\mu| z))^{\frac{k}{2}-2}   \mathrm{d}z \right) b_{ij}(\tau)\, \mathrm{d}\tau\, \mathrm{d}\varphi.
	\end{split}
	\end{equation*}
	
\subsubsection*{Term $P_1$}
	Introducing  constant $\tilde{C}_n$
	\begin{equation}
	\tilde{C}_n = \overline{r}^n, \quad 0<\overline{r}<1,
	\end{equation}
	which clearly decays in $n$, we have
	\begin{equation*}
	P_1 = \left\| b_{ij} \right\|_{L^1(\mathrm{d}\sigma)} \left(p^\frac{k}{2} + q^\frac{k}{2} \right) \leq  \left\| b_{ij} \right\|_{L^1(\mathrm{d}\sigma)} 2 \tilde{C}_\frac{k}{2} E^\frac{k}{2}.
	\end{equation*}
\subsubsection*{Term $P_2$} 	
Taking into account definition of $\mu$  from \eqref{define_mu}, the parity arguments yield 
\begin{equation*}
P_2 = \tfrac{k}{2} \left(p^\frac{k}{2} + q^\frac{k}{2} \right) \int_0^{2\pi} \int_{-1}^1 \tau \cos \alpha  \, b_{ij}(\tau) \, \mathrm{d} \tau \, \mathrm{d}\varphi,
\end{equation*}	
after bounding $\cos\alpha \leq 1$. 
Using the estimate above for $P_1$ and the fact that $\tau \cos \alpha\leq1$, we finally obtain
\begin{equation*}
P_2 \leq    \left\| b_{ij} \right\|_{L^1(\mathrm{d}\sigma)}  k \, \tilde{C}_\frac{k}{2} E^\frac{k}{2}.
\end{equation*}	
Since the constant $ \tilde{C}_\frac{k}{2}$ has power decay in $k$, the constant $ k \, \tilde{C}_\frac{k}{2} $ also decreases in $k$.
\subsubsection*{Term $P_3$} 
We can compute explicitly the integral with respect to $z$
\begin{multline*}
 \int_{0}^{1} (1-z) (1- a (1-A z))^{n-2} \mathrm{d}z \\= \frac{1}{a^2 A^2} \frac{1}{n(n-1)} \left(  \left(1+ a(A-1)\right)^n -  (1-a)^n - a A(1-a)^{n-1}  n\right),
\end{multline*}
for any $0<a<1$ and $A>0$. If $A=0$, then we easily obtain 
\begin{equation*}
\int_{0}^{1} (1-z) (1- a )^{n-2} \mathrm{d}z = \frac{1}{2}  (1-a)^{n-2}.
\end{equation*}
In our case  $a=\underline{r}$ and $A=|\mu|$, $\mu$ being a function of variables of integration $\tau$ and $\varphi$ defined in \eqref{define_mu} that satisfies $|\mu|\leq 1$, and thus $P_3$ becomes
\begin{multline*}
P_3  =2\frac{\overline{r}^2}{ \underline{r}^2} E^{\frac{k}{2}}  \int_{0}^{2\pi} \int_{-1}^1    \left(  \left(1+ \underline{r}(|\mu|-1)\right)^{\frac{k}{2}} \right. \\ \left. -  (1-\underline{r})^{\frac{k}{2}} - \underline{r} |\mu|(1-\underline{r})^{\frac{k}{2}-1}  \tfrac{k}{2}\right)  b_{ij}(\tau) \, \mathrm{d}\tau \, \mathrm{d}\varphi\\
=: P_{3_1} + P_{3_2} + P_{3_3},
\\ \leq E^{\frac{k}{2}} \left( \check{C}^{b_{ij}}_{\frac{k}{2}} + \left\| b_{ij} \right\|_{L^1(\mathrm{d}\sigma)} \left( \bar{C}_{\frac{k}{2}} + \hat{C}_{\frac{k}{2}} \right) \right),
\end{multline*}
where we have denoted 
	\begin{equation*}
\begin{split}
P_{3_1} &:= 2\frac{\overline{r}^2}{ \underline{r}^2} E^{\frac{k}{2}}  \int_{0}^{2\pi} \int_{-1}^1      \left(1+ \underline{r}(|\mu|-1)\right)^{\frac{k}{2}} b_{ij}(\tau) \, \mathrm{d}\tau \, \mathrm{d}\varphi,\\
P_{3_2} &:= - 2\frac{\overline{r}^2}{ \underline{r}^2} (1-\underline{r})^{\frac{k}{2}} E^{\frac{k}{2}}  \int_{0}^{2\pi} \int_{-1}^1     b_{ij}(\tau) \, \mathrm{d}\tau \, \mathrm{d}\varphi,\\
P_{3_3}  &:= - \frac{\overline{r}^2}{ \underline{r}} \, k\, (1-\underline{r})^{\frac{k}{2}-1} E^{\frac{k}{2}}  \int_{0}^{2\pi} \int_{-1}^1     |\mu| \, b_{ij}(\tau) \, \mathrm{d}\tau \, \mathrm{d}\varphi.\\
\end{split}
\end{equation*}
\subsubsection*{Term $P_{3_1}$} We rewrite term $P_{3_1}$, 
\begin{equation*}
P_{3_1} = \check{C}^{b_{ij}}_{\frac{k}{2}} E^{\frac{k}{2}},
\end{equation*}
by introducing the constant $\check{C}^{b_{ij}}_n$
\begin{equation}\label{Povzner C check}
\check{C}^{b_{ij}}_n = 2\frac{\overline{r}^2}{ \underline{r}^2}  \int_{0}^{2\pi} \int_{-1}^1   \left(1+ \underline{r}(|\mu|-1)\right)^{n} b_{ij}(\tau) \, \mathrm{d}\tau \, \mathrm{d}\varphi.
\end{equation}
In order to study its properties,  we first note that $1+\underline{r}(|\mu|-1) \leq 1$, since $|\mu|\leq 1$, and the equality holds only when $|\mu|=1$ (or $\sigma = \left\{ \pm \hat{V} \right\}$). Therefore, the sequence of functions 
\begin{equation*}
A_n(x):=  \left(1+ \underline{r}(x-1)\right)^{n}
\end{equation*}
decreases monotonically in $n$ and tends to 0 as $n\rightarrow \infty$  for every $x\in (0,1)$ up to a set of measure zero. Finally, we conclude by monotone convergence Theorem that $ \check{C}^{b_{ij}}_{\frac k2} \searrow 0$ as $k\rightarrow \infty$.

When  $b_{ij} \in  L^\infty(S^2; \mathrm{d}\sigma)$, we can obtain the explicit decay rate of the constant $ \check{C}^{b_{ij}}_{\frac k2} $, since in this case the integral \eqref{Povzner C check} significantly simplifies. The rate will be calculated in the Remark \ref{Remark Povzner constant Linf} below.

\subsubsection*{Term $P_{3_2}$} For the term $P_{3_2}$ we immediately obtain 
\begin{equation*}
P_{3_2} =  \left\| b_{ij} \right\|_{L^1(\mathrm{d}\sigma)} \bar{C}_{\frac{k}{2}} E^{\frac{k}{2}},
\end{equation*}
with the constant
\begin{equation*}
\bar{C}_{n} = - 2 \frac{\overline{r}^2}{ \underline{r}^2} (1-\underline{r})^n.
\end{equation*}
\subsubsection*{Term $P_{3_3}$} We first estimate the term $P_{3_3}$ using $|\mu|\leq 1$,
\begin{equation*}
P_{3_3} \leq   \frac{\overline{r}^2}{ \underline{r}} \, k\, (1-\underline{r})^{\frac{k}{2}-1} E^{\frac{k}{2}}  \int_{0}^{2\pi} \int_{-1}^1    b_{ij}(\tau) \, \mathrm{d}\tau \, \mathrm{d}\varphi.
\\ \leq  \left\| b_{ij} \right\|_{L^1(\mathrm{d}\sigma)} \hat{C}_{\frac{k}{2}}   E^{\frac{k}{2}},
\end{equation*}
and  the constant is defined with
\begin{equation}\label{Povzner C hat}
\hat{C}_n = 2 \frac{\overline{r}^2}{ \underline{r}} \, n\, (1-\underline{r})^{n-1}.
\end{equation}
Gathering estimates for $P_1$, $P_2$ and $P_3$ completes the proof of \eqref{Povzner estimate gain} with
	$$\boldsymbol{\mathcal{C}}^{ij}_{n} =  \left\| b_{ij} \right\|_{L^1(\mathrm{d}\sigma)} \left( (2n+2) \tilde{C}_n +   \bar{C}_n +   \hat{C}_n \right) + \check{C}^{b_{ij}}_n, n>2,$$
and $\boldsymbol{\mathcal{C}}^{ij}_{n} = \left\| b_{ij} \right\|_{L^1(\mathrm{d}\sigma)} 2  \tilde{C}_n$, if $1<n\leq2$. Thus, the constant $\boldsymbol{\mathcal{C}}^{ij}_{\frac k2} $ issuing from Povzner lemma  satisfies  $\boldsymbol{\mathcal{C}}^{ij}_{\frac k2} \rightarrow 0$, as $k\rightarrow \infty$, and so there exists $k^{ij}_*=k^{ij}_*(r_{ij}, b_{ij})$ for which $\boldsymbol{\mathcal{C}}^{ij}_{\frac k2} <\left\| b_{ij} \right\|_{L^1(\mathrm{d}\sigma)}$, for $k>k^{ij}_*$.	
\end{proof}
\begin{remark}[The case $b_{ij} \in  L^\infty(S^2; \mathrm{d}\sigma)$]\label{Remark Povzner constant Linf} When the angular kernel is assumed bounded, some calculations are simpler. Pulling out the $L^\infty$ norm of $b_{ij}$, we have
\begin{equation*}
 \int_{0}^{2\pi} \int_{-1}^{1}    b_{ij}(\tau) \, \mathrm{d}\tau \, \mathrm{d}\varphi \leq 4\pi  \left\| b_{ij} \right\|_{L^\infty(\mathrm{d}\sigma)},
\end{equation*}
and so terms $P_1$ and $P_{3_2}$ become
\begin{equation*}
P_1 \leq  8\pi  \left\| b_{ij} \right\|_{L^\infty(\mathrm{d}\sigma)} \tilde{C}_\frac{k}{2} E^\frac{k}{2}, \qquad P_{3_2} =  4\pi  \left\| b_{ij} \right\|_{L^\infty(\mathrm{d}\sigma)} \bar{C}_{\frac{k}{2}} E^{\frac{k}{2}}.
\end{equation*}

Moreover,  when $b_{ij}$ is assumed bounded, the  starting  integral \eqref{pomocna 2} do not depend  $\sigma \cdot \hat{u}$ anymore, so we may instead of $\hat{u}$ take $\hat{V}$ as a zenith of polar coordinates in \eqref{sigma decomposed}, which amounts to take $\alpha=0$ in \eqref{define_mu} that implies $\mu=\tau$.  In this case, thanks to the parity arguments, term $P_2$ vanishes, and term $P_{3_3}$ can be explicitly calculated, without using any estimate,
\begin{equation*}
P_{3_3}  = - \frac{\overline{r}^2}{ \underline{r}} \, k\, (1-\underline{r})^{\frac{k}{2}-1} E^{\frac{k}{2}}  \int_{0}^{2\pi} \int_{-1}^1     |\tau| \, b_{ij}(\tau) \, \mathrm{d}\tau \, \mathrm{d}\varphi\\
=  - 2 \pi\,  \left\| b_{ij} \right\|_{L^\infty(\mathrm{d}\sigma)} \hat{C}_\frac{k}{2} E^{\frac{k}{2}},
\end{equation*}
with the constant $ \hat{C}_\frac{k}{2}$ from \eqref{Povzner C hat}.

Finally, let us compute explicitly  the constant $ \check{C}^{b_{ij}}_{n}$ from \eqref{Povzner C check}  when $b_{ij}(\sigma \cdot \hat{u}) \in  L^\infty(S^2; \mathrm{d}\sigma)$. Namely, pulling out the $L^\infty$ norm of $b_{ij}$ from the integral and using $\mu=\tau$, we get
\begin{multline*}
\check{C}^{b_{ij}}_n = 2\frac{\overline{r}^2}{ \underline{r}^2}   \left\| b_{ij} \right\|_{L^\infty(\mathrm{d}\sigma)} \int_{0}^{2\pi} \int_{-1}^1   \left(1+ \underline{r}(|\tau|-1)\right)^{n}  \mathrm{d}\tau \, \mathrm{d}\varphi
\\
=8 \pi \, \frac{\overline{r}^2}{ \underline{r}^3}   \left\| b_{ij} \right\|_{L^\infty(\mathrm{d}\sigma)} \left( \frac{1}{n+1} - \frac{(1-\underline{r})^{n+1}}{n+1} \right),
\end{multline*}
that shows its decay rate. 

To summarize, the constant from the Povzner lemma  in the case of bounded angular part reads
\begin{equation}\label{c povzner}
\boldsymbol{\mathcal{C}}^{ij}_n = 4 \pi  \left\| b_{ij} \right\|_{L^\infty(\mathrm{d}\sigma)} \boldsymbol{\mathcal{C}}^{\infty}_n(r),
\end{equation}
where we have denoted 
\begin{equation}\label{c povnzer renorm}
\boldsymbol{\mathcal{C}}^{\infty}_n(r) =  2 \tilde{C}_n + \bar{C}_n - \frac{1}{2} \hat{C}_n + 2 \frac{\overline{r}^2}{ \underline{r}^3}   \left( \frac{1}{n+1} - \frac{(1-\underline{r})^{n+1}}{n+1} \right), \quad n>2,
\end{equation}
 and $\boldsymbol{\mathcal{C}}^{\infty}_n(r) =  2 \tilde{C}_n $ if $1< n \leq 2$,
recalling \eqref{rmin} and \eqref{rmax}. Moreover, it satisfies  $\boldsymbol{\mathcal{C}}^{ij}_{\frac k2} < 4 \pi  \left\| b_{ij} \right\|_{L^\infty(\mathrm{d}\sigma)} $, or equivalently $\boldsymbol{\mathcal{C}}^{\infty}_n(r) <1$, for sufficiently large $k^{ij}_*$ depending on $r_{ij}$ and $b_{ij}$.
\end{remark}

\subsection{Study of the Povzner constant for $b_{ij}(\sigma \cdot \hat{u}) \in  L^\infty(S^2; \mathrm{d}\sigma)$}	In this paragraph we study in detail the constant \eqref{c povzner} from the Povzner lemma \ref{Povzner intro} in the case of bounded angular part. More precisely, we study its normalized part \eqref{c povnzer renorm} 
\begin{equation}\label{c povnzer renorm 2}
\boldsymbol{\mathcal{C}}^{\infty}_n(r) =  2 \overline{r}^n  - 2 \frac{\overline{r}^2}{ \underline{r}^2} (1-\underline{r})^n - \frac{\overline{r}^2}{ \underline{r}} \, n\, (1-\underline{r})^{n-1} + 2 \frac{\overline{r}^2}{ \underline{r}^3}   \left( \frac{1}{n+1} - \frac{(1-\underline{r})^{n+1}}{n+1} \right),
\end{equation}
for $n>2$ and $\boldsymbol{\mathcal{C}}^{\infty}_n(r) =  2 \overline{r}^n$ if $1< n \leq2$, with $\overline{r} = \max\left\{r, 1-r \right\}$ and $\underline{r} = \min\left\{r,1-r\right\}$, 
and elaborate more on its decay rate in $n$ depending on $r$. 

First, taking $r=\frac{1}{2}$ we expect to recover the same  properties as for the single gas when decay rate of the Povzner constant  \cite{GambaAlonso18} was $\frac{2}{n+1}$, that monotonically decreases and tends to zero in $n>1$. In our case,
\begin{equation*}
\boldsymbol{\mathcal{C}}^{\infty}_n\left(\frac{1}{2}\right) =
\begin{cases}
 \frac{4}{n+1} - \left(\frac{1}{2}\right)^n \left( n+ \frac{2}{n+1} \right), \text{if}  \ n>2,\\
2 \left(\frac{1}{2}\right)^n, \text{if}  \ 1<n\leq2.
\end{cases}
\end{equation*}
that keeps the same properties as for the single gas, which can be illustrated as in   Figure \ref{figure-povzner-single}.

For general $r\in(0,1)$ decay properties of the constant issuing from the Povzner lemma \eqref{c povnzer renorm} strongly depend on $r$ or on the fact how much species masses $m_i$, $i=1,\dots,I$ are disparate. It is clear that, since $0<r<1$, the constant $\boldsymbol{\mathcal{C}}^{\infty}_n(r)$ will tend to zero as $n$ goes to infinity. Here we are interested in  a more subtle question: determine  $n_*$  such that it holds $\boldsymbol{\mathcal{C}}^{\infty}_n(r)<1$ for $n\geq n_*$ and any fixed $0<r<1$. Converge of $\boldsymbol{\mathcal{C}}^{\infty}_n(r)$ in $n$ towards zero for any $0<r<1$ ensures existence of such $n_*$. It can be observed that  $n_*$ grows as much as $r$ is deviated from $\frac{1}{2}$, since the constants in $\boldsymbol{\mathcal{C}}^{\infty}_n(r)$ with power decay rate will decay more  slowly as $r$ deviates    from $\frac{1}{2}$. This behavior is illustrated in Figure \ref{figure-povzner-mixture}. We can reformulate the question: for some fixed value of $n$ determine the interval of $r$ for which it holds $\boldsymbol{\mathcal{C}}^{\infty}_n(r)<1$, that is illustrated in Figure \ref{figure-povzner-mixture-2}. 

\begin{figure}
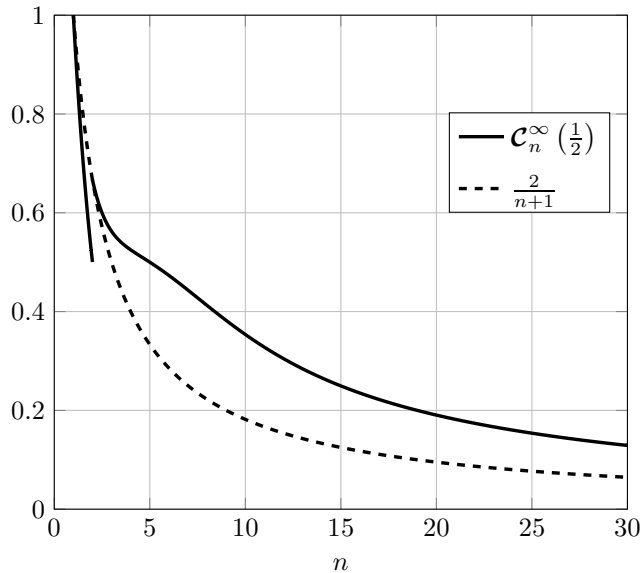

	\begin{center}
	
	\end{center}
	\caption{Comparison of the Povzner constant for $r=\frac{1}{2}$ in the mixture setting and the single component gas for $n>1$.}
	\label{figure-povzner-single}
\end{figure}

\begin{figure}
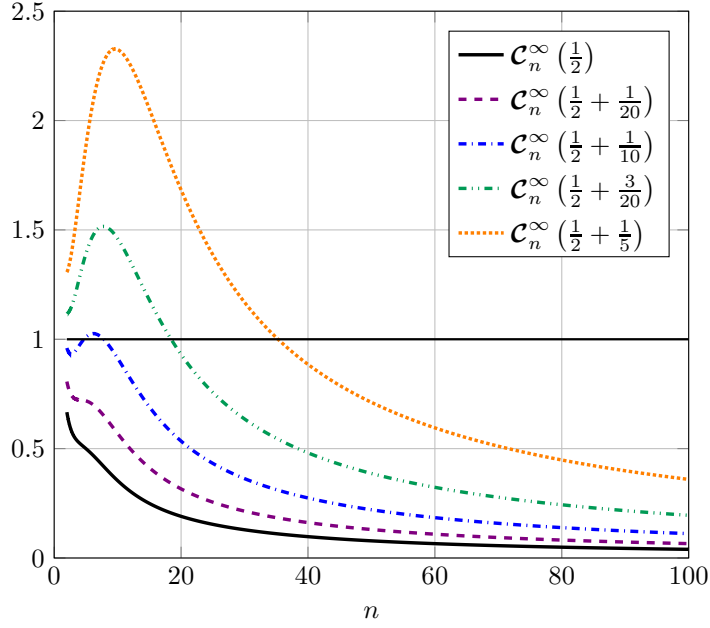

	\begin{center}
		%
	
	\end{center}
\caption{Constant $\boldsymbol{\mathcal{C}}^{\infty}_n(r)$ from Povzner lemma \ref{Povzner intro} for some fixed value of $r=:r_*$. This figure illustrates the non-monotonic behavior in $n$ variable, and the growth of $n$ needed to ensure   that $\boldsymbol{\mathcal{C}}^{\infty}_{n}(r_*)<1$ caused by a deviation  of $r$  with respect to $\frac{1}{2}$.}
	\label{figure-povzner-mixture}
\end{figure}

\begin{figure}
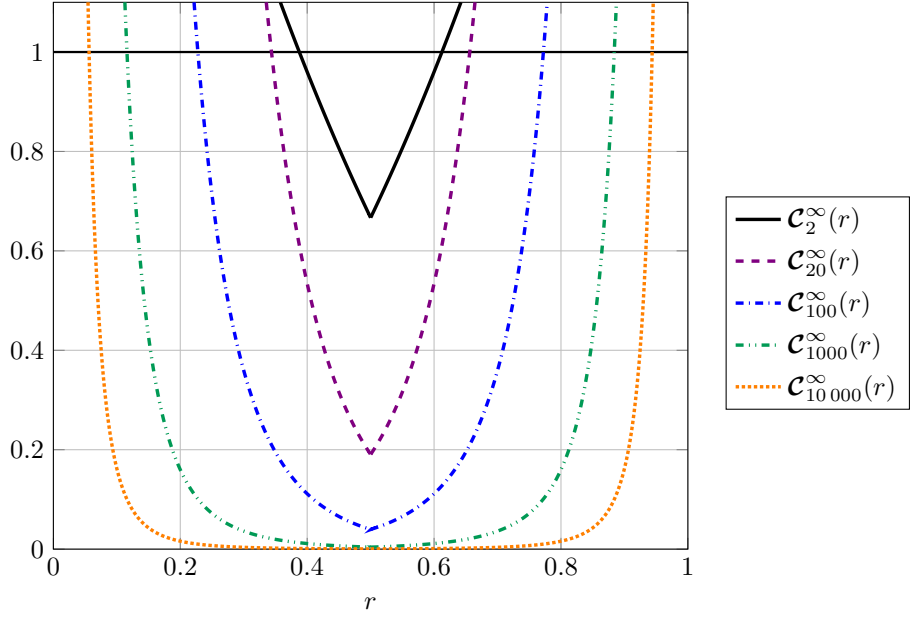

	\begin{center}
		%
	
	\end{center}
	\caption{Constant $\boldsymbol{\mathcal{C}}^{\infty}_n(r)$ from Povzner lemma \ref{Povzner intro} for some fixed value of $n=:n_*$. This figure illustrates the interval of $r$ for  which it holds $\boldsymbol{\mathcal{C}}^{\infty}_{n_*}(r)<1$.}
	\label{figure-povzner-mixture-2}
\end{figure}

\section{Proof of Existence and Uniqueness Theorem \ref{theorem existence uniqueness}}\label{Section Ex Uni proof}
Before proving Theorem \ref{theorem existence uniqueness}, we first study a property of the collision operator that is a consequence of the Povzner lemma \ref{povzner constant prop} and  lemma \ref{lemma lower bound}.

\begin{lemma}\label{lemma Q za tangent cond}
	Let $\mathbb{F}=\left[f_i\right]_{i=1,\dots,I} \in \Omega$   and $k_*$ as defined in \eqref{kstar}. Then, the following estimate holds
	\begin{equation}\label{Q za tangent cond}
	\sum_{i=1}^I \int_{ \mathbb{R}^3} \left[ \mathbb{Q}(\mathbb{F}) \right]_i \left\langle v \right\rangle_i^{k_*} \mathrm{d} v
	\leq - A_{k_*} \, \mathfrak{m}_{k_*}[\mathbb{F}](t)^{1+ \frac{\overline{\gamma}}{k_*}}+  B_{k_*} \,   \mathfrak{m}_{k_*}[\mathbb{F}](t),
	\end{equation}
with positive constants 
	\begin{equation}\label{const Ak Bk}
	\begin{split}
	A_{k_*} &= \min_{1\leq i, j \leq I} \left(\left\| b_{ij} \right\|_{L^1(\mathrm{d}\sigma)} - \boldsymbol{\mathcal{C}}^{ij}_{\frac{k_*}{2}} \right)  \frac{c_{lb}}{\max_{1\leq i \leq I} m_i}  \left(I C_0 \right)^{-\frac{\overline{\gamma}}{k_*}}, \\ 
	B_{k_*}&= 2 \, C_{2}  \max_{1\leq i, j \leq I} \left(  \left( \frac{\sum_{i=1}^I m_i}{\sqrt{m_i m_j}} \right)^{\gamma_{ij}} \boldsymbol{\mathcal{C}}^{ij}_{\frac{k_*}{2}} \right)    \sum_{\ell=1}^{\lfloor\frac{k_* + 1}{2}\rfloor} \left( \begin{matrix}
		k_* \\ \ell 
		\end{matrix} \right),
	\end{split}
	\end{equation}
where $C_0$ and $C_{2}$ are from the characterization of the set $\Omega$, $c_{lb}$ is from the lower bound \eqref{lower bound}, and  $\boldsymbol{\mathcal{C}}^{ij}_{\frac{k_*}{2}}$ is a constant from the Povzner lemma \ref{Povzner intro} with   $k_* > \overline{k}$, as defined in \eqref{kstar},   ensuring the property \eqref{povzner constant prop} for any pair $(i,j)$ that yields positivity of the constant $A_{k_*}$.
\end{lemma}

\begin{remark}\label{coercive_constant} 
It is important to notice that the strict positivity of the constant $A_{k_*}$ can be view as a {\bf coercive condition} that secures global in time solutions, without the need to require boundedness of entropy. 
\end{remark}

\begin{proof}
We start with the weak form \eqref{weak form any test}. Taking test function $\psi_i(x)=\left\langle v \right\rangle_i^{k_*}$, and  cross section \eqref{cross section}, we have
\begin{multline} \label{pomocna 28} 
\sum_{i=1}^I \int_{ \mathbb{R}^3} \left[ \mathbb{Q}(\mathbb{F}) \right]_i \left\langle v \right\rangle_i^{k_*} \mathrm{d} v =\sum_{i=1}^I \sum_{{j=1 }}^I \int_{\mathbb{R}^3} \left\langle v \right\rangle_i^{k_*}  Q_{ij}(f_i,f_j) \, \mathrm{d}v\\ = \frac{1}{2}\sum_{i=1}^I \sum_{{j=1}}^I \iiint_{\mathbb{R}^3 \times \mathbb{R}^3 \times S^2} \left|v-v_*\right|^{\gamma_{ij}} f_i(v) f_j(v_*) \\ \times \left( \left\langle v' \right\rangle_i^{k_*} +  \left\langle v'_* \right\rangle_j^{k_*} -  \left\langle v \right\rangle_i^{k_*} - \left\langle v_* \right\rangle_j^{k_*}   \right) b_{ij}(\sigma \cdot \hat{u}) \,\mathrm{d}\sigma \mathrm{d}v_* \mathrm{d}v,
\end{multline}
where collisional rules are  \eqref{collisional rules bi}. The primed quantities integrated over sphere $S^2$ are estimated via Povzner lemma. Indeed, by Lemma \ref{Povzner intro} , \eqref{pomocna 28} becomes
\begin{multline}\label{poly moment start}
\sum_{i=1}^I \int_{ \mathbb{R}^3} \left[ \mathbb{Q}(\mathbb{F}) \right]_i \left\langle v \right\rangle_i^{k_*} \mathrm{d} v  \leq \frac{1}{2} \sum_{i=1}^I \sum_{j=1}^I \iint_{\mathbb{R}^3 \times \mathbb{R}^3} f_i (v) f_j(v_*) \left|v-v_*\right|^{\gamma_{ij}} \\ \times \left( \boldsymbol{\mathcal{C}}^{ij}_{\frac{k_*}{2}} \left(\left\langle v \right\rangle_i^2 + \left\langle v_* \right\rangle_j^2 \right)^{\frac{k_*}{2}} - \left\| b_{ij} \right\|_{L^1(\mathrm{d}\sigma)} \left( \left\langle v \right\rangle_i^{k_*} + \left\langle v_* \right\rangle_j^{k_*} \right) \right) \mathrm{d} v_* \mathrm{d}v,
\end{multline}
where $\boldsymbol{\mathcal{C}}^{ij}_{\frac{k_*}{2}}$ is a constant from Povzner lemma \ref{Povzner intro}  with $k_* \geq \overline{k}=\max_{1\le i,j\le I}\{ k^{ij}_*\} $ chosen large enough to ensure \eqref{povzner constant prop} uniformly in  $i,j$-pairs.  On one hand,   
we use  polynomial inequalities from Lemmas  \ref{binomial} and \ref{moment products}
\begin{equation*}
\begin{split}
\left(\left\langle v \right\rangle_i^2 + \left\langle v_* \right\rangle_j^2 \right)^{\frac{k_*}{2}} & \leq \left(\left\langle v \right\rangle_i + \left\langle v_* \right\rangle_j\right)^{k_*} \\& \leq \left\langle v \right\rangle_i^{k_*} + \left\langle v_* \right\rangle_j^{k_*} + \sum_{\ell=1}^{\ell_{k_*}} \left( \begin{matrix}
k_* \\ \ell 
\end{matrix} \right) \left( \left\langle v \right\rangle_i^{\ell} \left\langle v_* \right\rangle_j^{k_*- \ell}+ \left\langle v \right\rangle_i^{k_*- \ell} \left\langle v_* \right\rangle_j^{\ell} \right),\\
& \leq \left\langle v \right\rangle_i^{k_*} + \left\langle v_* \right\rangle_j^{k_*} + \left( \left\langle v \right\rangle_i \left\langle v_* \right\rangle_j^{k_*- 1}+ \left\langle v \right\rangle_i^{k_*- 1} \left\langle v_* \right\rangle_j \right) \left( \sum_{\ell=1}^{\ell_{k_*}} \left( \begin{matrix}
k_* \\ \ell 
\end{matrix} \right) \right) 
\end{split}
\end{equation*}
with $\ell_{k_*} = \lfloor\frac{k_* + 1}{2}\rfloor$, 
and therefore
\begin{multline}\label{pomocna 30}
\sum_{i=1}^I \int_{ \mathbb{R}^3} \left[ \mathbb{Q}(\mathbb{F}) \right]_i \left\langle v \right\rangle_i^{k_*} \mathrm{d} v 
	\leq  \frac{1}{2} \sum_{i=1}^I  \sum_{j=1}^I  \iint_{\mathbb{R}^3 \times \mathbb{R}^3} f_i (v) f_j(v_*)   \left|v-v_*\right|^{\gamma_{ij}}  \\  \times \left\{ - \left(\left\| b_{ij} \right\|_{L^1(\mathrm{d}\sigma)} - \boldsymbol{\mathcal{C}}^{ij}_{\frac{k_*}{2}} \right) \left(\left\langle v \right\rangle_i^{k_*} + \left\langle v_* \right\rangle_j^{k_*}\right)  \right. \\ 	 \left. +  \boldsymbol{\mathcal{C}}^{ij}_{\frac{k_*}{2}} \left(    \sum_{\ell=1}^{\ell_{k_*}} \left( \begin{matrix}
		k_* \\ \ell 
		\end{matrix} \right)   \right)  \left(    \left\langle v \right\rangle_i \left\langle v_* \right\rangle_j^{k_*- 1}+ \left\langle v \right\rangle_i^{k_*- 1} \left\langle v_* \right\rangle_j   \right)  \right\} \mathrm{d}v_* \mathrm{d}v.
\end{multline}
On the other hand we use upper and lower bound of the non-angular cross section $\left|v-v_*\right|^{\gamma_{ij}}$. For the upper bound, from \eqref{estimate on u^g} it follows
\begin{equation*}
\left|v-v_*\right|^{\gamma_{ij}} \leq  \left( \frac{\sum_{i=1}^I m_i}{\sqrt{m_i m_j}} \right)^{\gamma_{ij}} \left(\left\langle v \right\rangle_i^{\gamma_{ij}} + \left\langle v_* \right\rangle_j^{\gamma_{ij}}\right)  \leq \left( \frac{\sum_{i=1}^I m_i}{\sqrt{m_i m_j}} \right)^{\gamma_{ij}}  \left(\left\langle v \right\rangle_i^{\bar{\bar{\gamma}}} + \left\langle v_* \right\rangle_j^{\bar{\bar{\gamma}}}\right),
\end{equation*}
for $\bar{\bar{\gamma}}=\max_{1\leq i, j \leq I}\gamma_{ij}\in(0,1]$.  For the   lower bound, we use Lemma \ref{lemma lower bound},  but
 we first  check that all assumptions are satisfied from the fact that $\mathbb{F} \in \Omega$. Indeed, bounds on $\mathfrak{m}_{0}$ implies
\begin{equation*}
c_{0} \min_{1\leq i \leq I} m_i \leq \sum_{i=1}^I \int_{ \mathbb{R}^3} m_i \, f_i \, \mathrm{d}v \leq C_0 \max_{1\leq i \leq I} m_i.
\end{equation*}
From the other side, bounds on $\mathfrak{m}_2$ yield
\begin{equation*}
\left(c_{2} - C_0 \right) \sum_{j=1}^I m_j \leq \sum_{i=1}^I \int_{ \mathbb{R}^3} m_i \left|v\right|^2  f_i \, \mathrm{d}v \leq	\left(C_{2} - c_{0} \right) \sum_{j=1}^I m_j.
\end{equation*}
Therefore, for constants $c$ and $C$ from assumptions of Lemma \ref{lemma lower bound} we can choose
\begin{equation*}
\begin{split}
c&:=\min\left\{ 	c_{0} \min_{1\leq i \leq I} m_i, 	\left(c_{2} - C_0 \right) \sum_{j=1}^I m_j   \right\}, \\
C&:= \max\left\{ C_0 \max_{1\leq i \leq I} m_i,  \left(C_{2} - c_{0} \right) \sum_{j=1}^I m_j \right\}.
\end{split}
\end{equation*}
Note that positivity of $c$ is guaranteed by the definition of the set $\Omega$.  Finally, since it can be estimated
\begin{equation*}
\sum_{i=1}^I \int_{ \mathbb{R}^3} m_i \left|v\right|^{2+\epsilon}  f_i \, \mathrm{d}v \leq \mathfrak{m}_{2+\varepsilon} \left( \sum_{j=1}^I m_j \right)^{1+\frac{\varepsilon}{2}} \max_{1\leq i \leq I} m_i^{-\frac{\varepsilon}{2}},
\end{equation*}
we can choose
\begin{equation*}
B:= C_{2+\varepsilon} \left( \sum_{j=1}^I m_j \right)^{1+\frac{\varepsilon}{2}} \max_{1\leq i \leq I} m_i^{-\frac{\varepsilon}{2}}.
\end{equation*}
Then \eqref{lower bound} implies
\begin{equation*}
\sum_{i=1}^I \int_{  \mathbb{R}^3} f_i(v)  \left|v-v_*\right|^{\gamma_{ij}} \mathrm{d}v \geq \frac{1}{\max_{1\leq i \leq I} m_i} c_{lb} \left\langle v_* \right\rangle_j^{\overline{\gamma}}, 
\end{equation*}
and 
\begin{equation*}
\sum_{j=1}^I \int_{  \mathbb{R}^3} f_j(v_*)  \left|v-v_*\right|^{\gamma_{ij}} \mathrm{d}v \geq \frac{1}{\max_{1\leq j \leq I} m_j} c_{lb} \left\langle v \right\rangle_i^{\overline{\gamma}}.
\end{equation*}
With these estimates, \eqref{pomocna 30} becomes
\begin{equation*}
\sum_{i=1}^I \int_{ \mathbb{R}^3} \left[ \mathbb{Q}(\mathbb{F}) \right]_i \left\langle v \right\rangle_i^{k_*} \mathrm{d} v    \leq  - D_{k_*} \mathfrak{m}_{k_*+\overline{\gamma}} +  E_{k_*} \left( \mathfrak{m}_{1+\bar{\bar{\gamma}}} \, \mathfrak{m}_{k_*-1}  +  \mathfrak{m}_{k_*-1+\bar{\bar{\gamma}}} \, \mathfrak{m}_{1}   \right),
\end{equation*}
where $D_{k_*}$ and $E_{k_*}$ are positive constants
\begin{equation*}
\begin{split}
D_{k_*}&= \min_{1\leq i, j \leq I} \left( \left\| b_{ij} \right\|_{L^1(\mathrm{d}\sigma)}  - \boldsymbol{\mathcal{C}}^{ij}_{\frac{k_*}{2}} \right)  \frac{c_{lb}}{\max_{1\leq i \leq I} m_i}  , \\ 
E_{k_*}&=  \max_{1\leq i, j \leq I} \left( \left( \frac{\sum_{i=1}^I m_i}{\sqrt{m_i m_j}} \right)^{\gamma_{ij}} \boldsymbol{\mathcal{C}}^{ij}_{\frac{k_*}{2}} \right)    \sum_{\ell=1}^{\ell_{k_*}} \left( \begin{matrix}
	k_* \\ \ell 
	\end{matrix} \right).
\end{split}
\end{equation*}
In particular, $D_{k_*}$ is positive since, by assumption,  $k_*\ge \overline{k}$ defined in \eqref{kstar} large enough ensuring \eqref{povzner constant prop} for the constant $\boldsymbol{\mathcal{C}}^{ij}_{\frac{k_*}{2}}$ from Povzner lemma \eqref{Povzner estimate gain}.

Arriving in moment notation, we can use monotonicity of moments \eqref{monotonicity of norm},  together with an estimate on $\mathfrak{m}_2$ from characterization of set $\Omega$, to get the following estimate
\begin{equation*}
\sum_{i=1}^I \int_{ \mathbb{R}^3} \left[ \mathbb{Q}(\mathbb{F}) \right]_i \left\langle v \right\rangle_i^{k_*} \mathrm{d} v  \leq  - D_{k_*} \, \mathfrak{m}_{k_*+\overline{\gamma}} + 2 E_{k_*}   C_{2} \, \mathfrak{m}_{k_*}.
\end{equation*}
It remains to use a control from below  derived in \eqref{jensen} for the highest order moment $\mathfrak{m}_{k_*+\overline{\gamma}}$, taking $k=k_*$, $\lambda=\overline{\gamma}$ and $C_{\mathfrak{m}_{0}}=C_0$ there,
$$
\mathfrak{m}_{k_*+\overline{\gamma}} \geq  \left(I C_0 \right)^{-\frac{\overline{\gamma}}{k_*}}  \mathfrak{m}_{k_*}^{1+\frac{\overline{\gamma}}{k_*}},
$$
which yields final estimate \eqref{Q za tangent cond}.
\end{proof}
We turn to the proof of Existence and Uniqueness Theorem \ref{theorem existence uniqueness}. Our proof follows the one given in \cite{GambaAlonso18} for the single Boltzmann equation. In particular, our aim is to apply Theorem \ref{Theorem general} from a general ODE theory in Banach spaces. In order to do so, we first show that the collision operator  is a mapping $\mathbb{Q}: \Omega \rightarrow L_2^1$. Indeed, take any $\mathbb{F} \in \Omega$. Then,
	\begin{equation}\label{pomocna 32}
	\left\| \mathbb{Q}(\mathbb{F}) \right\|_{L_2^1} = \sum_{i=1}^I \int_{ \mathbb{R}^3} \left| \left[\mathbb{Q}(\mathbb{F}) \right]_{i}(v) \right| \left\langle v \right\rangle_i^2 \mathrm{d}v  \\ \leq \sum_{i=1}^I  \sum_{j=1}^I \int_{ \mathbb{R}^3} \left| Q_{ij}(f_i,f_j)(v) \right| \left\langle v \right\rangle_i^2 \mathrm{d}v.
	\end{equation}
	The absolute value  $\left| Q_{ij}(f_i,f_j)(v) \right|$ is written with the help of sign function and shorter notation
	$$
\left| Q_{ij}(f_i,f_j)(v) \right| = Q_{ij}(f_i,f_j)(v)  \, s_{ij}(v), \quad s_{ij}(v):=	\text{sign}\left( Q_{ij}(f_i,f_j)(v)  \right).
	$$
Then $s_{ij}(v) \left\langle v \right\rangle_i^2$ in \eqref{pomocna 32} are viewed as test functions, so the weak form \eqref{weak form any test} implies  	
\begin{multline*}
	\left\| \mathbb{Q}(\mathbb{F}) \right\|_{L_2^1} \leq \frac{1}{2} \sum_{i=1}^I  \sum_{j=1}^I \iiint_{\mathbb{R}^3 \times \mathbb{R}^3 \times S^2} f_i(v) f_j(v_*) \, \mathcal{B}_{ij}(v,v_*,\sigma)  \\ \times  \left( s_{ij}(v') \left\langle v' \right\rangle_i^{2} +  s_{ji}(v'_*) \left\langle v'_* \right\rangle_j^{2}  - s_{ij}(v) \left\langle v \right\rangle_i^{ 2} - s_{ji}(v_*) \left\langle v_* \right\rangle_j^{2} \right)  \mathrm{d}\sigma \mathrm{d}v_* \mathrm{d} v.
\end{multline*}
Since the sign function is upper bounded  by 1, we obtain
\begin{multline*}
\left\| \mathbb{Q}(\mathbb{F}) \right\|_{L_2^1} \leq \frac{1}{2} \sum_{i=1}^I  \sum_{j=1}^I \iiint_{\mathbb{R}^3 \times \mathbb{R}^3 \times S^2} f_i(v) f_j(v_*) \, \mathcal{B}_{ij}(v,v_*,\sigma)  \\ \times  \left(  \left\langle v' \right\rangle_i^{2} + \left\langle v'_* \right\rangle_j^{2}  + \left\langle v \right\rangle_i^{2} + \left\langle v_* \right\rangle_j^{2} \right)  \mathrm{d}\sigma \mathrm{d}v_* \mathrm{d} v.
\end{multline*}
Using conservation of energy \eqref{CL micro energy}, together with the form of cross section \eqref{cross section}, implies
\begin{multline*}
\left\| \mathbb{Q}(\mathbb{F}) \right\|_{L_2^1} \leq   \sum_{i=1}^I  \sum_{j=1}^I \left\| b_{ij} \right\|_{L^1(\mathrm{d}\sigma)} \iint_{\mathbb{R}^3 \times \mathbb{R}^3 } f_i(v) f_j(v_*) \, \left|v-v_*\right|^{\gamma_{ij}}  \\ \times  \left(   \left\langle v \right\rangle_i^{ 2} + \left\langle v_* \right\rangle_j^{ 2} \right)   \mathrm{d}v_* \mathrm{d} v.
\end{multline*}
Finally, using upper bound \eqref{estimate on u^g product}, we obtain the estimate in terms of norms,
\begin{multline*}
\left\| \mathbb{Q}(\mathbb{F}) \right\|_{L_2^1} \leq \max_{1\leq i, j \leq I}\left( \left\| b_{ij} \right\|_{L^1(\mathrm{d}\sigma)}  \left( \frac{\sum_{i=1}^I m_i}{\sqrt{m_i m_j}} \right)^{\gamma_{ij}}  \right)  \\ \times \sum_{i=1}^I  \sum_{j=1}^I \iint_{\mathbb{R}^3 \times \mathbb{R}^3 } f_i(v) f_j(v_*) \, \left\langle v \right\rangle_i^{\bar{\bar{\gamma}}} \left\langle v_* \right\rangle_j^{\bar{\bar{\gamma}}}   \left(   \left\langle v \right\rangle_i^{2} + \left\langle v_* \right\rangle_j^{2} \right)   \mathrm{d}v_* \mathrm{d} v\\
= 2   \max_{1\leq i, j \leq I}\left( \left\| b_{ij} \right\|_{L^1(\mathrm{d}\sigma)}  \left( \frac{\sum_{i=1}^I m_i}{\sqrt{m_i m_j}} \right)^{\gamma_{ij}}  \right) \left(  \left\| \mathbb{F} \right\|_{L_{2+\bar{\bar{\gamma}}}^1}  \left\| \mathbb{F} \right\|_{L_{\bar{\bar{\gamma}}}^1}  \right).
\end{multline*}
Since $\mathbb{F} \in \Omega$ the right hand side is bounded, and therefore $\mathbb{Q}(\mathbb{F}) \in L_2^1$. \\

The next task is  to show that the mapping $\mathbb{F} \mapsto \mathbb{Q}(\mathbb{F})$, when restricted to $\Omega$ satisfies  (i) H\"{o}lder continuity, (ii) sub-tangent and (iii) one-sided Lipschitz conditions. Indeed, the  proof is divided into proofs of these three properties.

Assume that $\mathbb{F},\mathbb{G} \in \Omega$ and cross section $\mathcal{B}_{ij}$ is given in \eqref{cross section}. Then, the following three properties hold
\begin{itemize}
	\item[(i)] H\"{o}lder continuity condition:
	\begin{equation}\label{holder estimate}
	\left\| \mathbb{Q}(\mathbb{F}) - \mathbb{Q}(\mathbb{G}) \right\|_{L_2^1} \leq C_H \left\| \mathbb{F}-\mathbb{G} \right\|_{L_2^1}^{\frac{1}{2}},
	\end{equation}
		\item[(ii)] Sub-tangent condition: 
	\begin{equation*}
	\lim\limits_{h\rightarrow 0+} \frac{\text{dist}\left(\mathbb{F} + h \mathbb{Q}(\mathbb{F}), \Omega \right)}{h} =0,
	\end{equation*}
	where
	\begin{equation*}
	\text{dist}\left(\mathbb{H},\Omega\right)=\inf_{\omega \in \Omega} \left\| \mathbb{H}-\omega \right\|_{L_2^1}.
	\end{equation*}
	\item[(iii)] One-sided Lipschitz condition: 
	\begin{equation*}
 \left[\mathbb{Q}(\mathbb{F})-\mathbb{Q}(\mathbb{G}), \mathbb{F}-\mathbb{G}\right] \leq C_L \left\| \mathbb{F}-\mathbb{G} \right\|_{L_2^1},
	\end{equation*}
	where, by Remark \ref{Lip remark},
	\begin{multline*}
	 \left[\mathbb{Q}(\mathbb{F})-\mathbb{Q}(\mathbb{G}), \mathbb{F}-\mathbb{G}\right] =\lim_{h\rightarrow 0^-}\frac{\left(\left\| \left(\mathbb{F}-\mathbb{G}\right) + h \left(\mathbb{Q}(\mathbb{F})-\mathbb{Q}(\mathbb{G})\right) \right\|_{L_2^1} - \left\| \left(\mathbb{F}-\mathbb{G}\right) \right\|_{L_2^1} \right)}{h}  \\\leq \sum_{i=1}^I \int_{ \mathbb{R}^3} \left(\left[\mathbb{Q}(\mathbb{F})\right]_i(v) - \left[\mathbb{Q}(\mathbb{G})\right]_i(v) \right) \, \text{sign}\left(f_i(v)-g_i(v)\right) \left\langle v \right\rangle_i^2 \mathrm{d}v.
	\end{multline*}
\end{itemize}
Constants $C_H$ and $C_L$ depend on  $\left\| b_{ij} \right\|_{L^1(\mathrm{d}\sigma)}$, number of species $I$ and their masses $m_{i}$, $i=1,\dots,I$, and constants from characterization of the set $\Omega$.
\begin{proof}[Proof of (i) H\"{o}lder continuity condition]
	
	Let $\mathbb{F}=\left[f_i\right]_{1\leq i \leq I}$ and  $\mathbb{G}=\left[g_i\right]_{1\leq i \leq I}$ belong to $\Omega$. We need to estimate the following expression
	\begin{equation}\label{pomocna 4}
	I_{H}:=\left\| \mathbb{Q}(\mathbb{F}) - \mathbb{Q}(\mathbb{G}) \right\|_{L_2^1} 
	= \sum_{i=1}^I  \int_{\mathbb{R}^3}\left|  \sum_{j=1}^I \left( Q_{ij}(f_i,f_j) - Q_{ij}(g_i,g_j)  \right)  \right| \left\langle v \right\rangle_i^2 \mathrm{d} v.
	\end{equation}
	Using the binary structure  of collision operator \eqref{boltzmann i}, it follows
	\begin{equation}\label{collision operator diff and sum}
	Q_{ij}(f_i,f_j) - Q_{ij}(g_i,g_j)  = \frac{1}{2} \left(Q_{ij}(f_i-g_i, f_j + g_j) + Q_{ij}(f_i+g_i, f_j-g_j) \right).
	\end{equation}
	Therefore, using properties of absolute value, \eqref{pomocna 4} becomes
	\begin{equation}\label{pomocna 7}
	I_{H} \leq \frac{1}{2} \sum_{i=1}^I \sum_{j=1}^I \int_{\mathbb{R}^3} \left( \left|  Q_{ij}(f_i-g_i, f_j + g_j)   \right| +   \left|  Q_{ij}(f_i+g_i, f_j-g_j)  \right| \right) \left\langle v \right\rangle_i^2 \mathrm{d} v.
	\end{equation}
	The absolute value of collision operator will be written with the help of sign function, using  $\left|\cdot\right|= \cdot \,\text{sign}(\cdot)$. Since, at the end, all sign functions will be bounded by 1, we will not go deeply into details of its structure. So, let us for the moment denote 
	\begin{equation*}
	\text{sign}(Q_{ij}(f_i-g_i, f_j + g_j)) = s_{ij}^{-+}, \qquad \text{sign}(Q_{ij}(f_i+g_i, f_j - g_j)) = s_{ij}^{+-}.
	\end{equation*}
	Then, \eqref{pomocna 7} becomes
	\begin{multline}\label{pomocna 8}
	I_{H} \leq \frac{1}{2} \sum_{i=1}^I \sum_{j=1}^I \int_{\mathbb{R}^3} \left( Q_{ij}(f_i-g_i, f_j + g_j)  s_{ij}^{-+}\left\langle v \right\rangle_i^2 \right.  \\ \left. +    Q_{ij}(f_i+g_i, f_j-g_j) s_{ij}^{+-}\left\langle v \right\rangle_i^2 \right)  \mathrm{d} v.
	\end{multline}
	Now we  use the weak form \eqref{weak form ij+ji}, and  in order to do so, we have to match pairs. Indeed, we notice that the pair for $ij$-th element of the first sum is the $ji$-th element of the second sum. That is,  \eqref{weak form ij+ji} implies, after dropping the sign function,
	\begin{multline*}
	\int_{ v \in \mathbb{R}^3} \left( Q_{ij}(f_i-g_i, f_j + g_j)  s_{ij}^{-+}\left\langle v \right\rangle_i^2  +    Q_{ji}(f_j+g_j, f_i-g_i) s_{ji}^{+-}\left\langle v \right\rangle_j^2 \right) \mathrm{d}v
	\\ 
	\leq   \iiint_{\mathbb{R}^3 \times \mathbb{R}^3 \times S^2} \left|f_i(v)-g_i(v)\right| (f_j(v_*) + g_j(v_*)) \\ \times \left( \left\langle v' \right\rangle_i^2 +  \left\langle v'_* \right\rangle_j^2  + \left\langle v \right\rangle_i^2   + \left\langle v_* \right\rangle_j^2\right) \mathcal{B}_{ij}(v,v_*,\sigma) \mathrm{d}\sigma \mathrm{d}v_* \mathrm{d} v 
	\\
	= 2 \iiint_{\mathbb{R}^3 \times \mathbb{R}^3 \times S^2} \left|f_i(v)-g_i(v)\right| (f_j(v_*) + g_j(v_*)) \\ \times  \left(\left\langle v \right\rangle_i^2   + \left\langle v_* \right\rangle_j^2\right) \mathcal{B}_{ij}(v,v_*,\sigma) \mathrm{d}\sigma \mathrm{d}v_* \mathrm{d} v,
	\end{multline*}
	the last equality is due to the conservation law at the microscopic level \eqref{CL micro energy}. Therefore, \eqref{pomocna 8} becomes
	\begin{multline*}
	I_{H} \leq  \sum_{i=1}^I \sum_{j=1}^I \iiint_{\mathbb{R}^3 \times \mathbb{R}^3 \times S^2} \left|f_i(v)-g_i(v)\right| (f_j(v_*) + g_j(v_*)) \\ \times  \left(\left\langle v \right\rangle_i^2   + \left\langle v_* \right\rangle_j^2\right) \mathcal{B}_{ij}(v,v_*,\sigma) \mathrm{d}\sigma \mathrm{d}v_* \mathrm{d} v.
	\end{multline*}

	Now we use the form of cross section \eqref{cross section}.	Inequality \eqref{estimate on u^g}  yields the following upper bound of the previous expression 
	\begin{multline*}
	I_{H} \leq  \max_{1\leq i, j \leq I}\left( \left\| b_{ij} \right\|_{L^1(\mathrm{d}\sigma)}  \left( \frac{\sum_{i=1}^I m_i}{\sqrt{m_i m_j}} \right)^{\gamma_{ij}}  \right) \\ \sum_{i=1}^I \sum_{j=1}^I \iint_{\mathbb{R}^3 \times \mathbb{R}^3 } \left|f_i(v)-g_i(v)\right| (f_j(v_*) + g_j(v_*)) \\ \times  \left(\left\langle v \right\rangle_i^{2+\bar{\bar{\gamma}}} + \left\langle v \right\rangle_i^{2}  \left\langle v_* \right\rangle_j^{\bar{\bar{\gamma}}}  + \left\langle v_* \right\rangle_j^2 \left\langle v \right\rangle_i^{\bar{\bar{\gamma}}} + \left\langle v_* \right\rangle_j^{2+\bar{\bar{\gamma}}}\right) \mathrm{d}v_* \mathrm{d} v
	\\ \leq  \max_{1\leq i, j \leq I}\left( \left\| b_{ij} \right\|_{L^1(\mathrm{d}\sigma)}  \left( \frac{\sum_{i=1}^I m_i}{\sqrt{m_i m_j}} \right)^{\gamma_{ij}}  \right) \left(  \left\| \mathbb{F} - \mathbb{G} \right\|_{L_{2+\bar{\bar{\gamma}}}^1} \left\| \mathbb{F}+\mathbb{G} \right\|_{L_0^1} \right. \\ \left.  + \left\| \mathbb{F} - \mathbb{G}  \right\|_{L_2^1}  \left\| \mathbb{F}+\mathbb{G} \right\|_{L_{\bar{\bar{\gamma}}}^1} + \left\| \mathbb{F}- \mathbb{G} \right\|_{L_{\bar{\bar{\gamma}}}^1}\left\| \mathbb{F}+\mathbb{G} \right\|_{L_2^1}+ \left\|\mathbb{F}-\mathbb{G} \right\|_{L_0^1}\left\|\mathbb{F}+\mathbb{G} \right\|_{L_{2+\bar{\bar{\gamma}}}^1}  \right).
	\end{multline*}
	Monotonicity of the norm \eqref{monotonicity of norm} yields
	\begin{multline*}
	I_{H} \leq  2   \max_{1\leq i, j \leq I}\left( \left\| b_{ij} \right\|_{L^1(\mathrm{d}\sigma)}  \left( \frac{\sum_{i=1}^I m_i}{\sqrt{m_i m_j}} \right)^{\gamma_{ij}}  \right) \\ \times \left\| \mathbb{F} - \mathbb{G} \right\|_{L_{2+\bar{\bar{\gamma}}}^1} \left( \left\| \mathbb{F}+\mathbb{G} \right\|_{L_2^1}+\left\|\mathbb{F}+\mathbb{G} \right\|_{L_{2+\bar{\bar{\gamma}}}^1}  \right).
	\end{multline*}
	By the interpolation inequality \eqref{interpolation inequality vector}, it follows
	\begin{multline}\label{pomocna 15}
	I_{H} \leq  2 I   \max_{1\leq i, j \leq I}\left( \left\| b_{ij} \right\|_{L^1(\mathrm{d}\sigma)}  \left( \frac{\sum_{i=1}^I m_i}{\sqrt{m_i m_j}} \right)^{\gamma_{ij}}  \right) \\  \times  \left\| \mathbb{F} - \mathbb{G} \right\|_{L_2^1}^{1/2} \left\| \mathbb{F} - \mathbb{G} \right\|_{L_{2+2\bar{\bar{\gamma}}}^1}^{1/2} \left( \left\| \mathbb{F}+\mathbb{G} \right\|_{L_2^1}+\left\|\mathbb{F}+\mathbb{G} \right\|_{L_{2+\bar{\bar{\gamma}}}^1}  \right).
	\end{multline}
	Then we can bound term by term:
	\begin{equation*}
	\left\| \mathbb{F} - \mathbb{G} \right\|_{L_{2+2\bar{\bar{\gamma}}}^1}^{1/2} \leq \left\| \mathbb{F}   \right\|_{L_{2+2\bar{\bar{\gamma}}}^1}^{1/2} +\left\| \mathbb{G} \right\|_{L_{2+2\bar{\bar{\gamma}}}^1}^{1/2} \leq 2 \, C_{2+2\bar{\bar{\gamma}}}^{1/2},
	\end{equation*}
	and in the same fashion
	\begin{equation*}
	\left\| \mathbb{F}+\mathbb{G} \right\|_{L_2^1}  \leq 2 C_{2}, \qquad \left\| \mathbb{F}+\mathbb{G} \right\|_{L_{2+\bar{\bar{\gamma}}}^1} \leq 2 C_{2+\bar{\bar{\gamma}}},
	\end{equation*}
	since both $\mathbb{F}$ and $\mathbb{G}$ belong to $\Omega$. Therefore, \eqref{pomocna 15} becomes
	\begin{equation*}
	I_{H} \leq  8  \max_{1\leq i, j \leq I}\left( \left\| b_{ij} \right\|_{L^1(\mathrm{d}\sigma)}  \left( \frac{\sum_{i=1}^I m_i}{\sqrt{m_i m_j}} \right)^{\gamma_{ij}}  \right)   C_{2+2\bar{\bar{\gamma}}}^{1/2} \left( C_{2} + C_{2+\bar{\bar{\gamma}}} \right) \left\| \mathbb{F} - \mathbb{G} \right\|_{L_2^1}^{1/2},
	\end{equation*}
	which concludes the proof of H\"{o}lder continuity.
\end{proof}

\begin{proof}[Proof of (ii) sub-tangent condition]
	In order to prove sub-tangent condition, we first observe that, since we are in cut-off case, it is possible to split collision operator $\mathbb{Q}(\mathbb{F})$ into gain and loss term.  Namely,
	\begin{equation*}
	\left[ \mathbb{Q}(\mathbb{F}) \right]_i = \left[ \mathbb{Q}^+(\mathbb{F}) \right]_i - f_i(v) \left[ \nu(\mathbb{F}) \right]_i, 
	\end{equation*}
	where $\mathbb{Q}^+$ is a positive operator, and collision frequency $\nu(\mathbb{F})$, for any component $1\leq i\leq I$ reads 
	\begin{equation*} 
	\left[ \nu(\mathbb{F}) \right]_i=\sum_{j=1}^{I} \iint_{\mathbb{R}^3 \times S^2 } f_j(v_*) \mathcal{B}_{ij}(v,v_*,\sigma) \mathrm{d}\sigma \mathrm{d}v_* \geq 0.
	\end{equation*}
	In our case, $\nu(\mathbb{F})$ is finite whenever $\mathbb{F} \in \Omega$. Indeed, for the cross section \eqref{cross section}-\eqref{gamma max}, and since  $\left|v-v_*\right|^{\gamma_{ij}} \leq \left|v-v_*\right|^{\bar{\bar{\gamma}}}$, for $\left|v-v_*\right| \geq 1$ and $\left|v-v_*\right|^{\bar{\bar{\gamma}}} \leq \left|v\right|^{\bar{\bar{\gamma}}} + \left|v_*\right|^{\bar{\bar{\gamma}}}$,
	\begin{multline*}
	0 \leq \left[ \nu(\mathbb{F}) \right]_i(v) \leq \left( \max_{1\leq i, j \leq I} \left\| b_{ij} \right\|_{L^1(\mathrm{d}\sigma)} \right) \sum_{j=1}^{I} \int_{\mathbb{R}^3 } f_j(v_*)  \left|v-v_*\right|^{\gamma_{ij}}  \mathrm{d}v_*\\ 
	\leq    \left( \max_{1\leq i, j \leq I} \left\| b_{ij} \right\|_{L^1(\mathrm{d}\sigma)} \right) \left(  \sum_{j=1}^{I} \int_{\left|v-v_*\right|<1} f_j(v_*)  \, \mathrm{d}v_* \right. 
\\  \left.
	+ \sum_{j=1}^{I} \int_{\left|v-v_*\right|\geq 1} f_j(v_*)  \left|v-v_*\right|^{\bar{\bar{\gamma}}}  \mathrm{d}v_* \right) \\ 
	\leq   \left( \max_{1\leq i, j \leq I} \left\| b_{ij} \right\|_{L^1(\mathrm{d}\sigma)} \right)   \left( C_0 +  \left|v\right|^{\bar{\bar{\gamma}}} C_0 + \left( \frac{\sum_{i=1}^I m_i}{\min_{1\leq j \leq I} m_j} \right)^{{\bar{\bar{\gamma}}}/2}\left\| \mathbb{F} \right\|_{L_{\bar{\bar{\gamma}}}^1} \right) \\
	\leq K  \left( 1+\left|v\right|^{\bar{\bar{\gamma}}}\right),
	\end{multline*}
	where
	\begin{equation}\label{constant c tangency}
	K=  \left( \max_{1\leq i, j \leq I} \left\| b_{ij} \right\|_{L^1(\mathrm{d}\sigma)} \right) \left( 2 C_0 + \left( \frac{\sum_{i=1}^I m_i}{\min_{1\leq j \leq I} m_j} \right) C_{2} \right).
	\end{equation}
	\begin{proposition}\label{prop tangency}
		Fix $\mathbb{F} \in \Omega$. Then, for any $\varepsilon>0$ there exists $h_1>0$, such that $B(\mathbb{F}+h\mathbb{Q}(\mathbb{F}), h \varepsilon) \cap \Omega \neq \emptyset$, for any $0<h<h_1$.
	\end{proposition}
	\begin{proof}
		Set $\chi_R(v)$ the characteristic function of the ball of radius $R>0$ and introduce the truncated function $\mathbb{F}_R(t,v)=\chi_R(v)\mathbb{F}(t,v)$. Let
		\begin{equation}\label{wR}
		\mathbb{W}_R = \mathbb{F} + h \mathbb{Q}(\mathbb{F}_R).
		\end{equation}
		The idea of the proof is to find  $R$ such that from on one hand  $\mathbb{W}_R \in \Omega$, and on the another hand $\mathbb{W}_R \in B(\mathbb{F}+h\mathbb{Q}(\mathbb{F}), h \varepsilon)$, with $h$ explicitly calculated.
		\subsubsection*{Step 1.} We first show that it is possible to find an $h_1$ such that $\mathbb{W}_R$ remains non-negative for  as long  $0<h<h_1$. Indeed, for any $\mathbb{F} \in \Omega$ its truncation $\mathbb{F}_R \in \Omega$ as well. Since  $\mathbb{Q}^+$ is a positive operator, we have
		\begin{equation*}
		\left[\mathbb{W}_R\right]_i = f_i + h \left[\mathbb{Q}^+(\mathbb{F}_R)\right]_i - h \left[\mathbb{F}_R\right]_i \left[\nu(\mathbb{F}_R)\right]_i \geq f_i\left( 1 - h \, K\, \left( 1 + R^{\bar{\bar{\gamma}}}\right) \right) \geq 0,
		\end{equation*}
		for any $0<h<\frac{1}{K(1+R^{\bar{\bar{\gamma}}})}$, and $1\leq i \leq I$,
		with $K $ from \eqref{constant c tangency}.
		\subsubsection*{Step 2.} Since $\mathbb{F}_R \in \Omega$, we use   conservative properties of the collision operator detailed in \eqref{cons operator mass} and \eqref{cons operator energy}, to obtain
		\begin{equation*}
		\sum_{i=1}^I	\int_{\mathbb{R}^3} \left[ \mathbb{Q}(\mathbb{F}_R) \right]_i\mathrm{d}v = 0, 	\qquad \sum_{i=1}^I \left[ \mathbb{Q}(\mathbb{F}_R) \right]_i \left\langle v \right\rangle_i^2 \mathrm{d} v=0.
		\end{equation*}
		From \eqref{wR}, we get
		\begin{equation*}
		\mathfrak{m}_{0}[\mathbb{W}_R] = \mathfrak{m}_{0}[\mathbb{F}],  \qquad \mathfrak{m}_2[\mathbb{W}_R] = \mathfrak{m}_2[\mathbb{F}],
		\end{equation*}
		independently of $R$, which yields all needed lower and upper bounds on this quantities.
		\subsubsection*{Step 3} Finally, we need to show that $L^1_{k_*}$ norm of  $\mathbb{W}_R$ is bounded.%, which will imply, by the monotonicity of norms, that moments of order $2^+$ and $2+\gamma$ are also bounded. 
		
		Let the map $\mathcal{L}_{\overline{\gamma}, k_*}:[0,\infty) \rightarrow \mathbb{R}$, be defined with $\mathcal{L}_{\overline{\gamma}, k_*}(x)=- A_{k_*}x^{1+ \frac{\overline{\gamma}}{k_*}}+  B_{k_*}   x$, where $\overline{\gamma}\in (0,1]$ and  $k_*$ as defined in \eqref{kstar} 
		%is chosen large enough to ensure that \eqref{povzner constant prop} holds, 
		that yields  positivity of constants $A_{k_*}$ and $B_{k_*}$. It has only one root, denoted with $x^*_{\overline{\gamma}, k_*}$, at which $\mathcal{L}_{\overline{\gamma}, k_*}$ changes from positive to negative. Thus, for any $x\geq0$, we may write
		\begin{equation*}
		\mathcal{L}_{\overline{\gamma}, k_*}(x) \leq \max_{0\leq x \leq x^*_{\overline{\gamma}, k_*}} \mathcal{L}_{\overline{\gamma}, k_*}(x) =: \mathcal{L}^*_{\overline{\gamma}, k_*}.
		\end{equation*}
		Now, Lemma \ref{lemma Q za tangent cond} implies
		\begin{equation*}
		\sum_{i=1}^I \int_{ \mathbb{R}^3}\left[ \mathbb{Q}(\mathbb{F}) \right]_i \left\langle v \right\rangle_i^{k_*} \mathrm{d} v \leq  \mathcal{L}_{{\overline{\gamma}, k_*}}\left( \mathfrak{m}_{k_*}[\mathbb{F}] \right) \leq  \mathcal{L}^*_{\overline{\gamma}, k_*}.
		\end{equation*}
	Define
		\begin{equation*}
		\xi_{{\overline{\gamma}, k_*}} := x^*_{\overline{\gamma}, k_*} +  \mathcal{L}^*_{\overline{\gamma}, k_*}.
		\end{equation*}
		For any $\mathbb{F} \in \Omega$ we have two possibilities: either $\mathfrak{m}_{k_*}[\mathbb{F}] \leq x^*_{\overline{\gamma}, k_*}$ or   $\mathfrak{m}_{k_*}[\mathbb{F}]> x^*_{\overline{\gamma}, k_*}$.
		For the former, it follows that
		\begin{equation*}
		\mathfrak{m}_{k_*}[ \mathbb{W}_R]  \leq x^*_{\overline{\gamma}, k_*} + h \left( \sum_{i=1}^I \int_{ \mathbb{R}^3} \left[ \mathbb{Q}(\mathbb{F}_R) \right]_i \left\langle v \right\rangle_i^{k_*} \mathrm{d} v  \right) \leq x^*_{{\overline{\gamma}, k_*}} + h \, \mathcal{L}^*_{{\overline{\gamma}, k_*}} \leq \xi_{{\overline{\gamma}, k_*}},
		\end{equation*}
		where we have assumed, without loss of generality, that $h\leq 1$. For the latter, we choose $R=R(\mathbb{F})$ sufficiently large such that $\mathfrak{m}_{k_*}[\mathbb{F}_R]> x^*_{{\overline{\gamma}, k_*}}$, and therefore,
		\begin{equation*}
		\mathcal{L}_{{\overline{\gamma}, k_*}}\left(\mathfrak{m}_{k_*}[\mathbb{F}_R]\right) \leq 0.
		\end{equation*}
		As a consequence, 
		\begin{equation*}
		\mathfrak{m}_{k_*}[\mathbb{W}_R]  \leq  x^*_{{\overline{\gamma}, k_*}} \leq  \xi_{{\overline{\gamma}, k_*}}.
		\end{equation*}
		Therefore, we constructed a constant $C_{k_*}$ from characterization of the set $\Omega$, that is $ \xi_{{\overline{\gamma}, k_*}}$.\\

		The conclusion is that $\mathbb{W}_R \in \Omega$ for any $0<h<h_*$, where 
		\begin{equation*}
		h_*=\min\left\{1, \frac{1}{K \left(1 + R(\mathbb{F})^{\bar{\bar{\gamma}}}\right)}\right\},
		\end{equation*}
		and $K$ is from \eqref{constant c tangency}.\\
		
		Now, H\"{o}lder estimate \eqref{holder estimate} implies 
		\begin{equation*}
		h^{-1} \left\| \mathbb{F}+ h \mathbb{Q}(\mathbb{F}) - \mathbb{W}_R \right\|_{L_2^1} = \left\|  \mathbb{Q}(\mathbb{F}) - \mathbb{Q}(\mathbb{F}_R) \right\|_{L_2^1} \leq C_H \left\|  \mathbb{F} - \mathbb{F}_R \right\|_{L_2^1}^\frac{1}{2} \leq \varepsilon, 
		\end{equation*}
		for $R:=R(\varepsilon)$ sufficiently large. Then, for this choice of $R$, $\mathbb{W}_R \in B(\mathbb{F}+ h \mathbb{Q}(\mathbb{F}),h \epsilon)$.  \\
		
		Finally, choosing $R=\max\{R(\mathbb{F}),R(\epsilon)\}$ and $h_1$ as 
		\begin{equation}\label{h1}
		h_1=\min\left\{1, \frac{1}{K \left(1 + R^{\bar{\bar{\gamma}}}\right)}\right\},
		\end{equation}
		with $c$ given in \eqref{constant c tangency},
		one concludes that $\mathbb{W}_R \in B(\mathbb{F}+h \mathbb{Q}(\mathbb{F}), h \epsilon) \cap \Omega$. 
	\end{proof}
	Once the Proposition \ref{prop tangency} is proved, it immediately follows  
	\begin{equation*}
	h^{-1} \text{dist}\left(\mathbb{F}+ h \mathbb{Q}(\mathbb{F}), \Omega \right) \leq \varepsilon, \qquad \forall \, 0<h<h_1,
	\end{equation*}
	with $h_1$ from \eqref{h1},	which concludes the proof of tangency condition.
\end{proof}

\begin{proof}[Proof of (iii) one-sided Lipschitz condition]
	From definition and representation \eqref{collision operator diff and sum}, we have
	\begin{multline*}
	I_L:=\left[  \mathbb{Q}(\mathbb{F}) - \mathbb{Q}(\mathbb{G}), \mathbb{F}-\mathbb{G} \right]
	\\ \leq \sum_{i=1}^I  \sum_{j=1}^I  \int_{\mathbb{R}^3}  \left( Q_{ij}(f_i,f_j) -  Q_{ij}(g_i,g_j)  \right)  \text{sign}(f_i(v)-g_i(v)) \left\langle v \right\rangle_i^2 \mathrm{d}v 
	\\= \frac{1}{2}\sum_{i=1}^I  \sum_{j=1}^I  \int_{\mathbb{R}^3}  \left(Q_{ij}(f_i-g_i, f_j + g_j) + Q_{ij}(f_i+g_i, f_j-g_j) \right)  \text{sign}(f_i(v)-g_i(v))\left\langle v \right\rangle_i^2 \mathrm{d}v.
	\end{multline*}
	Changing $i\leftrightarrow j$ in the second integral, we precisely obtain binary structure of the weak form \eqref{weak form ij+ji} that yields
	\begin{multline*}
	I_L\leq \frac{1}{2}\sum_{i=1}^I  \sum_{j=1}^I  \int_{\mathbb{R}^3}  \left(Q_{ij}(f_i-g_i, f_j + g_j)\: \text{sign}(f_i(v)-g_i(v)) \:\left\langle v \right\rangle_i^2 \right. \\ \left.+ Q_{ji}(f_j+g_j, f_i-g_i) \: \text{sign}(f_j(v)-g_j(v) ) \: \left\langle v \right\rangle_j^2 \right)  \mathrm{d}v
	\\
	=\frac{1}{2}\sum_{i=1}^I  \sum_{j=1}^I  \iiint_{\mathbb{R}^3 \times \mathbb{R}^3 \times S^2} \mathcal{B}_{ij}(v,v_*,\sigma) \left( f_i(v) - g_i(v)\right) \left( f_j(v_*) + g_j(v_*)\right) \\ \times \left( \text{sign}(f_i(v')-g_i(v')) \:\left\langle v' \right\rangle_i^2 + \text{sign}(f_j(v'_*)-g_j(v'_*)) \:\left\langle v'_* \right\rangle_j^2 \right. \\ \left. - \text{sign}(f_i(v)-g_i(v)) \:\left\langle v \right\rangle_i^2  - \text{sign}(f_j(v_*)-g_j(v_*) ) \: \left\langle v_* \right\rangle_j^2 \right)  \mathrm{d}\sigma \mathrm{d}v_* \mathrm{d}v.
	\end{multline*}
	Using the upper bound of the sign function, one has
	\begin{multline*}
	I_L \leq \frac{1}{2}\sum_{i=1}^I  \sum_{j=1}^I  \iiint_{\mathbb{R}^3 \times \mathbb{R}^3 \times S^2} \mathcal{B}_{ij}(v,v_*,\sigma) \\ \times \left( \left| f_i(v) - g_i(v)\right| \left( f_j(v_*) + g_j(v_*)\right) \left( \left\langle v' \right\rangle_i^2 + \left\langle v'_* \right\rangle_j^2 \right) \right. \\ \left. - \left|f_i(v)-g_i(v)\right|  \left( f_j(v_*) + g_j(v_*)\right)  \:\left\langle v \right\rangle_i^2  \right. \\ \left.+ \left|f_i(v)-g_i(v) \right|   \left( f_j(v_*) + g_j(v_*)\right)  \left\langle v_* \right\rangle_j^2 \right)  \mathrm{d}\sigma \mathrm{d}v_* \mathrm{d}v.
	\end{multline*}
	Then, conservation of energy implies
	\begin{multline*}
	I_L \leq\sum_{i=1}^I  \sum_{j=1}^I  \iiint_{\mathbb{R}^3 \times \mathbb{R}^3 \times S^2} \mathcal{B}_{ij}(v,v_*,\sigma) \\ \times \left| f_i(v) - g_i(v)\right| \left( f_j(v_*) + g_j(v_*)\right) \left\langle v_* \right\rangle_j^2   \mathrm{d}\sigma \mathrm{d}v_* \mathrm{d}v.
	\end{multline*}
	Now, specifying the collision cross section \eqref{cross section} and using \eqref{estimate on u^g product} 
	\begin{equation*}
	\left|v-v_*\right|^{\gamma_{ij}} \leq  \left( \frac{\sum_{i=1}^I m_i}{\sqrt{m_i m_j}} \right)^{\gamma_{ij}} \left\langle v \right\rangle_i^{\gamma_{ij}}  \left\langle v_* \right\rangle_j^{\gamma_{ij}} \leq \left( \frac{\sum_{i=1}^I m_i}{\sqrt{m_i m_j}} \right)^{\gamma_{ij}} \left\langle v \right\rangle_i^{\bar{\bar{\gamma}}}  \left\langle v_* \right\rangle_j^{\bar{\bar{\gamma}}},
	\end{equation*}
	we obtain 
	\begin{equation*}
	I_L \leq  \max_{1\leq i, j \leq I}\left( \left\| b_{ij} \right\|_{L^1(\mathrm{d}\sigma)}  \left( \frac{\sum_{i=1}^I m_i}{\sqrt{m_i m_j}} \right)^{\gamma_{ij}}  \right) \left\| \mathbb{F} - \mathbb{G} \right\|_{L_{\bar{\bar{\gamma}}}^1}  \left\| \mathbb{F} + \mathbb{G} \right\|_{L_{2+\bar{\bar{\gamma}}}^1}.
	\end{equation*}
	Thanks to the monotonicity of norms \eqref{monotonicity of norm}
	\begin{equation*}
	\left\| \mathbb{F} - \mathbb{G} \right\|_{L_{\bar{\bar{\gamma}}}^1} \leq  \left\| \mathbb{F} - \mathbb{G} \right\|_{L_2^1},
	\end{equation*}
	we finally obtain
	\begin{equation*}
	I_L \leq 2    \max_{1\leq i, j \leq I}\left( \left\| b_{ij} \right\|_{L^1(\mathrm{d}\sigma)}  \left( \frac{\sum_{i=1}^I m_i}{\sqrt{m_i m_j}} \right)^{\gamma_{ij}}  \right) C_{2+\bar{\bar{\gamma}}} \left\| \mathbb{F} - \mathbb{G} \right\|_{L_2^1},
	\end{equation*}
	which completes the proof of one-sided Lipschitz condition.	
\end{proof}

\section{Proof of Theorem \ref{theorem bound on norm} (Generation  and propagation of polynomial moments)}\label{section proof generation of poly}

The proof consists of several steps. First, once the existence and uniqueness  of   vector value solution $\mathbb{F}$ to the Boltzmann system \eqref{Cauchy problem} is proven, we can derive from the Boltzmann system  an ordinary differential inequality for the scalar polynomial moment $\mathfrak{m}_k[\mathbb{F}](t)$. Then, the comparison principle for ODEs will yield estimates that guarantee both generation and propagation of these polynomial moments. 
\subsubsection*{Step 1. (Ordinary Differential Inequality for the polynomial moment)}
\begin{lemma}\label{ODI poly}
	Let $\mathbb{F}=\left[f_i\right]_{i=1,\dots,I}$ be a solution of the Boltzmann system \eqref{Cauchy problem}. Then the polynomial moment \eqref{poly moment} satisfies the following Ordinary Differential Inequality
	\begin{equation}\label{ODI F}
	\frac{\mathrm{d}}{\mathrm{d} t} \mathfrak{m}_k[\mathbb{F}](t) 
	\leq - A_k \, \mathfrak{m}_k[\mathbb{F}](t)^{1+\frac{\overline{\gamma}}{k}}+  B_k \,  \mathfrak{m}_k[\mathbb{F}](t),
	\end{equation}
	for   $k\ge k_*$ as defined in \eqref{kstar}, with positive constants $A_k$ and $B_k$ as defined in  Lemma~\ref{lemma Q za tangent cond}, equation \eqref{const Ak Bk}, after replacing $k_*$ by $k\ge k_*$.
\end{lemma}
\begin{proof}
	Consider $i-$th equation of the Boltzmann system \eqref{Cauchy problem},
	\begin{equation*}
	\partial_t f_i(t,v) =  \sum_{{j=1}}^I Q_{ij}(f_i,f_j)(t,v), \qquad i=1,\dots,I.
	\end{equation*}
	Integration with respect to velocity  $v$ with weight $\left\langle v \right\rangle_i^k$, $k\geq0$, and summation over all species $i=1,\dots,I$ yields 
	\begin{equation} \label{ode}
	\frac{\mathrm{d}}{\mathrm{d}t} \mathfrak{m}_k[\mathbb{F}](t) =  \sum_{i=1}^I \sum_{{j=1 }}^I \int_{\mathbb{R}^3} \left\langle v \right\rangle_i^k  Q_{ij}(f_i,f_j)(t,v)\mathrm{d}v,
	\end{equation}
	after recalling definition  \eqref{poly moment} of polynomial moment. Using results from Lemma \ref{lemma Q za tangent cond} for $k\ge k_*$ as defined in \eqref{kstar}, we conclude the estimate \eqref{ODI F}.
\end{proof}

\subsubsection*{Step 2. (Comparison principle)}

The starting point is the inequality \eqref{ODI F}.  We  associate to it an ODE of Bernoulli type 
\begin{equation}\label{associated ode}
 y'(t)
= - a \,y(t)^{1+c}+  b \,  y(t),
\end{equation}
whose solution will be an upper bound for $\mathfrak{m}_k[\mathbb{F}](t)$. Indeed, solution to \eqref{associated ode} reads
\begin{equation}\label{ode comparison}
y(t) = \left(\frac{a}{b}\left(1-e^{-c \, b\, t}\right) + y(0)^{-c} e^{-c\,b\,t}  \right)^{-\frac{1}{c}}.
\end{equation}
\subsubsection*{Step 3. (Generation of polynomial moments)}
Dropping initial data in \eqref{ode comparison} yields 
\begin{equation*}
y(t)  \leq  \left(\frac{a}{b}\left(1-e^{-c \,b \,t}\right)   \right)^{-\frac{1}{c}}, \qquad \forall t > 0.
\end{equation*}
Setting $y(t):= \mathfrak{m}_k[\mathbb{F}](t)$, $a:= A_k$, $b:=B_k$ and $c:= \overline{\gamma}/k$ this  implies generation estimate  \eqref{bound on norm} with 
 $$\mathfrak{C}^{\mathfrak{m}} = \left(\frac{A_k}{B_k}\right)^{-\frac{k}{\overline{\gamma}}}, \qquad  \text {for any} \ k\ge k_*.$$ 
 
 \begin{remark}
 	For later purposes, we derive also the following inequality by approximating the last result. Namely, for $t<1$, we may write 
 	$$
 	\left(1-e^{-c \,  b \,  t}\right)^{-\frac{1}{c}} = \left(c \, b \, t\right)^{-\frac{1}{c}} \left(1 + \frac{b}{2} t + o(t)\right) \leq \left(c \, b \right)^{-\frac{1}{c}}  e^{\frac{b}{2} t} \ t^{-\frac{1}{c}} \leq \left(c \, b \right)^{-\frac{1}{c}}  e^{\frac{b}{2} } \ t^{-\frac{1}{c}}.
 	$$
 	On the other hand, for $t\geq 1$, it follows
 	\begin{equation*}
 	\left(1-e^{-c \,b \,t}   \right)^{-\frac{1}{c}} \leq \left(1-e^{-c \,b }  \right)^{-\frac{1}{c}}.
 	\end{equation*}
 	Therefore,
 	\begin{equation}\label{generation poly decay}
 	y(t) \leq  \left(\frac{a}{b}  \right)^{-\frac{1}{c}} \begin{cases}
 	\left(c \, b \right)^{-\frac{1}{c}}  e^{\frac{b}{2} }  \ t^{-\frac{1}{c}}, & t<1 \\
 	\left(1-e^{-c \,b }  \right)^{-\frac{1}{c}}, & t\geq 1.
 	\end{cases}
 	\end{equation}
 	In other words, plugging $y(t):= \mathfrak{m}_k[\mathbb{F}](t)$, $a:= A_k$, $b:=B_k$ and $c:= \overline{\gamma}/k$, it yields
 	\begin{equation}\label{poly gen max t}
 	\mathfrak{m}_k[\mathbb{F}](t)  \leq \mathfrak{B}^{\mathfrak{m}} \ \max\{ 1, t^{-\frac{k}{\overline{\gamma}}}\}, \quad \forall t>0,
 	\end{equation}
 	where the constant is $$\mathfrak{B}^{\mathfrak{m}} = \mathfrak{C}^{\mathfrak{m}} \max\left\{ \left(\frac{\overline{\gamma}}{B_k k }\right)^{-\frac{k}{\overline{\gamma}}} e^{\frac{B_k}{2}}, \left( 1- e^{-\frac{\overline{\gamma}}{B_k k }}\right)^{-\frac{k}{\overline{\gamma}}} \right\}, \qquad  \text {for any} \ k\ge k_*.$$
 \end{remark}

\subsubsection*{Step 4. (Propagation of polynomial moments)}

For propagation result, when $y(0)$ is assumed to be finite, we first notice that $y(t)$ is a monotone function of $t$, which  approaches to $y(0)$ as  $t\rightarrow 0$ on one hand, and converges to  $ (a/b)^{-1/c}$ when $t\rightarrow\infty$  on the other hand. Therefore, $$y(t) \leq \max\{y(0), (a/b)^{-1/c}\},$$ for all $t\geq 0$. 
Again, taking $y(t):= \mathfrak{m}_k[\mathbb{F}](t)$, $a:= A_k$, $b:=B_k$ and $c:= \overline{\gamma}/k$,  for any $k\ge k_*$,   implies  the propagation estimate \eqref{poly ODI prop}.

\section{Generation and propagation of  exponential moments}\label{Section gen prop exp mom}

Let $\mathbb{F}$ be a solution of the Boltzmann system \eqref{Cauchy problem}. In this section we prove both generation and propagation of exponential moment \eqref{exp moment} related to  $\mathbb{F}$.
The proof strongly relies on generation and propagation of polynomial moments stated in Theorem \ref{bound on norm}. Moreover, it uses polynomial moment ODI, but  written in a slightly different manner than in Section \ref{ODI poly}, which we make precise in the following Lemma.
\begin{lemma}
	Let $\mathbb{F}$ be a solution of the Boltzmann system \eqref{Cauchy problem} with $\overline{\gamma}=\bar{\bar{\gamma}}$. Then there exists positive constants $K_1$ and $K_2$ such that the following two polynomial moments ODI hold
\begin{itemize}
	\item ODI needed for propagation of exponential moments
\begin{multline}\label{poly ODI prop}
\frac{\mathrm{d}}{\mathrm{d} t} \mathfrak{m}_{s k}[\mathbb{F}](t) \leq - K_1 \mathfrak{m}_{s k+ \overline{\gamma}}[\mathbb{F}](t) \\+   K_2 \left( \max_{1\leq i, j \leq I}\boldsymbol{\mathcal{C}}^{ij}_{\frac{s k}{2}} \right) \sum_{\ell=1}^{\ell_{k}} \left( \begin{matrix}
k\\ \ell 
\end{matrix} \right) \left( \mathfrak{m}_{s \ell+\overline{\gamma}}[\mathbb{F}](t) \ \mathfrak{m}_{s k- s \ell}[\mathbb{F}](t) + \mathfrak{m}_{s k- s \ell+\overline{\gamma}}[\mathbb{F}](t) \ \mathfrak{m}_{s \ell}[\mathbb{F}](t)\right).
\end{multline}
	\item  ODI needed for generation of exponential moments
		\begin{multline}\label{poly ODI gen}
	\frac{\mathrm{d}}{\mathrm{d} t} \mathfrak{m}_{\overline{\gamma} k}[\mathbb{F}](t) \leq - K_1 \mathfrak{m}_{\overline{\gamma} k+ \overline{\gamma}  }[\mathbb{F}](t) \\+   K_2  \left( \max_{1\leq i, j \leq I} \boldsymbol{\mathcal{C}}^{ij}_{\frac{\overline{\gamma} k}{2}}\right) \sum_{\ell=1}^{\ell_{k}} \left( \begin{matrix}
	k\\ \ell 
	\end{matrix} \right) \left( \mathfrak{m}_{\overline{\gamma} \ell+\overline{\gamma}}[\mathbb{F}](t) \ \mathfrak{m}_{\overline{\gamma} k- \overline{\gamma} \ell}[\mathbb{F}](t) + \mathfrak{m}_{\overline{\gamma} k- \overline{\gamma} \ell+\overline{\gamma}}[\mathbb{F}](t) \ \mathfrak{m}_{\overline{\gamma} \ell}[\mathbb{F}](t)\right).
	\end{multline}
\end{itemize}
\end{lemma}
\begin{proof}
	We briefly point out that the main steps in the proofs are adaption of the proof given in \cite{GambaTask18}. Let us consider polynomial moment
	\begin{equation*}
	\mathfrak{m}_{\delta q}[\mathbb{F}](t)=: \mathfrak{m}_{\delta q}, \quad 0<\delta\leq 2, \ q\geq 0,  \ \ \text{ with}\   \delta q>k_* ,
	\end{equation*}
	 with $k_*$  as defined in \eqref{kstar}, and derive an ODI for it starting from \eqref{poly moment start}
	so that   $\boldsymbol{\mathcal{C}}^{ij}_{\frac{\delta q }{2}} < \left\| b_{ij} \right\|_{L^1(\mathrm{d}\sigma)}$ holds uniformly for any pair $i,j=1,\dots,I$, with $\boldsymbol{\mathcal{C}}^{ij}_n$ being the constant from Povzner lemma \ref{Povzner intro}. Once we derive it, \eqref{poly ODI prop} will follow setting $\delta := s$, and \eqref{poly ODI gen} will follow with $\delta:=\overline{\gamma}$. Indeed, from \eqref{poly moment start} we get
	\begin{multline*}
	\mathfrak{m}_{\delta q}' = \sum_{i=1}^I \int_{ \mathbb{R}^3} \left[ \mathbb{Q}(\mathbb{F}) \right]_i \left\langle v \right\rangle_i^{\delta q} \mathrm{d} v  \leq \frac{1}{2}\sum_{i=1}^I \sum_{j=1}^I \iint_{\mathbb{R}^3 \times \mathbb{R}^3} f_i (v) f_j(v_*) \left|v-v_*\right|^{\gamma_{ij}} \\ \times \left( \boldsymbol{\mathcal{C}}^{ij}_{\frac{\delta q}{2}} \left(\left\langle v \right\rangle_i^2 + \left\langle v_* \right\rangle_j^2 \right)^{\frac{\delta q}{2}} -  \left\| b_{ij} \right\|_{L^1(\mathrm{d}\sigma)} \left( \left\langle v \right\rangle_i^{\delta q} + \left\langle v_* \right\rangle_j^{\delta q} \right) \right) \mathrm{d} v_* \mathrm{d}v.
	\end{multline*}
	Before applying  Lemma  \ref{binomial}, we first estimate, since $(\delta/2) \leq 1$,
	\begin{equation*}
	 \left(\left\langle v \right\rangle_i^2 + \left\langle v_* \right\rangle_j^2 \right)^{\frac{\delta q}{2}}  \leq  \left(\left\langle v \right\rangle_i^{\delta} + \left\langle v_* \right\rangle_j^{ \delta} \right)^{{ q}},
	\end{equation*}
	and then apply it, which gives the following 
		\begin{multline*}
 \left(\left\langle v \right\rangle_i^{\delta} + \left\langle v_* \right\rangle_j^{ \delta} \right)^{{ q}}	\leq \left\langle v \right\rangle_i^{\delta q} + \left\langle v_* \right\rangle_j^{\delta q} \\ + \sum_{\ell=1}^{\ell_{q}} \left( \begin{matrix}
		q\\ \ell 
	\end{matrix} \right) \left( \left\langle v \right\rangle_i^{ \delta \ell} \left\langle v_* \right\rangle_j^{ \delta q -  \delta \ell}+ \left\langle v \right\rangle_i^{ \delta q-  \delta \ell} \left\langle v_* \right\rangle_j^{\delta \ell} \right),
\end{multline*}
with $\ell_q = \lfloor\frac{q + 1}{2}\rfloor$. The bound from above and below of the non-angular part of the cross-section, $ \left|v-v_*\right|^{\gamma_{ij}}$, is used as in Section \ref{ODI poly}. This implies  a polynomial moment ODI 
\begin{equation*}
\mathfrak{m}_{\delta q}'(t) \leq  - K_1  \mathfrak{m}_{\delta q + \overline{\gamma}} +   \boldsymbol{\mathcal{C}}^{ij}_{\frac{\delta q}{2}} \, K_2  \sum_{\ell=1}^{\ell_{q}} \left( \begin{matrix}
q\\ \ell 
\end{matrix} \right) \left( \mathfrak{m}_{ \delta \ell + \overline{\gamma}} \mathfrak{m}_{\delta q -\delta \ell} + \mathfrak{m}_{\delta q -\delta \ell +\overline{\gamma}} \mathfrak{m}_{ \delta \ell } \right),
\end{equation*}
where $K_1$ and $K_2$ are  positive constants  since $\delta q \ge k_*$, with $k_*$ as defined in \eqref{kstar}, 
\begin{equation*}
\begin{split}
K_1 &= \min_{1\leq i, j \leq I}  \left(\left\| b_{ij} \right\|_{L^1(\mathrm{d}\sigma)}  - \boldsymbol{\mathcal{C}}^{ij}_{\frac{\delta q }{2}}\right)  \frac{c_{lb} }{\max_{1\leq i \leq I} m_i}, \\ 
K_2 &=\frac{1}{2} \left( \max_{1\leq i, j \leq I} \left( \frac{\sum_{i=1}^I m_i}{\sqrt{m_i m_j}} \right)^{\gamma_{ij}}\right),
\end{split}
\end{equation*}  
which completes the proof.
\end{proof}

\section{Proof of Theorem  \ref{theorem gen prop ML} (b) (Propagation of   exponential moments)}

Using Taylor series of an exponential function, one can represent  exponential moment  as 
\begin{equation*}
\mathcal{E}_{s}[\mathbb{F}]({\alpha},t) = \sum_{k=0}^\infty \frac{\alpha^k}{k!} \, \mathfrak{m}_{s k}[\mathbb{F}](t).
\end{equation*}
We will show that the exponential rate $\alpha=\alpha(k_*)$, that is, depending on the $k_*$ parameter defined in \eqref{kstar}  for $\overline{\gamma}=\bar{\bar{\gamma}}$.

We consider its partial sum an a shifted by $\overline{\gamma}$ one, namely,
\begin{equation}\label{pomocna 18}
\mathcal{E}_{s}^n[\mathbb{F}]({\alpha},t) = \sum_{k=0}^n \frac{\alpha^k}{k!} \, \mathfrak{m}_{s k}[\mathbb{F}](t), \quad \mathcal{E}_{s;\overline{\gamma}}^n[\mathbb{F}]({\alpha},t)= \sum_{k=0}^n \frac{\alpha^k}{k!} \, \mathfrak{m}_{s k+ \overline{\gamma}}[\mathbb{F}](t).
\end{equation}
In order to have lighter writing, we will drop from moment notation dependence on $t$ and $\alpha$, and relation to $\mathbb{F}$, and we will  instead write
\begin{equation*}
\mathcal{E}_{s}^n[\mathbb{F}]({\alpha},t)=:\mathcal{E}_{s}^n, \qquad \mathcal{E}_{s;\overline{\gamma}}^n[\mathbb{F}]({\alpha},t):=\mathcal{E}_{s;\overline{\gamma}}^n,  \qquad  \mathfrak{m}_{s k+ \overline{\gamma}}[\mathbb{F}](t)=:\mathfrak{m}_{s k+ \overline{\gamma}}. 
\end{equation*}
When it will be important to highlight dependence on $t$ and $\alpha$, we will also, for example, write $\mathcal{E}_{s}^n(\alpha,t)$ instead of $\mathcal{E}_{s}^n$.\\
 
The idea of  proof is to show that the partial sum $\mathcal{E}_{s}^n$ is bounded uniformly in time $t$ and $n$. To this end, we first derive ordinary differential inequality (ODI) for it.

\changelocaltocdepth{1}

\subsection*{ODI for $\mathcal{E}_{s}^n$}
Taking derivative with respect to time $t$ of \eqref{pomocna 18}, we get 
\begin{align*}
\frac{\mathrm{d} \, }{\mathrm{d} t}\mathcal{E}_{s}^n %&= \sum_{k=0}^n \frac{\alpha^{ k}}{k!}   \mathfrak{m}_{s k}' \\ & 
= \sum_{k=0}^{k_0-1} \frac{\alpha^{ k}}{k!}  \mathfrak{m}_{s k}' +  \sum_{k=k_0}^n \frac{\alpha^{ k}}{k!}  \mathfrak{m}_{s k}',
\end{align*}
where $k_0$ is an index that will be determined later on.  We use a polynomial moment ODE \eqref{poly ODI prop} for the second term that yields 
\begin{multline}\label{pomocna 10}
\frac{\mathrm{d} \, }{\mathrm{d} t}\mathcal{E}_{s}^n \leq \sum_{k=0}^{k_0-1} \frac{\alpha^{k}}{k!}  \mathfrak{m}_{s k}'  - K_1  \sum_{k=k_0}^n \frac{\alpha^{ k}}{k!} \,   \mathfrak{m}_{s k+ \overline{\gamma}} \\ +  K_2 \sum_{k=k_0}^n  \left( \max_{1\leq i, j \leq I}\boldsymbol{\mathcal{C}}^{ij}_{\frac{sk}{2}}\right) \frac{\alpha^{ k}}{k!}   \sum_{\ell=1}^{\ell_{k}} \left( \begin{matrix}
k \\ \ell 
\end{matrix} \right) \left( \mathfrak{m}_{s \ell+\overline{\gamma}}  \mathfrak{m}_{s k- s \ell} + \mathfrak{m}_{s k- s \ell+\overline{\gamma}} \mathfrak{m}_{s \ell}\right)\\
=: S_0 - K_1 S_1 +  K_2 S_2.
\end{multline}
We estimate each sum $S_0$, $S_1$ and $S_2$ separately.
\subsubsection*{Term $S_0$} Propagation of polynomial moments \eqref{poly propagation} ensures bound on $\mathfrak{m}_{s k}$ uniformly in time, which implies from \eqref{ODI F} bound on its derivative, i.e. there exist a constant $c_{k_0}$ such that 
\begin{equation}\label{pomocna 9}
\mathfrak{m}_{s k}, \mathfrak{m}_{s k}' \leq c_{k_0} \quad \text{for all} \ k\in\{0,1,\dots,k_0\}.
\end{equation} 
For $S_0$ this yields 
\begin{equation}\label{S0 estimate}
S_0 \leq c_{k_0}  \sum_{k=0}^{k_0-1} \frac{\alpha^{k}}{k!}  \leq c_{k_0} e^{\alpha} \leq 2 \, c_{k_0},
\end{equation}
for $\alpha$ small enough to satisfy
\begin{equation}\label{alpha first cond}
e^{\alpha} \leq 2.
\end{equation}

\subsubsection*{Term $S_1$}  We complete first the term $S_1$ to make appear shifted partial sum $\mathcal{E}_{s;\overline{\gamma}}^n$ by means of
\begin{equation*}
S_1 =  \sum_{k=k_0}^n \frac{\alpha^{k}}{k!} \, D_k \, \mathfrak{m}_{s k+ \overline{\gamma}} = \mathcal{E}_{s;\overline{\gamma}}^n - \sum_{k=0}^{k_0-1} \frac{\alpha^{ k}}{k!} \, D_k \, \mathfrak{m}_{s k+ \overline{\gamma}}.
\end{equation*}
From the bound \eqref{pomocna 9} we can estimate $\mathfrak{m}_{s k+ \overline{\gamma}}$ as well,
\begin{equation*}
\mathfrak{m}_{s k+ \overline{\gamma}} \leq c_{k_0}, \quad k=0,\dots,k_0-1,
\end{equation*}
which together with considerations for  the term $S_0$ yields
\begin{equation}\label{S1 estimate}
S_1 \geq \mathcal{E}_{s;\overline{\gamma}}^n - 2 c_{k_0}.
\end{equation}

\subsubsection*{Term $S_2$} Term $S_2$ can be separated into two terms, namely
\begin{equation*}
S_2 = \sum_{k=k_0}^n  \left( \max_{1\leq i, j \leq I}\boldsymbol{\mathcal{C}}^{ij}_{\frac{sk}{2}} \right) \frac{\alpha^{ k}}{k!}   \sum_{\ell=1}^{\ell_{k}} \left( \begin{matrix}
k \\ \ell 
\end{matrix} \right) \left( \mathfrak{m}_{s \ell+\overline{\gamma}}  \mathfrak{m}_{s k- s \ell} + \mathfrak{m}_{s k- s \ell+\overline{\gamma}} \mathfrak{m}_{s \ell}\right)=: S_{2_1} + S_{2_2}. 
\end{equation*}
Their treatment is the same, so let perform  an estimate on $S_{2_1}$. Rearranging we can write
\begin{equation*}
S_{2_1} =  \sum_{k=k_0}^n \left( \max_{1\leq i, j \leq I} \boldsymbol{\mathcal{C}}^{ij}_{\frac{sk}{2}} \right) \sum_{\ell=1}^{\ell_{k}}\frac{\alpha^\ell \, \mathfrak{m}_{s \ell+\overline{\gamma}}}{\ell!}  \frac{\alpha^{k-\ell} \, \mathfrak{m}_{s k- s \ell}}{(k-\ell)!} \leq \left( \max_{1\leq i, j \leq I} \boldsymbol{\mathcal{C}}^{ij}_{\frac{sk_0}{2}}\right)  \mathcal{E}_{s;\overline{\gamma}}^n \, \mathcal{E}_{s}^n,
\end{equation*}
the last inequality is due to the decreasing property of $\boldsymbol{\mathcal{C}}^{ij}_k$ in $k\ge k_*$, uniformly for any $i,j$, with $k_*$ defined in \eqref{kstar}.  Therefore, we can estimate 
\begin{equation}\label{S2 estimate}
S_{2} \leq  2 \left( \max_{1\leq i, j \leq I} \boldsymbol{\mathcal{C}}^{ij}_{\frac{s{k_0}}{2}} \right) \mathcal{E}_{s;\overline{\gamma}}^n \, \mathcal{E}_{s}^n.\\
\end{equation}

Finally, desired ODI for $\mathcal{E}_{s}^n$ is obtained from \eqref{pomocna 10} gathering all estimates \eqref{S0 estimate}, \eqref{S1 estimate} and \eqref{S2 estimate}. Namely,
\begin{equation}\label{pomocna 12}
\frac{\mathrm{d} \, }{\mathrm{d} t}\mathcal{E}_{s}^n \leq  - K_1 \mathcal{E}_{s;\overline{\gamma}}^n
+ 2 c_{k_0} (1+K_1)    + 2\, K_2  \left( \max_{1\leq i, j \leq I} \boldsymbol{\mathcal{C}}^{ij}_{\frac{s{k_0}}{2}} \right) \mathcal{E}_{s;\overline{\gamma}}^n \, \mathcal{E}_{s}^n.
\end{equation}

\subsection*{Bound on $\mathcal{E}_{s}^n$} For each $n \in \mathbb{N}$ we define
\begin{equation*}
T_n := \sup\{ t\geq 0: \mathcal{E}_{s}^n(\alpha,\tau) \leq 4 M_0, \ \forall \tau \in [0,t] \}, 
\end{equation*}
where $M_0$ is a bound on initial data in \eqref{initial data exp prop}. We will show that $\mathcal{E}_{s}^n(t)$ is uniformly bounded in $t$ and $n$ by proving that  $T_n=\infty$ for all $n\in \mathbb{N}$.\\

The sequence $T_n$ is well-defined and positive. Indeed,  since $\alpha \leq \alpha_0$, at time $t=0$ we have
\begin{equation*}
\mathcal{E}_{s}^n(\alpha,0) = \sum_{k=0}^n \frac{\alpha^{ k}}{k!} \mathfrak{m}_{s k}(0) \leq  \sum_{k=0}^n \frac{\alpha_0^{ k}}{k!} \mathfrak{m}_{s k}(0) \leq \mathcal{E}_{s}({\alpha_0},0) < 4 M_0,
\end{equation*}
uniformly in $n$, by assumption \eqref{initial data exp prop}. Since each term $\mathfrak{m}_{s k}(t)$ is continuous function of $t$, so is $\mathcal{E}_{s}^n(\alpha, t)$. Therefore, $\mathcal{E}_{s}^n(\alpha, t)< 4M_0$ on some time interval $[0,t_n)$, $t_n>0$. Thus $T_n$ is well-defined and positive for every $n\in \mathbb{N}$.\\

For $t\in [0,T_n]$ it follows $\mathcal{E}_{s}^n(\alpha, t)\leq4 M_0$, which from \eqref{pomocna 12} implies
\begin{equation}\label{pomocna 20}
\frac{\mathrm{d} \, }{\mathrm{d} t}\mathcal{E}_{s}^n \leq  -  \mathcal{E}_{s;\overline{\gamma}}^n \left( K_1 - 8 \,K_2 \left( \max_{1\leq i, j \leq I} \boldsymbol{\mathcal{C}}^{ij}_{\frac{s k_0}{2}} \right)  M_0  \right) \\ +  2 c_{k_0} \left(  1  + K_1 \right).
\end{equation}
Since $\boldsymbol{\mathcal{C}}^{ij}_{\frac{s k_0}{2}}$,  for any $i,j$,  converges to zero as  $\frac{s k_0}{2}> k_*$ goes to infinity we can choose   $k_0> 2\frac{k_*}{s}$ such that
\begin{equation*}
K_1 - 8 \,K_2 \left( \max_{1\leq i, j \leq I}  \boldsymbol{\mathcal{C}}^{ij}_{k_*} \right)  M_0    > \frac{K_1}{2}, 
\end{equation*} 
or, equivalently, 
\begin{equation}\label{K_1kstar}
K_1 <   16 \,K_2 \left( \max_{1\leq i, j \leq I}{  \boldsymbol{\mathcal{C}}^{ij}_{k_*} }\right)  M_0, 
\end{equation}  
with $K_1$ depending on  $k_*$ as defined in \eqref{kstar}. 
Hence,  \eqref{pomocna 20} becomes 
\begin{equation}\label{pomocna 21}
\frac{\mathrm{d} \, }{\mathrm{d} t}\mathcal{E}_{s}^n \leq  -\frac{K_1}{2}\,  \mathcal{E}_{s;\overline{\gamma}}^n+  2 { c_{k_*}} \left(  1  + K_1 \right).
\end{equation}
Next step consists in finding lower bound for $\mathcal{E}_{s;\overline{\gamma}}^n$ in terms of $\mathcal{E}_{s}^n$. Indeed,
we can estimate
\begin{multline*}
\mathcal{E}_{s;\overline{\gamma}}^n =  \sum_{k=0}^n \frac{\alpha^{ k}}{k!} \mathfrak{m}_{s k+ \overline{\gamma}} \geq \sum_{k=0}^n \frac{\alpha^{ k}}{k!} \sum_{i=1}^I \int_{\left\{ \left\langle v \right\rangle_i \geq \alpha^{-1/2} \right\}} f_i(t,v) \left\langle v \right\rangle_i^{sk+\overline{\gamma}} \mathrm{d}v 
\\ \geq  \alpha^{-\overline{\gamma}/2} \left( \mathcal{E}_{s}^n - \sum_{k=0}^n \frac{\alpha^{ k}}{k!} \sum_{i=1}^I \int_{\left\{ \left\langle v \right\rangle_i < \alpha^{-1/2} \right\}} f_i(t,v) \left\langle v \right\rangle_i^{sk} \mathrm{d}v  \right)
\\ \geq  \alpha^{-\overline{\gamma}/2} \left( \mathcal{E}_{s}^n - \sum_{k=0}^n \frac{\alpha^{k (1-\frac{s}{2})}}{k!} \mathfrak{m}_{0}(0)  \right) \geq  \alpha^{-\overline{\gamma}/2} \left( \mathcal{E}_{s}^n - \mathfrak{m}_{0}(0) e^{\alpha^{1-\frac{s}{2}}} \right).
\end{multline*}
Plugging this result into \eqref{pomocna 21} yields
\begin{equation*}
\frac{\mathrm{d} \, }{\mathrm{d} t}\mathcal{E}_{s}^n \leq  -\frac{K_1}{2}\, \alpha^{-\overline{\gamma}/2} \mathcal{E}_{s}^n +  \frac{K_1}{2} \alpha^{-\overline{\gamma}/2}  \mathfrak{m}_{0}(0) e^{\alpha^{1-\frac{s}{2}}}  +  2 c_{k_0} \left(  1  + K_1 \right).
\end{equation*}
By the maximum principle for ODEs, it follows
 \begin{multline}\label{pomocna 22}
\mathcal{E}_{s}^n(\alpha, t) \leq \max\left\{ \mathcal{E}_{s}^n(\alpha, 0), \mathfrak{m}_{0}(0) \, e^{\alpha^{1-\frac{s}{2}}}  + \frac{4 c_{k_0} \left(  1  + K_1 \right)}{ K_1 \, \alpha^{-\overline{\gamma}/2}}  \right\}  \\  \\
\leq  M_0 +  \mathfrak{m}_{0}(0) \, e^{\alpha^{1-\frac{s}{2}}}  + \alpha^{\overline{\gamma}/2}  \ \frac{4 c_{k_0} \left(  1  + K_1 \right)}{ K_1 \, },
\end{multline}
 for any $t\in [0,T_n].$
On the other hand, since $s \leq 2$, the following limit property holds
 \begin{equation*}
 \mathfrak{m}_{0}(0) \, e^{\alpha^{1-\frac{s}{2}}}  + \alpha^{\overline{\gamma}/2}  \ \frac{4 c_{k_0} \left(  1  + K_1 \right)}{ K_1 \, } \rightarrow  \mathfrak{m}_{0}(0), \quad \text{as} \ \alpha\rightarrow 0,
 \end{equation*} 
and $\mathfrak{m}_{0}(0) < \mathcal{E}_{s}^n(\alpha_0,0)$ for any $\alpha_0$, and therefore, by \eqref{initial data exp prop}, $\mathfrak{m}_{0}(0)< M_0$. Thus, we can  choose sufficiently small $\alpha = \alpha_1$ such that 
\begin{equation}\label{alpha second cond}
\mathfrak{m}_{0}(0) \, e^{\alpha^{1-\frac{s}{2}}}  + \alpha^{\overline{\gamma}/2}  \ \frac{4 c_{k_*} \left(  1  + K_1 \right)}{ K_1 \, }< 3 M_0,
\end{equation}
for any $s\leq 2$ and  $K_1=K_1(k_*)$ from \eqref{K_1kstar}. In that case, inequality \eqref{pomocna 22} implies the following strict inequality
 \begin{equation}\label{pomocna 23}
\mathcal{E}_{s}^n(\alpha, t) < 4 M_0, 
\end{equation}
for any $t\in [0,T_n]$ and $0<\alpha(k_*) \leq \alpha_1$,  with $\alpha$ depending on $k_*$ defined in \eqref{kstar}.\\
\subsection*{Conclusion I} If $k_0$ is chosen such that \eqref{pomocna 21} holds, and the choice of  $\alpha$ is such that $0< \alpha \leq \alpha_0$ and  \eqref{alpha first cond}, \eqref{alpha second cond} are satisfied, which amounts to take   $\alpha =\min \left\{ \alpha_0, \ln 2, \alpha_1 \right\}$, then we have strict inequality \eqref{pomocna 23}, $\mathcal{E}_{s}^n(\alpha, t) < 4 M_0$, that holds on the closed interval $ [0,T_n]$ uniformly in $n$. Because of the continuity of $\mathcal{E}_{s}^n(\alpha,t)$ with respect to time $t$, this strict inequality actually  holds on a slightly larger time interval $[0, T_n + \varepsilon)$, $\varepsilon >0$. This contradicts the maximality of $T_n$ unless $T_n = + \infty$. Therefore,  
$\mathcal{E}_{s}^n(\alpha, t) \leq 4 M_0$ for all $t\geq 0$ and $n\in \mathbb{N}$. Thus, letting $n\rightarrow \infty$ we conclude
\begin{equation*}
\mathcal{E}_{s}[\mathbb{F}]({\alpha},t) = \lim\limits_{n\rightarrow \infty} \mathcal{E}_{s}^n[\mathbb{F}]({\alpha},t) \leq 4 M_0, \quad \forall t\geq 0,
\end{equation*}
 i.e. the solution $\mathbb{F}$ to  system of Boltzmann equations with finite initial exponential moment of order $s$ and rate $\alpha_0$ will propagate exponential moments of the same order $s$ and a rate $\alpha$ that satisfies  $\alpha =\min \left\{ \alpha_0, \ln 2, \alpha_1 \right\}$.   It is also very interesting to note that  the rate $\alpha$ depends on the $k_*$ parameter from \eqref{kstar}, which depends on uniform in the $i,j$ pairs upper bounds for the intermolecular potentials  $\gamma_{ij}$ and for  controls of the $k_*^{ij}$ as defined in \eqref{povzner constant prop} in the Povzner lemma~\ref{Povzner intro}.

\section{Proof of Theorem  \ref{theorem gen prop ML} (a)  (Generation of   exponential  moments)}

We consider    an exponential moment of order $\overline{\gamma}=\bar{\bar{\gamma}}$ and rate $\alpha t$,  where $\alpha$ depends on $k_*$ from \eqref{kstar}, for the solution $\mathbb{F}$ of the Boltzmann system, namely 
\begin{equation*}
\mathcal{E}_{\overline{\gamma}}[\mathbb{F}](\alpha t,t)=  \sum_{i=1}^I \int_{\mathbb{R}^3} f_i(t,v) \, e^{\alpha t \left\langle v \right\rangle_i^{\overline{\gamma}}} \mathrm{d}v =\sum_{k=0}^\infty \frac{(\alpha t)^k}{k!} \mathfrak{m}_{\overline{\gamma} k}[\mathbb{F}](t).
\end{equation*}
Consider its partial sum, and a shifted one
\begin{equation*}
\mathcal{E}^n_{\overline{\gamma}}[\mathbb{F}](\alpha t,t)=  \sum_{k=0}^n \frac{(\alpha t)^k}{k!} \mathfrak{m}_{\overline{\gamma} k}[\mathbb{F}](t), \qquad \mathcal{E}^n_{\overline{\gamma}; \overline{\gamma}}[\mathbb{F}](\alpha t,t)=  \sum_{k=0}^n \frac{(\alpha t)^k}{k!} \mathfrak{m}_{\overline{\gamma} k+ \overline{\gamma}  }[\mathbb{F}](t).
\end{equation*}
As usual, we will most of the time relieve notation by omitting explicit dependence on time $t$ and relation to $\mathbb{F}$, and write
 \begin{equation*}
 \mathcal{E}^n_{\overline{\gamma}}[\mathbb{F}](\alpha t,t)=: \mathcal{E}^n_{\overline{\gamma}}, \quad  \mathcal{E}^n_{\overline{\gamma}; \overline{\gamma}}[\mathbb{F}](\alpha t,t):= \mathcal{E}^n_{\overline{\gamma}; \overline{\gamma}}.
 \end{equation*}
 
 Fix $\alpha$ and $\overline{\gamma}$ and define
 \begin{equation*}
 \bar{T}_n := \sup\left\{  t\in [0,1]: \mathcal{E}^n_{\overline{\gamma}}[\mathbb{F}](\alpha t,t)\leq 4 \bar{M}_0 \right\}.
 \end{equation*}
 $ \bar{T}_n $ is well defined. Indeed, taking $\bar{M}_0:= \sum_{i=1}^I f_i(t,v) \left\langle v \right\rangle_i^2 \mathrm{d}v =  \sum_{i=1}^I f_i(0,v) \left\langle v \right\rangle_i^2 \mathrm{d}v$, for $t=0$, we get $\mathcal{E}^n_{\overline{\gamma}}(0,0) \leq   \mathcal{E}_{\overline{\gamma}}(0,0)=\mathfrak{m}_{0}(0) < 4 \bar{M}_0$.  By continuity of partial sum $\mathcal{E}^n_{\overline{\gamma}}$ with respect to $t$, $\mathcal{E}^n_{\overline{\gamma}}(\alpha t,t) \leq 4 \bar{M}_0$ on a slightly larger time interval $t\in[0,t_n)$, $t_n>0$, and thus $\bar{T}_n>0$.

\subsection*{ODI for $ \mathcal{E}^n_{\overline{\gamma}}$}
Taking time derivative of $ \mathcal{E}^n_{\overline{\gamma}}$  yields 
\begin{equation*}
\frac{\mathrm{d}}{\mathrm{d} t} \mathcal{E}^n_{\overline{\gamma}} = \alpha \sum_{k=1}^n \frac{(\alpha t)^{k-1}}{(k-1)!} \mathfrak{m}_{\overline{\gamma} k} + \sum_{k=0}^{k_0-1}  \frac{(\alpha t)^{k}}{k!}  \mathfrak{m}_{\overline{\gamma} k}' +  \sum_{k=k_0}^n \frac{(\alpha t)^{k}}{k!}   \mathfrak{m}_{\overline{\gamma} k}'.
\end{equation*}
For the first term we simply re-index the sum and use definition of shifted partial sum, and for the last one we use polynomial moment ODI \eqref{poly ODI gen}, which together implies
\begin{multline}\label{pomocna 11}
\frac{\mathrm{d}}{\mathrm{d} t} \mathcal{E}^n_{\overline{\gamma}} \leq \alpha \,  \mathcal{E}^n_{\overline{\gamma}; \overline{\gamma}} + \sum_{k=0}^{k_0-1}  \frac{(\alpha t)^{k}}{k!}  \mathfrak{m}_{\overline{\gamma} k}' - K_1  \sum_{k=k_0}^n \frac{(\alpha t)^{k}}{k!}  
 \mathfrak{m}_{\overline{\gamma} k+ \overline{\gamma}  }\\+   K_2  \sum_{k=k_0}^n \frac{(\alpha t)^{k}}{k!}  \left( \max_{1\leq i, j \leq I} \boldsymbol{\mathcal{C}}^{ij}_{\frac{\overline{\gamma} k}{2}} \right) \sum_{\ell=1}^{\ell_{k}} \left( \begin{matrix}
 k\\ \ell 
 \end{matrix} \right) \left( \mathfrak{m}_{\overline{\gamma} \ell+\overline{\gamma}}  \mathfrak{m}_{\overline{\gamma} k- \overline{\gamma} \ell} + \mathfrak{m}_{\overline{\gamma} k- \overline{\gamma} \ell+\overline{\gamma}} \mathfrak{m}_{\overline{\gamma} \ell}\right)
\\ =: \alpha \,  \mathcal{E}^n_{\overline{\gamma}; \overline{\gamma}} + S_0 - K_1 S_1 + K_2 \left(S_{2_1} + S_{2_2} \right).
\end{multline}
\subsubsection*{Term $S_0$}  From polynomial moment generation estimate \eqref{poly gen max t} we can bound polynomial moment of any order, as well as its derivative by means of \eqref{ODI F}. In particular,
\begin{equation*}
\mathfrak{m}_{\overline{\gamma} k}  \leq \mathfrak{B}^{\mathfrak{m}} \ \max_{t>0}\{ 1, t^{-{k}}\},  \quad  \mathfrak{m}_{\overline{\gamma} k}'  \leq B_{\overline{\gamma} k} \mathfrak{B}^{\mathfrak{m}} \ \max_{t>0}\{ 1, t^{-{k}}\}.
\end{equation*}
Denote
\begin{equation*}
\bar{c}_{k_0} := \max_{k\in \{0,\dots k_0-1\}} \left\{ \mathfrak{B}^{\mathfrak{m}},  B_{\overline{\gamma} k} \mathfrak{B}^{\mathfrak{m}}  \right\}. 
\end{equation*}
For $S_0$ taking $t\leq 1$, we have $  \mathfrak{m}_{\overline{\gamma} k}' \leq \bar{c}_{k_0} t^{-k}$ and therefore
\begin{equation*}
S_0 := \sum_{k=0}^{k_0-1}  \frac{(\alpha t)^{k}}{k!}  \mathfrak{m}_{\overline{\gamma} k}' \leq \bar{c}_{k_0} \sum_{k=0}^{k_0-1}  \frac{\alpha^{k}}{k!} \leq \bar{c}_{k_0} e^{\alpha} \leq 2 \bar{c}_{k_0},
\end{equation*}
for $\alpha$ such that
\begin{equation}\label{gen alpha 1}
e^\alpha \leq 2.
\end{equation}
\subsubsection*{Term $S_1$}  Using boundedness of $\mathfrak{m}_{\overline{\gamma} k+ \overline{\gamma}  }$, we can write
\begin{equation*}
S_1:= \sum_{k=k_0}^n \frac{(\alpha t)^{k}}{k!} \mathfrak{m}_{\overline{\gamma} k+ \overline{\gamma}  } = \mathcal{E}^n_{\overline{\gamma}; \overline{\gamma}} -  \sum_{k=0}^{k_0-1} \frac{(\alpha t)^{k}}{k!} \mathfrak{m}_{\overline{\gamma} k+ \overline{\gamma}  } 
\geq  \mathcal{E}^n_{\overline{\gamma}; \overline{\gamma}} - 2 \bar{c}_{k_0} \frac{1}{t},
\end{equation*}
for $\alpha$ chosen as in \eqref{gen alpha 1}.
\subsubsection*{Term $S_2$} Terms $S_{2_1}$ and $S_{2_2}$ are treated in the same fashion. We will detail calculation for $S_{2_1}$.  We first reorganize the terms in sum and  get
\begin{multline*}
S_{2_1}:=\sum_{k=k_0}^n \frac{(\alpha t)^{k}}{k!}  \left( \max_{1\leq i, j \leq I}  \boldsymbol{\mathcal{C}}^{ij}_{\frac{\overline{\gamma} k}{2}} \right) \sum_{\ell=1}^{\ell_{k}} \left( \begin{matrix}
k\\ \ell 
\end{matrix} \right)  \mathfrak{m}_{\overline{\gamma} \ell+\overline{\gamma}}  \mathfrak{m}_{\overline{\gamma} k- \overline{\gamma} \ell} 
\\
= \sum_{k=k_0}^n  \left( \max_{1\leq i, j \leq I} \boldsymbol{\mathcal{C}}^{ij}_{\frac{\overline{\gamma} k}{2}} \right) \sum_{\ell=1}^{\ell_{k}} \frac{(\alpha t)^{\ell}  \mathfrak{m}_{\overline{\gamma} \ell+\overline{\gamma}}}{\ell!}   \frac{(\alpha t)^{k-\ell}\mathfrak{m}_{\overline{\gamma} k- \overline{\gamma} \ell} }{(k-\ell)!} \leq \left( \max_{1\leq i, j \leq I} \boldsymbol{\mathcal{C}}^{ij}_{\frac{\overline{\gamma} k_0}{2}}  \right) \mathcal{E}^n_{\overline{\gamma}; \overline{\gamma}} \, \mathcal{E}^n_{\overline{\gamma}}.
\end{multline*}
since constant $\boldsymbol{\mathcal{C}}^{ij}_{\frac{\overline{\gamma} k}{2}} $ decays with respect to $k$,  for any $i,j$ and large enough $k_0\ge 2\frac{k_*}{\overline{\gamma}}$, with $k_*$ from \eqref{kstar} to ensure \eqref{povzner constant prop},    and therefore $\boldsymbol{\mathcal{C}}^{ij}_{\frac{\overline{\gamma} k}{2}} \leq \boldsymbol{\mathcal{C}}^{ij}_{ k_*}$.
Gathering all estimates together, \eqref{pomocna 11} becomes
\begin{equation}\label{pomocna 25}
\frac{\mathrm{d}}{\mathrm{d} t} \mathcal{E}^n_{\overline{\gamma}} \leq \alpha \,  \mathcal{E}^n_{\overline{\gamma}; \overline{\gamma}} + 2 \bar{c}_{k_0} - K_1 \left(  \mathcal{E}^n_{\overline{\gamma}; \overline{\gamma}} - 2 \bar{c}_{k_0} \frac{1}{t} \right) +   K_2  \left( \max_{1\leq i, j \leq I}  \boldsymbol{\mathcal{C}}^{ij}_{k_*}  \right) \mathcal{E}^n_{\overline{\gamma}; \overline{\gamma}} \, \mathcal{E}^n_{\overline{\gamma}},
\end{equation}
for $\alpha$ satisfying \eqref{gen alpha 1}.

\subsection*{Bound on $ \mathcal{E}^n_{\overline{\gamma}}$} Consider $t\in[0,\bar{T}_n]$. On this interval, $\mathcal{E}^n_{\overline{\gamma}}(\alpha t,t)\leq 4 \bar{M}_0$, as well as since $\bar{T}_n\leq 1$ yields $t^{-1}\geq 1$, which implies for \eqref{pomocna 25} the following estimate
\begin{equation*}
\frac{\mathrm{d}}{\mathrm{d} t} \mathcal{E}^n_{\overline{\gamma}} \leq  - \mathcal{E}^n_{\overline{\gamma}; \overline{\gamma}} \left( - \alpha + K_1 - K_2  \left( \max_{1\leq i, j \leq I} { \boldsymbol{\mathcal{C}}^{ij}_{k_*} } \right) 4 \bar{M}_0\right)+    \frac{ 2 \bar{c}_{k_*} }{\left(1+  K_1 \right){t}}.  
\end{equation*}
Since $\boldsymbol{\mathcal{C}}^{ij}_{\frac{\overline{\gamma} k_0}{2}}  $ converges to zero as  $k_0\ge 2\frac{k_*}{\overline{\gamma}}$, uniformly $i, j$, so choosing such large  $k_0$ and small enough $\alpha$ such that 
\begin{equation*}
 - \alpha + K_1 - K_2   \, \left( \max_{1\leq i, j \leq I}{ \boldsymbol{\mathcal{C}}^{ij}_{k_*}}  \right)  4 \bar{M}_0 > \frac{K_1}{2}.
\end{equation*}
with $K_1=K_1(k_*)$, yields 
\begin{equation*}
\frac{\mathrm{d}}{\mathrm{d} t} \mathcal{E}^n_{\overline{\gamma}} \leq  - \frac{K_1}{2} \mathcal{E}^n_{\overline{\gamma}; \overline{\gamma}} +   \frac{ K_3}{t},
\end{equation*}
for   $K_3(k_*) :=2 \bar{c}_{k_*} \left(1+  K_1(k_*) \right)$. Finally, shifted moment can be bounded as follows
\begin{equation*}
\mathcal{E}^n_{\overline{\gamma}; \overline{\gamma}}(\alpha t,t) = \sum_{k=1}^{n+1} \frac{ (\alpha t)^k \mathfrak{m}_{\overline{\gamma} k}(t) }{k!} \frac{k}{\alpha t} \geq  \frac{1}{\alpha t}  \sum_{k=2}^{n} \frac{ (\alpha t)^k \mathfrak{m}_{\overline{\gamma} k}(t) }{k!} \geq  \frac{\mathcal{E}^n_{\overline{\gamma}}(\alpha t,t) - \bar{M}_0}{\alpha t},
\end{equation*}
that yields
\begin{equation*}
\frac{\mathrm{d}}{\mathrm{d} t} \mathcal{E}^n_{\overline{\gamma}} \leq  - \frac{K_1}{2\alpha t} \left( \mathcal{E}^n_{\overline{\gamma}} - \bar{M}_0 -  \frac{2\alpha}{K_1} { K_3} \right).
\end{equation*}
Now we choose $\alpha$ small enough so that
 $$
\bar{M}_0 +  \frac{2\alpha}{K_1} { K_3} < 2 \bar{M}_0, \quad \text{or, equivalently} \quad  \alpha=\alpha({k_*})<\frac{ K_1(k_*) \bar{M}_0}{2  K_3(k_*)},
$$
which implies 
\begin{equation*}
\frac{\mathrm{d}}{\mathrm{d} t} \mathcal{E}^n_{\overline{\gamma}}(\alpha t,t) \leq  - \frac{K_1}{2\alpha t} \left( \mathcal{E}^n_{\overline{\gamma}}(\alpha t,t) - 2\bar{M}_0  \right).
\end{equation*}
As in \cite{GambaTask18}, integrating this differential inequality with an integrating factor $t^{\frac{K_1}{2\alpha}}$, yields 
\begin{equation}\label{pomocna 26}
\mathcal{E}^n_{\overline{\gamma}}(\alpha t,t)  \leq \max\left\{  \mathcal{E}^n_{\overline{\gamma}}(0,0), 2\bar{M}_0  \right\} \leq 2 \bar{M}_0, \qquad \forall t\in [0,\bar{T}_n],
\end{equation}
since $\mathcal{E}_{\overline{\gamma}}(0,0)=\mathfrak{m}_{0}(0)<2\bar{M}_0$.

\subsection*{Conclusion II}  From  \eqref{pomocna 26} the following bound on 
$\mathcal{E}^n_{\overline{\gamma}}(\alpha t,t)$  holds
\begin{equation*}
\mathcal{E}^n_{\overline{\gamma}}(\alpha t,t)  \leq 2 \bar{M}_0 < 4\bar{M}_0, \qquad \forall t \in [0,\bar{T}_n].
\end{equation*}
Exploring the continuity of the partial sum $\mathcal{E}^n_{\overline{\gamma}}(\alpha t,t)$ this inequality holds on a slightly larger interval, which contradicts maximality of $\bar{T}_n$, unless $\bar{T}_n=1$. Therefore, we can conclude $\bar{T}_n=1$ for all $n \in \mathbb{N}$, or in other words 
\begin{equation*}
\mathcal{E}^n_{\overline{\gamma}}(\alpha t,t) \leq   4\bar{M}_0, \qquad \forall t \in [0,1], \quad \forall n \in \mathbb{N}.
\end{equation*}
Letting $n\rightarrow \infty$, we conclude 
\begin{equation}\label{pomocna 27}
\mathcal{E}^n_{\overline{\gamma}}(\alpha t,t) \leq   4\bar{M}_0, \qquad \forall t \in [0,1].
\end{equation}
In particular, for time $t=1$, \eqref{pomocna 27} can be seen as an initial condition for propagation \eqref{initial data exp prop}, and thus the exponential moment of the order $\overline{\gamma}$ and a rate  $0<\bar{\alpha}\leq \alpha(k_*)$ stays uniformly bounded for all $t>1$,  for $k_*$ as in \eqref{kstar}.

\bigskip

\section{Acknowledgements.} 
The authors would like to thank Professor Ricardo J. Alonso for fruitful discussions on the topic.  This work has been partially supported by NSF grants DMS 1715515 and  RNMS (Ki-Net) DMS-1107444.  Milana Pavi\'c-\v Coli\'c acknowledges the support of Ministry of Education, Science and Technological Development, Republic of Serbia within the Project No. ON174016.   This work was completed while Milana Pavi\'c-\v Coli\'c was a Fulbright Scholar from the University of Novi Sad, Serbia, visiting the Institute of Computational Engineering and Sciences (ICES) at the University of Texas Austin co funded by a JTO Fellowship.  ICES support is also gratefully acknowledged.

\appendix

%\tocless
\section{Existence and Uniqueness Theory for ODE in Banach spaces}\label{Appendix exi and uni}

\begin{theorem}
	\label{Theorem general}
Let $E:=(E,\left\| \cdot \right\|)$ be a Banach space, $\mathcal{S}$ be a bounded, convex and closed subset of $E$, and $\mathcal{Q}:\mathcal{S}\rightarrow E$ be an operator satisfying the following properties:
	\begin{itemize}
		\item[(a)] H\"{o}lder continuity condition
		\begin{equation*}
		\left\| \mathcal{Q}[u] - \mathcal{Q}[v] \right\| \leq C \left\| u-v \right\|^{\beta}, \ \beta \in (0,1), \ \forall u, v \in \mathcal{S};
		\end{equation*}
		\item[(b)] Sub-tangent condition
		\begin{equation*}
		\lim\limits_{h\rightarrow 0+} \frac{\text{dist}\left(u + h \mathcal{Q}[u], \mathcal{S} \right)}{h} =0, \ \forall u \in \mathcal{S};
		\end{equation*}
		\item[(c)] One-sided Lipschitz condition
		\begin{equation*}  
		\left[ \mathcal{Q}[u] - \mathcal{Q}[v], u - v \right] \leq C \left\| u-v \right\|, \ \forall u, v \in \mathcal{S},
		\end{equation*}
		where 	$\left[\varphi,\phi\right]=\lim_{h\rightarrow 0^-} h^{-1}\left(\left\| \phi + h \varphi \right\| - \left\| \phi \right\| \right)$.
	\end{itemize}
Then the equation
\begin{equation*}
\begin{split}
\partial_t  u = \mathcal{Q}[u], \ \text{for} \ t\in(0,\infty), \ \text{with initial data}  \
u(0)= u_0 \ \text{in} \ \mathcal{S},
\end{split}
\end{equation*}
has a unique solution in $C([0,\infty),\mathcal{S})\cap C^1((0,\infty),E)$.
\end{theorem}
The proof of this Theorem on ODE flows on Banach spaces  can be found in the unpublished notes \cite{Bressan}  or in  \cite{GambaAlonso18}.
\begin{remark}\label{Lip remark}
	In Section \ref{Section Ex Uni proof}, we will concentrate on $E:=L_2^1$. Therefore, for one-sided Lipschitz condition, we will use the following inequality,   
	$$
	\left[\varphi,\phi\right]\leq  \sum_{i=1}^I \int_{ \mathbb{R}^3}  \varphi_i(v) \, \text{sign}(\phi_i (v))\left\langle v \right\rangle_i^2 \mathrm{d}v,
	$$
	for $\varphi=\left[\varphi_i\right]_{1\leq i\leq I}$ and $\phi=\left[\phi_i\right]_{1\leq i\leq I}$, as pointed out in \cite{GambaAlonso18}.
\end{remark}

\section{Upper and lower bound of the cross section}

In this section, we derive an  upper and lower estimate for the non-angular part of the cross section, $\left|v-v_*\right|^{\gamma_{ij}}$, $\gamma_{ij} \in (0,1]$, with $1\leq i, j, \leq I$. First, for the upper estimate, by triangle inequality, we have 
\begin{align}
\sqrt{\frac{m_i}{\sum_{i=1}^I m_i}} \sqrt{\frac{m_j}{\sum_{i=1}^I m_i}} \left|v-v_*\right| \leq \min\left\{ \sqrt{\frac{m_i}{\sum_{i=1}^I m_i}}, \sqrt{\frac{m_j}{\sum_{i=1}^I m_i}}  \right\}  \left|v-v_*\right| \nonumber \\ \leq \min\left\{ \sqrt{\frac{m_i}{\sum_{i=1}^I m_i}}, \sqrt{\frac{m_j}{\sum_{i=1}^I m_i}}  \right\} \left(  \left|v\right| + \left|v_*\right|  \right)  \leq  \sqrt{\frac{m_i}{\sum_{i=1}^I m_i}} \left|v\right| +  \sqrt{\frac{m_j}{\sum_{i=1}^I m_i}} \left|v_*\right| \label{pomocna 14}
\\ \leq  \sqrt{1 + \frac{m_i}{\sum_{i=1}^I m_i} \left|v\right|^2} + \sqrt{1 + \frac{m_j}{\sum_{i=1}^I m_i} \left|v_*\right|^2}. \nonumber
\end{align}
Therefore,
\begin{equation}\label{estimate on u^g}
\left|v-v_*\right|^{\gamma_{ij}} \leq  \left( \frac{\sum_{i=1}^I m_i}{\sqrt{m_i m_j}} \right)^{\gamma_{ij}} \left(\left\langle v \right\rangle_i^{\gamma_{ij}} + \left\langle v_* \right\rangle_j^{\gamma_{ij}}\right),
\end{equation}
for $\gamma_{ij}\in(0,1]$, and any $i, j \in \left\{1,\dots,I \right\}$. 

From \eqref{pomocna 14}  it also follows 
\begin{multline*}
\sqrt{\frac{m_i}{\sum_{i=1}^I m_i}} \sqrt{\frac{m_j}{\sum_{i=1}^I m_i}} \left|v-v_*\right|  \leq  \sqrt{\frac{m_i}{\sum_{i=1}^I m_i}} \left|v\right| +  \sqrt{\frac{m_j}{\sum_{i=1}^I m_i}} \left|v_*\right| 
\\
= \left( \frac{m_i}{\sum_{i=1}^I m_i} \left|v\right|^2  + \frac{m_j}{\sum_{i=1}^I m_i} \left|v_*\right|^2 + 2 \frac{\sqrt{m_i m_j}}{\sum_{i=1}^I m_i} \left|v\right|  \left|v_*\right|  \right)^{1/2}
\leq \left\langle v \right\rangle_i \left\langle v_* \right\rangle_j.
\end{multline*}
Therefore,
\begin{equation}\label{estimate on u^g product}
\left|v-v_*\right|^{\gamma_{ij}} \leq  \left( \frac{\sum_{i=1}^I m_i}{\sqrt{m_i m_j}} \right)^{\gamma_{ij}} \left\langle v \right\rangle_i^{\gamma_{ij}}  \left\langle v_* \right\rangle_j^{\gamma_{ij}},
\end{equation}
for $\gamma_{ij} \in(0,1]$ and $1\leq i, j \leq I$. 

Than, for the lower estimate we use ideas of Lemma  2.1 in \cite{GambaAlonsoMaja17}, to prove the following Lemma. Note that here  functions $F$ do not need to be solutions of the Boltzmann problem. Moreover, this lower bound may not hold for $F$ being a singular measure, since the estimate degenerates as $c$ goes to zero.
\begin{lemma}\label{lemma lower bound}
	Let $\gamma_{ij} \in [0,2]$, for any $i, j \in \left\{1,\dots, I \right\}$, and assume \\
	 $0\leq \left\{ F(t) = \left[ f_1(t) \dots f_I(t) \right]^T \right\}_{t\geq 0} \subset L_2^1$ satisfies
	\begin{multline*}
	c \leq \sum_{i=1}^I \int_{\mathbb{R}^3} m_i \, f_i (t, v)  \mathrm{d}v \leq C, \qquad c \leq \sum_{i=1}^I \int_{\mathbb{R}^3} f_i (t, v) m_i \left|v\right|^2 \mathrm{d}v \leq C,\\ \sum_{i=1}^I \int_{\mathbb{R}^3} f_i (t, v) m_i v \mathrm{d}v=0, 
	\end{multline*} 
	for some positive constants $c$ and $C$. Assume also boundedness of the moment
	\begin{equation*}
	\sum_{i=1}^I \int_{\mathbb{R}^3} f_i (t, v) m_i \left|v\right|^{2+\varepsilon}  \mathrm{d}v \leq B, \quad \varepsilon>0. 
	\end{equation*}
	Then, there exists a constant $c_{lb}$  characterized in \eqref{clb}, such that
	\begin{equation}\label{lower bound}
	\sum_{i=1}^I \int_{\mathbb{R}^3} m_i f_i (t, w) \left|v-w\right|^{\gamma_{ij}} \mathrm{d}w \geq c_{lb} \left\langle v \right\rangle_j^{\overline{\gamma}}, 
	\end{equation}
	for any $j \in \{1,\dots,I\}$,  with $\overline{\gamma}=\min_{1\leq i, j \leq I}\gamma_{ij}$.
\end{lemma}
\begin{proof}
	Case $\gamma_{ij}=0$ is trivial, so take $\gamma_{ij} \in (0,2]$, for any  $i,j,=1,\dots,I$. 
	
Let us  denote the open ball centered at the origin and of radius $r>0$ with $B(0,r) \subset \mathbb{R}^3$. We consider separately cases when  $v \in B(0,r)$ and $v \in B(0,r)^c$, with $r$   to be chosen later on depending on constants $c$, $C$, and $\gamma_{ij}$.
  
	For  $v \in B(0,r)^c$ we first consider the whole domain $\mathbb{R}^3$, and write, by the Young inequality, for any $v \in \mathbb{R}^3$ and $\gamma_{ij} \in (0,2]$ 
	\begin{equation*}
	\sum_{i=1}^I  m_i	 \int_{\mathbb{R}^3} f_i(t,w) \left|v-w\right|^{\gamma_{ij}} \mathrm{d} w  \geq \sum_{i=1}^I  m_i	 \int_{\mathbb{R}^3} f_i(t,w) \left(\tilde{c}  \left|v\right|^{\gamma_{ij}} - \left|w\right|^{\gamma_{ij}} \right) \mathrm{d} w,
		\end{equation*}
	 where $\tilde{c} =\min_{1\leq i, j \leq I} \left(\min\{1, 2^{1-\gamma_{ij}}\} \right)$. 	
	Since
	\begin{multline*}
 \sum_{i=1}^I  m_i	 \int_{\mathbb{R}^3} f_i(t,w) \left|w\right|^{\gamma_{ij}} \mathrm{d} w \\
 \leq  \sum_{i=1}^I  m_i	 \int_{B(0,1)} f_i(t,w)  \mathrm{d} w +  \sum_{i=1}^I  m_i	 \int_{B(0,1)^c} f_i(t,w) \left|w\right|^2 \mathrm{d} w \leq 2 C,
	\end{multline*}
	we obtain that for any $v\in \mathbb{R}^3$	it holds
		\begin{equation}\label{pomocna 1}
	\sum_{i=1}^I  m_i	 \int_{\mathbb{R}^3} f_i(t,w) \left|v-w\right|^{\gamma_{ij}} \mathrm{d} w  \\    \geq \tilde{c}   \sum_{i=1}^I \left|v\right|^{\gamma_{ij}}   m_i	 \int_{\mathbb{R}^3} f_i(t,w)  \mathrm{d} w - 2C.
	\end{equation}
 Now define the following two parameters, both smaller than one,
\begin{equation*}
\overline{\bf m} := \sqrt{ \frac {\min_{1\le j\le I} {m_j} }{\sum_{i=1}^I m_i}}, \quad \text{and} \quad \wideBar{\bf m} := \sqrt{ \frac {\max_{1\le j\le I} {m_j} }{\sum_{i=1}^I m_i}}.
\end{equation*} 	
In addition, we define the parameter $r_*$ by 
\begin{equation}\label{r star}
\overline{\bf m}  \,r_*\ :=\    \left(  \frac{{4}C}{\tilde c \, c}  \right)^{\frac 1{\overline\gamma}}   \ge\ 1,
\end{equation}
since $C\geq c$ by assumption and   $\tilde c \leq1$.\\
 Hence, for any $i,j=1,\dots,I$ and $v\in \mathbb{R}^3   \cap   B(0,r)^c$ we have the following lower bound
\begin{equation*}
\left| v  \right|^{\gamma_{ij}}  =  \left| v  \right|^{\gamma_{ij}} \left(  \mathbbm{1}_{\left|v\right|<1}(v) +  \mathbbm{1}_{\left|v\right|\geq1}(v) \right)  \geq   \left| v  \right|^{\overline{\gamma}},
\end{equation*}
for  any $r\geq r_*\geq 1$, where
\begin{equation*}
{\overline{\gamma}}= \min_{ \color{black} 1\leq i,j \leq I} \gamma_{ij}. 
\end{equation*}
Therefore,  using the choice of $r_*$ with the inequality \eqref{r star}, \eqref{pomocna 1} becomes
\begin{multline*}
\sum_{i=1}^I  m_i	 \int_{\mathbb{R}^3} f_i(t,w) \left|v-w\right|^{\gamma_{ij}} \mathrm{d} w   \ge  \tilde{c}  \,  c \,  \left|v\right|^{\overline{\gamma}}  - 2C \\
\ge \frac{\tilde{c} \,c }{ 2}\,  \left(\sqrt{\frac{m_j}{\sum_{i=1}^I m_i}} \, \left|v\right|\right)^{\overline{\gamma}}  +  \frac{\tilde{c} \,c }{ 2}  { \left( \overline{\bf m}\, r_*\right)^{\overline\gamma} } -  2C,
\end{multline*}
for every  $j\in\{1,\dots,I\}$. Therefore, for $v \in B(0,r_*)^c$ we have
\begin{equation}\label{pomocna 13}
\sum_{i=1}^I  m_i	 \int_{\mathbb{R}^3} f_i(t,w) \left|v-w\right|^{\gamma_{ij}} \mathrm{d} w   
\geq \frac{\tilde{c}  \, c}2   \left( \sqrt{\frac{m_j}{\sum_{i=1}^I m_i}}   \left|v\right|\right)^{\overline{\gamma}},
\end{equation}
for any  $j\in\{1,\dots,I\}$. 

On the other hand, let us study the case $v \in B(0,r^*)$. First  note that for any $R>0$,
	\begin{multline}\label{pomocna 3}
	\sum_{i=1}^I  m_i	 \int_{\left|v-w\right|\leq R} f_i(t,w) \left|v-w\right|^2 \mathrm{d} w \\=\sum_{i=1}^I  m_i	 \int_{\mathbb{R}^3} f_i(t,w) \left|v-w\right|^2 \mathrm{d} w - \sum_{i=1}^I  m_i \int_{\left|v-w\right|\geq R} f_i(t,w) \left|v-w\right|^2 \mathrm{d} w \\ \geq c \left|v\right|^2+ c  - \sum_{i=1}^I  m_i \int_{\left|v-w\right|\geq R} f_i(t,w) \left|v-w\right|^2 \mathrm{d} w \\
	\geq c (1+ \left|v\right|^2) - \frac{1}{R^{\varepsilon}} \sum_{i=1}^I  m_i \int_{\left|v-w\right|\geq R} f_i(t,w) \left|v-w\right|^{2+\varepsilon} \mathrm{d} w.
	\end{multline}	
	Next, we have 
	\begin{multline*}
	\sum_{i=1}^I  m_i \int_{\left|v-w\right|\geq R} f_i(t,w) \left|v-w\right|^{2+\varepsilon} \mathrm{d} w \leq 2^{1+\varepsilon} \max\{C, B\} \left(1+ \left|v\right|^{2+\varepsilon}\right) \\ \leq 2^{1+\varepsilon} \max\{C,B\} \left(1+ \left|v\right|^{2}\right)^{\frac{2+\varepsilon}{2}} \leq 2^{1+\varepsilon} \max\{C, B\} \left(1+ r_*^{2}\right)^{\frac{2+\varepsilon}{2}}.
	\end{multline*}	
	Choosing $R:=R(r_*, c, C, B)>0$	sufficiently large such that 
	\begin{equation}\label{R}
	 \frac{1}{R^\varepsilon}2^{1+\varepsilon} \max\{C, B\} \left(1+ r_*^{2}\right)^{\frac{2+\varepsilon}{2}} \leq \frac{c}{2}, \quad \text{or} \quad
	R \geq \left( 2^{2+\varepsilon} \left(\frac{\max\{C,B\}}{c} \right)\left(1+ r_*^{2}\right)^{\frac{2+\varepsilon}{2}}\right)^{\frac{1}{\varepsilon}},
	\end{equation}
	from \eqref{pomocna 3} we have 
	\begin{equation*}
	\sum_{i=1}^I  m_i	 \int_{\left|v-w\right|\leq R} f_i(t,w) \left|v-w\right|^2 \mathrm{d} w 
	\geq \frac{c}{2} \qquad \forall v \in B(0,r_*).
	\end{equation*}
	Moreover, for this choice of $R$, for any $\gamma_{ij}\in(0,2]$ we have
	\begin{multline*}
	\sum_{i=1}^I  m_i	 \int_{\mathbb{R}^3} f_i(t,w) \left|v-w\right|^{\gamma_{ij}} \mathrm{d} w  \geq \sum_{i=1}^I  m_i	 \int_{\left|v-w\right|\leq R} f_i(t,w) \left|v-w\right|^{\gamma_{ij}} \mathrm{d} w \nonumber \\
	\geq  \sum_{i=1}^I  R^{\gamma_{ij} -2  }  \, m_i	 \int_{\left|v-w\right|\leq R} f_i(t,w) \left|v-w\right|^2 \mathrm{d} w.
	\end{multline*}
		Since $R\geq 1$, we can bound $R^{\gamma_{ij} -2  } \geq R^{ { \overline{\gamma}} -2}$, which yields the estimate
	\begin{equation}\label{pomocna 17}
	\sum_{i=1}^I  m_i	 \int_{\mathbb{R}^3} f_i(t,w) \left|v-w\right|^{\gamma_{ij}} \mathrm{d} w    \geq \frac{c}{2 R^{2-{ \overline{\gamma}}}}, \qquad \forall v \in B(0,r_*).
	\end{equation}

	Finally, summarizing \eqref{pomocna 13} and \eqref{pomocna 17},
	\begin{multline*}
	\sum_{i=1}^I  m_i	 \int_{\mathbb{R}^3} f_i(t,w) \left|v-w\right|^{\gamma_{ij}} \mathrm{d} w  \geq \frac{c}{2 R^{2-{ \overline{\gamma}}}} \mathbbm{1}_{B(0,r_*)}(v) \\ + \frac{\tilde{c}  \, c}{2}  \left( \sqrt{\frac{m_j}{\sum_{i=1}^I m_i}}   \left|v\right|\right)^{\overline{\gamma}}\mathbbm{1}_{B(0,r_*)^c}(v)\\
\geq 	\frac{\tilde{c}  \, c}{2 R^{2-{ \overline{\gamma}}}} \left( \mathbbm{1}_{B(0,r_*)}(v)+  \left( \sqrt{\frac{m_j}{\sum_{i=1}^I m_i}}   \left|v\right|\right)^{\overline{\gamma}}\mathbbm{1}_{B(0,r_*)^c}(v) \right).
	\end{multline*}
	Then there exists a constant $c_{lb}$ such that
	\begin{equation*}
	\frac{\tilde{c}  \, c}{2 R^{2- { \overline{\gamma}}}} \left( \mathbbm{1}_{B(0,r_*)}(v)+  \left( \sqrt{\frac{m_j}{\sum_{i=1}^I m_i}}   \left|v\right|\right)^{\overline{\gamma}}\mathbbm{1}_{B(0,r_*)^c}(v) \right) \geq c_{lb} \left\langle v \right\rangle_j^{\overline{\gamma}}.
	\end{equation*}
	for any $j\in \left\{1,\dots,I\right\}$. In fact, one may even construct $c_{lb}$ in order to ensure the last inequality. For example, $c_{lb}$ can take the following value
\begin{multline}\label{clb}
	c_{lb}
= \frac{c}{2} \tilde{c}   \left( 2^{2+\varepsilon} \left(\frac{\max\{C,B\}}{c} \right)\left(1+ { \frac{1}{\overline{\bf m}^2}} \left( \frac{ {4} \, C}{\tilde{c}  \, c} \right)^\frac{2}{\overline{\gamma}}\right)^{\frac{2+\varepsilon}{2}}\right)^{\frac{-2+{ \overline{\gamma}}}{\varepsilon}}
\\  \times \left( 1 +  { \left( \frac{{\wideBar{\bf m}} }{\overline{\bf m}}  \right)^2} \left( \frac{ {4} \, C}{\tilde{c}  \, c} \right)^\frac{2}{{\overline{\gamma} }} \right)^{-\overline{\gamma}/2},
	\end{multline}
by taking into account \eqref{r star} and \eqref{R}. 
\end{proof}

\section{Some technical results }

\begin{lemma}[Polynomial inequality I, Lemma 2 from \cite{GambaBobPanf04}] \label{binomial}
Assume $p>1$, and let $n_p=\lfloor\frac{p+1}{2}\rfloor$. Then for all $x, y>0$, the following inequality holds
\begin{equation*}
\left(x+y\right)^p - x^p - y^p \leq \sum_{n=1}^{n_p} \left( \begin{matrix}
p \\ n 
\end{matrix} \right) \left(x^n y^{p-n} + x^{p-n}  y^n \right).
\end{equation*}	
\end{lemma}

\begin{lemma}[Polynomial inequality II] \label{moment products}
Let $b+1\leq a\leq \frac{p+1}{2}$. Then for any  $x, y\geq0$, 
	\begin{equation*}
	x^a y^{p-a} + x^{p-a} y^a \leq x^b y^{p-b} + x^{p-b} y^b.
	\end{equation*}	
\end{lemma}
\begin{proof} This Lemma is modified version of Lemma A.1 from \cite{GambaTask18}. Indeed, the proof is the same, one just needs to observe that $a-b\geq 0$ and $p-a-b\geq0$, and therefore
$$
\left(y^{a-b}-x^{a-b}\right)x^by^b\left(y^{p-a-b}-x^{p-a-b}\right)\geq 0,  
$$	
for any $x,y\geq0$.
\end{proof}

\begin{lemma}[Interpolation inequality]
	Let $k=\alpha k_1 + (1-\alpha) k_2 $, $\alpha\in(0,1)$, $0<k_1\leq k \leq k_2$. Then for any $g \in L_{k,i}^1$
	\begin{equation}\label{interpolation inequality}
	\left\| g \right\|_{L_{k,i}^1} \leq \left\| g \right\|_{L_{k_1,i}^1}^{\alpha}  \left\| g \right\|_{L_{k_2,i}^1}^{1-\alpha}.
	\end{equation}
\end{lemma}
We can extend this interpolation inequality for vector functions $\mathbb{G}=\left[g_i\right]_{1\leq i \leq I}$. Namely, under the same assumptions,
\begin{equation}\label{interpolation inequality vector}
\left\| \mathbb{G}  \right\|_{L_k^1} \leq I \left\| \mathbb{G} \right\|_{L_{k_1}^1}^\alpha \left\| \mathbb{G} \right\|_{L_{k_2}^1}^{1-\alpha}.
\end{equation}

\begin{lemma}[Jensen's inequality] Let $f(x)$ be positive and integrable in $\mathbb{R}^d$ and $G$ a convex function. Then
	\begin{equation*}
	G\left( \frac{1}{\int f(x) \mathrm{d}x} \int f(x) g(x) \mathrm{d}x \right) \leq \frac{1}{\int f(x) \mathrm{d}x} \int f(x) G(g(x)) \mathrm{d}x,
	\end{equation*} 
	for any positive function $g$.
\end{lemma}
We apply this lemma specifying $g(x)=\left\langle x \right\rangle_i^k$ and $G(x)=x^{1+\frac{\lambda}{k}}$, $\lambda \in (0,1]$ and $k\geq 1$. This implies
\begin{equation*}
\int_{ \mathbb{R}^3} f_i(v) \left\langle v \right\rangle_i^{k+\lambda} \mathrm{d}v \geq \left( \int_{ \mathbb{R}^3} f_i(v) \mathrm{d}v \right)^{-\frac{\lambda}{k}}  \left( \int_{ \mathbb{R}^3} f_i(v) \left\langle v \right\rangle_i^k \mathrm{d}v \right)^{1+\frac{\lambda}{k}} .
\end{equation*}

If additionally we have an upper bound on  zero order scalar polynomial moment, that is, if it holds
\begin{equation*}
 \int_{ \mathbb{R}^3} f_i(v) \mathrm{d}v =  \mathfrak{m}_{0,i}[\mathbb{F}] \leq \mathfrak{m}_0[\mathbb{F}] \leq C_{\mathfrak{m}_{0}},
\end{equation*}
then 
\begin{equation*}
\int_{ \mathbb{R}^3} f_i(v) \left\langle v \right\rangle_i^{k+\lambda} \mathrm{d}v \geq  C_{\mathfrak{m}_{0}}^{-\frac{\lambda}{k}}  \left( \int_{ \mathbb{R}^3} f_i(v) \left\langle v \right\rangle_i^k \mathrm{d}v \right)^{1+\frac{\lambda}{k}} .
\end{equation*}

Summing over $i=1,\dots,I$ after some manipulation we get a control from below for the moment $ \mathfrak{m}_{k+\lambda}[\mathbb{F}]$. In	deed,
	\begin{equation}\label{jensen}
 \mathfrak{m}_{k+\lambda}[\mathbb{F}] \geq \left( I   C_{\mathfrak{m}_{0}} \right)^{-\frac{\lambda}{k}} \mathfrak{m}_k[\mathbb{F}]^{1+\frac{\lambda}{k}}.
\end{equation}

%\section*{Acknowledgments}

\medskip
% The data information below will be filled by AIMS editorial staff
%Received ; revised .
\medskip


\begin{thebibliography}{99}

\bibitem{AlonsoLods18}{
	\newblock	R. Alonso, V. Bagland, Y. Cheng, and B. Lods, 
	\newblock { One-dimensional dissipative Boltzmann equation: measure
		solutions, cooling rate, and self-similar profile}, 
	\newblock \emph{SIAM J. Math. Anal.}, 50 (1): 1278--1321, 2018.
	
} 



\bibitem{Gamba13}{
\newblock	R. Alonso, J. A. Ca\~nizo, I. M. Gamba, and C. Mouhot, 
\newblock { A new approach to the creation and propagation of exponential
		moments in the Boltzmann equation}, 
\newblock \emph{Comm. Partial Differential Equations}, 38 (1): 155--169, 2013.

}

\bibitem{GambaAlonso18}
\newblock R. J. Alonso and I. M. Gamba,
\newblock  Revisiting the Cauchy problem for the Boltzmann equation for hard
potentials with integrable cross section: from generation of moments to
propagation of $L^{\infty}$ bounds.
\newblock preprint (2018).


\bibitem{GambaAlonsoMaja17}
\newblock R. J. Alonso, I. M. Gamba and M. Taskovi\'{c},
\newblock Exponentially-tailed regularity and time asymptotic for the homogeneous Boltzmann equation,
\newblock preprint, ArXiv {1711.06596}.


\bibitem{GambaAlonsoTran17}
\newblock  R. J. Alonso, I. M. Gamba and M. B.Tran,
\newblock The Cauchy problem for the quantum Boltzmann equation for bosons at very low temperature,
\newblock preprint, ArXiv 1609.07467.v2.



\bibitem{Bob97} %(MR1478067) [10.1007/BF02732431]
\newblock A. V. Bobylev,
\newblock {Moment inequalities for the Boltzmann equation and applications to spatially homogeneous problems},
\newblock \emph{J. Statist. Phys.}, {88}: 1183--1214, 1997.


\bibitem{BobGamba17}  %(MR3591124) [10.3934/krm.2017023]
\newblock A. V. Bobylev and I. M. Gamba,
\newblock Upper Maxwellian bounds for the Boltzmann equation with pseudo-Maxwell molecules.
\newblock \emph{Kinet. Relat. Models}, \textbf{10} (2017),  573--585.

\bibitem{GambaBobPanf04} %(MR2096050) [10.1023/B:JOSS.0000041751.11664.ea]
\newblock A. V. Bobylev, I. M. Gamba, and V. A. Panferov,
\newblock Moment inequalities and high-energy tails for Boltzmann equations with inelastic interactions,
\newblock \emph{J. Statist. Phys.}, \textbf{116} (2004), 1651--1682.

\bibitem{BGPS} %(MR1478067) [10.1007/BF02732431]
\newblock L. Boudin, B. Grec, M. Pavi\'c and F. Salvarani,
\newblock {Diffusion asymptotics of a kinetic model for gaseous mixtures},
\newblock \emph{Kin. and Rel. Models}, {6(1)}: 137--157, 2013.

\bibitem{Bressan} 
\newblock A. Bressan,
\newblock Notes on the Boltzmann equation,
Lecture notes for a summer course, S.I.S.S.A., 2005. (http://www.math.psu.edu/bressan/)


\bibitem{BriantDaus16} 
\newblock M. Briant and E. Daus, 
\newblock {The Boltzmann equation for a multi-species mixture close to
		global equilibrium}, 
\newblock \emph{Arch. Ration. Mech. Anal.}, 222(3): 1367--1443, 2016.
	
\bibitem{Des93} %(MR1233644) [10.1007/BF00375586]
	\newblock L. Desvillettes,
	\newblock Some applications of the method of moments for the homogeneous Boltzmann and Kac equations,
	\newblock \emph{Arch. Rational Mech. Anal.}, \textbf{123} (1993), 387--404.
	
\bibitem{DesMonSalv} %(MR1233644) [10.1007/BF00375586]
\newblock L. Desvillettes, R. Monaco and F. Salvarani
\newblock A kinetic model allowing to obtain the energy law of polytropic gases in the presence of chemical reactions,
\newblock \emph{Eur. J. Mech. B Fluids}, \textbf{24} (2005), 219--236.


\bibitem{GambaPanfVil09}  %(MR2533928) [10.1007/s00205-009-0250-9]
\newblock I. M. Gamba, V. Panferov, and C. Villani,
\newblock Upper Maxwellian bounds for the spatially homogeneous Boltzmann equation,
\newblock \emph{ Arch. Ration. Mech. Anal.}, \textbf{194} (2009), 253--282.

\bibitem{GambaSmithTran18}  I.M.Gamba, L.Smith and M.B.Tran 
\newblock On the wave turbulence theory for stratified flows in the ocean,
\newblock arXiv-1709.08266v2, Submitted for Publication (2018). 

\bibitem{LuMouhot12} %(MR2871802) [/10.1016/j.jde.2011.10.021]
\newblock X. Lu and C. Mouhot,
\newblock On measure solutions of the Boltzmann equation, part I: moment production and stability estimates,
\newblock \emph{J. Differential Equations}, \textbf{252} (2012), 3305--3363.

\bibitem{Martin} R. H. Martin,
\newblock {Nonlinear operators and differential equations in Banach spaces.}
\newblock \emph{Pure and Applied Mathematics. Wiley-Interscience}, 1976.

\bibitem{Mouhot06} %(MR2197542) [10.1007/s00220-005-1455-x]
\newblock C. Mouhot,
\newblock Rate of convergence to equilibrium for the spatially homogeneous Boltzmann equation with hard potentials,
\newblock \emph{ Comm. Math. Phys.}, \textbf{261} (2006), 629--672.


\bibitem{PT}
\newblock M. Pavi\'c-\v Coli\'c, and M. Taskovi\'c,
\newblock Propagation of stretched exponential moments for the Kac equation and Boltzmann equation with Maxwell molecules,
\newblock  \emph{Kinet. Relat. Models}, \textbf{11}(3) (2018),  597--613.

\bibitem{Pov}
\newblock A. J. Povzner,
\newblock The Boltzmann equation in the kinetic theory of gases,
\newblock  \emph{Amer. Math. Soc. Transl.}, \textbf{47}(2) (1965),  193--214.



\bibitem{Sir62}  
\newblock L. Sirovich,
\newblock Kinetic modeling of gas mixtures,
\newblock \emph{ Phys. Fluids,}, \textbf{5} (1962), 908--918.

\bibitem{GambaTask18}{
	\newblock 	M. Taskovi\'c, R. J. Alonso,  I. M. Gamba, and N. Pavlovi\'c,
	\newblock {On Mittag-Leffler moments for the Boltzmann equation for hard potentials without cutoff},
	\newblock \emph{SIAM J. Math. Anal.}, {50(1)}: 834--869, 2018.
}

\bibitem{wennberg97}
B.~Wennberg, \emph{Entropy dissipation and moment production for the
  {B}oltzmann equation}, Jour. Statist. Phys. \textbf{86} (1997), no.~5-6,
  1053--1066.
\end{thebibliography}
\end{document}